\pdfoutput=1
\PassOptionsToPackage{warn}{textcomp}
\documentclass[sigplan,10pt]{acmart}

\copyrightyear{2020} 
\acmYear{2020} 
\setcopyright{acmcopyright}
\acmConference[LICS '20]{Proceedings of the 35th Annual ACM/IEEE Symposium on Logic in Computer Science (LICS)}{July 8--11, 2020}{Saarbr\"ucken, Germany}
\acmBooktitle{Proceedings of the 35th Annual ACM/IEEE Symposium on Logic in Computer Science (LICS '20), July 8--11, 2020, Saarbr\"ucken, Germany}
\acmPrice{15.00}
\acmDOI{10.1145/3373718.3394783}
\acmISBN{978-1-4503-7104-9/20/07}

\usepackage[utf8]{inputenc} 
\usepackage{amsmath, amsthm, amssymb}
\usepackage{thmtools}
\usepackage[english]{babel}
\usepackage{pdfpages}
\usepackage{paralist}
\usepackage{graphicx}
\usepackage{subcaption}
\usepackage{caption}
\usepackage{longtable}
\usepackage{booktabs}
\usepackage{tabu}
\usepackage[referable]{threeparttablex}
\usepackage{varwidth}
\usepackage{multirow}
\usepackage{wrapfig}
\usepackage[inline]{enumitem}
\usepackage{microtype}

\usepackage{pifont}

\usepackage[ruled, linesnumbered]{algorithm2e}
\usepackage{parskip}

\usepackage{hyperref}
\usepackage{thm-restate}
\usepackage[capitalise]{cleveref}

\crefname{figure}{Figure}{Figure}
\crefformat{footnote}{#2\footnotemark[#1]#3}
\usepackage{tikz}
\usepackage{comment}

\usepackage{todonotes}
\usepackage{tikz}
\usetikzlibrary{arrows,automata,shapes,decorations,decorations.markings,calc, matrix,decorations.pathmorphing, patterns,backgrounds}

\usepackage{bm}
\usepackage{pgfplots}
\usepackage{enumerate,paralist}
\usepackage{pbox}

\usepackage{appendix}
\usepackage{calc}
\usepackage{fp}
\usepackage{multirow}
\usepackage{pbox}

\usepackage[mathscr]{euscript}
\usepackage{mathtools}

\theoremstyle{plain}
\newtheorem{remark}{Remark}

\DeclareMathAlphabet{\mathpzc}{OT1}{pzc}{m}{it}

\DeclareRobustCommand{\amgis}{\text{\reflectbox{$\sigma$}}}

\newcommand{\Paragraph}[1]{\smallskip\noindent{\bf #1}}
\newcommand{\SubParagraph}[1]{\smallskip\noindent{\em #1}}

\mathchardef\mhyphen="2D 

\newcommand{\ov}{\overline}

\newcommand{\Globals}{\mathcal{G}}
\newcommand{\Locks}{\mathcal{L}}
\newcommand{\Read}{\mathsf{r}}
\newcommand{\Write}{\mathsf{w}}
\newcommand{\Acquire}{\mathsf{acq}}
\newcommand{\Release}{\mathsf{rel}}
\newcommand{\lk}{\ell}

\newcommand{\Sequence}[1]{#1}
\newcommand{\Confl}[2]{#1 \Join #2}

\newcommand{\Trace}{t}
\newcommand{\Tree}{T}
\newcommand{\TreeNodes}{\mathcal{I}}
\newcommand{\TreeEdges}{\mathcal{R}}

\newcommand{\SysAcquires}{\mathcal{L}^A}
\newcommand{\SysReleases}{\mathcal{L}^R}
\newcommand{\SysReads}{\mathcal{R}}

\newcommand{\SysWrites}{\mathcal{W}}
\newcommand{\ReadsReleases}[1]{\SysReads\mathcal{L}({#1})}
\newcommand{\WritesAcquires}[1]{\SysWrites\mathcal{L}({#1})}
\newcommand{\WritesReads}[1]{\SysWrites\SysReads({#1})}
\newcommand{\Prec}{\ll}
\newcommand{\Vars}{\mathcal{V}}
\newcommand{\OPoset}{\mathcal{P}}
\newcommand{\OPosetQ}{\mathcal{Q}}
\newcommand{\OPosetS}{\mathcal{S}}
\newcommand{\OPosetK}{\mathcal{K}}

\newcommand{\Reversals}{\operatorname{Rv}}
\newcommand{\Distance}{\delta}

\newcommand{\Cycle}{\mathscr{C}}
\newcommand{\Event}{e}
\newcommand{\Frontier}[2]{\mathsf{Frontier_{#1}(#2)}}

\newcommand{\Past}{\mathsf{Cone}}
\newcommand{\LPast}{\mathsf{LCone}}
\newcommand{\CandidateSet}{\mathsf{CIS}}

\newcommand{\Concat}{\circ}

\newcommand{\Events}[1]{\SysEvents(#1)}
\newcommand{\Reads}[1]{\SysReads(#1)}
\newcommand{\Writes}[1]{\SysWrites(#1)}

\newcommand{\Acquires}[1]{\SysAcquires(#1)}
\newcommand{\Releases}[1]{\SysReleases(#1)}
\newcommand{\Project}{|}

\newcommand{\Domain}{\mathsf{dom}}

\newcommand{\Image}{\mathsf{img}}
\newcommand{\SeqTrace}{\tau}
\newcommand{\SeqSubTrace}{\sigma}
\newcommand{\SeqSubTraceRefl}{\amgis}

\newcommand{\True}{\mathsf{True}}
\newcommand{\False}{\mathsf{False}}

\newcommand{\System}{\mathcal{P}}
\newcommand{\Process}{p}
\newcommand{\Proc}[1]{\mathsf{p}(#1)}

\newcommand{\INDSET}{\operatorname{INDEPENDENT-SET}}

\newcommand{\Match}[2]{\mathsf{match}_{#1}(#2)}
\newcommand{\CS}[2]{\operatorname{CS}_{#1}({#2})}
\newcommand{\NestingDepth}{\gamma}
\newcommand{\LockFactor}{\zeta}

\newcommand{\OrthVec}{\mathsf{OV}}
\newcommand{\Implies}{\Rightarrow}
\newcommand{\Unordered}[3]{#1\parallel_{#2} #3}
\newcommand{\Ordered}[3]{#1 \not \parallel_{#2} #3}
\newcommand{\OpenAcquires}{\operatorname{OpenAcqs}}
\newcommand{\Refines}{\sqsubseteq}
\newcommand{\ReadPairs}[1]{\mathsf{Pairs}(#1)}
\newcommand{\ReadTriplets}[1]{\mathsf{Triplets}(#1)}
\newcommand{\StrictRefines}{\sqsubset}

\newcommand{\TO}{\mathsf{TO}}
\newcommand{\TOO}{\mathsf{TRF}}

\newcommand{\Location}[1]{\mathsf{loc}(#1)}

\newcommand{\SysEvents}{\mathcal{E}}

\newcommand{\Path}{\rightsquigarrow}

\newcommand{\Observation}{\operatorname{RF}}

\newcommand\numberthis{\addtocounter{equation}{1}\tag{\theequation}}

\newcommand{\set}[1]{\{#1\}}
\newcommand{\setpred}[2]{\{#1 \,|\, #2\}}

\usetikzlibrary{arrows.meta,calc,decorations.markings,math,arrows.meta}

\preto\tabular{\setcounter{magicrownumbers}{0}}
\newcounter{magicrownumbers}

\pgfdeclarelayer{bg}    
\pgfsetlayers{bg,main}  

\SetKw{Continue}{continue}

\SetCommentSty{mycommfont}

\makeatletter
\newcommand*{\centerfloat}{%
  \parindent \z@
  \leftskip \z@ \@plus 1fil \@minus \textwidth
  \rightskip\leftskip
  \parfillskip \z@skip}
\makeatother

\renewcommand{\smallskip}{}
\makeatletter

\def\thmt@rst@storecounters#1{%
\vspace{-1ex}%
  \bgroup
  \def\@currentlabel{}%
  \@for\thmt@ctr:=\thmt@innercounters\do{%
    \thmt@sanitizethe{\thmt@ctr}%
    \protected@edef\@currentlabel{%
      \@currentlabel
      \protect\def\@xa\protect\csname the\thmt@ctr\endcsname{%
        \csname the\thmt@ctr\endcsname}%
      \ifcsname theH\thmt@ctr\endcsname
        \protect\def\@xa\protect\csname theH\thmt@ctr\endcsname{%
          (restate \protect\theHthmt@dummyctr)\csname theH\thmt@ctr\endcsname}%
      \fi
      \protect\setcounter{\thmt@ctr}{\number\csname c@\thmt@ctr\endcsname}%
    }%
  }%
  \label{thmt@@#1@data}%
  \egroup
  
}%

\makeatother

\begin{document}

\title[]{The Complexity of Dynamic Data Race Prediction}


\author{Umang Mathur}
\affiliation{
\institution{University of Illinois, Urbana Champaign}            
\country{USA}                    
}
\email{umathur3@illinois.edu}          

\author{Andreas Pavlogiannis}
\affiliation{
\institution{Aarhus University}            
\country{Denmark}                    
}
\email{pavlogiannis@cs.au.dk}          

\author{Mahesh Viswanathan}
\affiliation{
\institution{University of Illinois, Urbana Champaign}            
\country{USA}                    
}
\email{vmahesh@illinois.edu}          


\begin{abstract}

Writing concurrent programs is notoriously hard due to scheduling
non-determinism.  The most common concurrency bugs are data races,
which are accesses to a shared resource that can be executed
concurrently.  Dynamic data-race prediction is the most standard
technique for detecting data races: given an observed, data-race-free
trace $\Trace$, the task is to determine whether $\Trace$ can be
reordered to a trace $\Trace^*$ that exposes a data-race.  Although
the problem has received significant practical attention for over
three decades, its complexity has remained elusive.  In this work, we
address this lacuna, identifying sources of intractability and
conditions under which the problem is efficiently solvable. Given a
trace $\Trace$ of size $n$ over $k$ threads, our main results are as
follows.

First, we establish a general $O(k\cdot n^{2\cdot (k-1)})$
upper-bound, as well as an $O(n^k)$ upper-bound when certain
parameters of $\Trace$ are constant.  In addition, we show that the
problem is NP-hard and even W[1]-hard parameterized by $k$, and
thus unlikely to be fixed-parameter tractable.  
Second, we study the
problem over acyclic communication topologies, such as server-clients hierarchies.  
We establish an
$O(k^2\cdot d\cdot n^2\cdot \log n)$ upper-bound, where $d$ is the
number of shared variables accessed in $\Trace$.  In addition, we show
that even for traces with $k=2$ threads, the problem has no
$O(n^{2-\epsilon})$ algorithm under the Orthogonal
Vectors conjecture.  Since any trace with 2 threads defines an acyclic
topology, our upper-bound for this case is 
optimal wrt polynomial improvements for up to moderate values of $k$ and
$d$.  Finally, motivated by existing heuristics, we study a
distance-bounded version of the problem, where the task is to expose a
data race by a witness trace that is similar to $\Trace$.  We develop
an algorithm that works in $O(n)$ time when certain parameters of
$\Trace$ are constant.
\end{abstract}

\ccsdesc[500]{Theory of computation~Parameterized complexity and exact algorithms}
\ccsdesc[500]{Software and its engineering~Software testing and debugging}

\keywords{Data Race Prediction, Complexity}


\maketitle


\section{Introduction}\label{sec:intro}

A concurrent program is said to have a data race if it
can exhibit an execution in which two conflicting accesses\footnote{Two accesses
are conflicting if they access the same memory location, with one of
them being a write access.} to the same memory
location are ``concurrent''.  Data races in concurrent programs are
often symptomatic of bugs in software like data
corruption~\cite{boehmbenign2011,racemob2013,Narayanasamy2007}, pose
challenges in defining the semantics of programming languages, and
have led to serious problems in the past~\cite{SoftwareErrors2009}; it
is no surprise that data races have been deemed \emph{pure
  evil}~\cite{evil2012}.  Automatically finding data races in programs
remains a widely studied problem because of its critical importance in
building correct concurrent software. Data-race detection techniques
can broadly be classified into static and dynamic. Given
that the race-detection problem in programs is undecidable, static
race detection approaches~\cite{Pratikakis11,Naik06} are typically
conservative, produce false alarms, and do not scale to large
software. On the other hand, since dynamic
approaches~\cite{Savage97,Mattern89,Pozniansky03,Flanagan09} have the
more modest goal of discovering data races by analyzing a
\emph{single} trace, they are lightweight, and can often scale to
production-level software. Moreover, many dynamic approaches are
\emph{sound}, i.e., do not raise false race reports. The effectiveness
and scalability of dynamic approaches has lead to many practical advances on the topic.
Despite a wide-spread interest on the problem, characterizing its complexity has remained elusive.

Informally, the \emph{dynamic race prediction} problem is the
following: given an observed trace $\Trace$ of a multi-threaded
program, determine if $\Trace$ demonstrates the presence of a data
race in the program that generates $\Trace$. 
This means that
either $\Trace$ has two conflicting data accesses that are
concurrent, or a different trace resulting from scheduling the threads
of $\Trace$ in a different order, witnesses such a race. Additional
traces that result from alternate thread schedules are captured by the
notion of a \emph{correct reordering} of $\Trace$, that characterizes
a set of traces that can be exhibited by \emph{any} program that
can generate $\Trace$; a precise definition of correct reordering is
given in \cref{subsec:model}. So formally, the data race prediction
problem is, given a trace $\Trace$, determine if there is a correct
reordering of $\Trace$ in which a pair of conflicting data accesses
are concurrent.

While the data race prediction problem is clearly in NP --- guess a
correct reordering and check if it demonstrates a data race --- its
precise complexity has not been identified. Evidence based on prior work,
suggests a belief that the problem might be NP-complete. First,
related problems, like data-race detection for programs with
strong synchronization
primitives~\cite{NetzerMiller90,NetzerMiller92,NetzerMiller89}, or
verifying sequential consistency~\cite{Gibbons97}, are known
to be NP-hard. Second, all known ``complete'' algorithms run in
worst-case exponential time. These approaches either rely on an
explicit enumeration of all correct reorderings~\cite{sen2005,Chen07},
or they are symbolic approaches that reduce the race prediction
problem to a constraint satisfaction
problem~\cite{Wang2009,Said11,Huang14}. On the other hand, a slew of
``partial order''-based methods have been proposed, whose goal is to
predict data races in polynomial time, but at the cost of being
incomplete and failing to detect data races in some traces. These
include algorithms based on the classical \emph{happens-before}
partial order~\cite{Lamport78,Mattern89,tsan2009,Flanagan09,Mathur18}, and
those based on newer partial orders that improve the prediction of
data races over
happens-before~\cite{Smaragdakis12,Kini17,Roemer18,Pavlogiannis19}.

In this paper we study the problem of data-race prediction from a complexity-theoretic perspective. 
Our goal is to understand whether the problem is intractable, the causes for intractability, and
conditions under it can can be solved efficiently. 
We provide partial answers to all these questions, 
and in some cases characterize the tractability/intractability landscape precisely in the form of optimality results.


\Paragraph{Contributions.}
Consider an input trace $\Trace$ of size $n$ over $k$ threads.
Our main contributions are as follows.
We refer to \cref{subsec:summary} for a formal summary.

Our first result shows that the data-race prediction problem is solvable in $O(k\cdot n^{2\cdot (k-1)})$ time,
and can be improved to $O(n^k)$ when certain additional parameters of $\Trace$ are constant.
We note that most benchmarks used in practice have a constant number of threads~\cite{Flanagan09,Smaragdakis12,Kini17,Mathur18,Roemer18,Pavlogiannis19} ,
and in such cases our upper-bound is polynomial.

The observation that data race predication is in polynomial
time for constantly many threads naturally leads to two follow-up
questions. Does the problem remain tractable for any $k$?
And if not, is it fixed parameter tractable (FPT) wrt
$k$, i.e., is there an algorithm with running time of the form
$O(f(k)\cdot n^{O(1)})$? 
Our second result answers both these questions in the negative, by showing that the problem is W[1]-hard. 
This formally establishes the NP-hardness of the problem, 
and excludes efficient algorithms when $k$ is even moderately large (e.g., $k=\Omega(\log n)$).

We then investigate whether there are practically relevant contexts
where data-race prediction is more efficiently solvable,
i.e., the degree of the polynomial is fixed and independent of $k$.
We consider the case of traces over acyclic communication topologies, such as pipelines, server-clients hierarchies and divide-and-conquer parallelism. 
Our third result shows that, perhaps surprisingly, over such topologies
data-race prediction can be solved in $O(k^2\cdot d\cdot n^2\cdot \log n)$ time, where 
$d$ is the total number of synchronization variables (locks) and global memory locations. 

In practice, the size $n$ of the trace is by far the dominating parameter, while $k$ and $d$ are many orders of magnitude smaller.
Hence, given the above upper-bound, the relevant question is whether the complexity on $n$ can be improved further.
Our fourth result shows that this is unlikely: we show that, under the Orthogonal Vectors conjecture, there is no $O(n^{2-\epsilon})$ algorithm even for traces restricted to only $2$ threads.
As any trace with $2$ threads induces an acyclic topology, our upper-bound is (conditionally) optimal wrt polynomial improvements.

Finally, the majority of practical data-race prediction heuristics search for a data race witness among correct reorderings that are very similar to the observed trace $\Trace$, i.e., by only attempting a few event reorderings on $\Trace$.
Motivated by these approaches, we investigate the complexity of a distance-bounded version
of data-race predication, where the goal is to expose a
data race by only looking at correct reorderings of $\Trace$ that are a small distance away.
Here, distance between traces is measured by the number of critical sections and write events whose
order is reversed from $\Trace$. 
Our fifth result is a linear-time (and thus, optimal) algorithm for this problem, when certain parameters of the trace $\Trace$ are constant.
This result gives a solid basis for the principled development of fast heuristics for dynamic data-race prediction.

\SubParagraph{Technical contributions.}
Towards our main results, we make several technical contributions that might be of independent interest.
We summarize some of them below.
\begin{compactenum}
\item We improve the lower-bound of the well-known problem on verifying sequential consistency with read-mapping (VSC-rm)~\cite{Gibbons97}  from the long-lasting NP-hardness to W[1]-hard.
\item We show that VSC-rm can be solved efficiently on tree communication topologies of any number of threads, which improves a recent result of~\cite{Pavlogiannis19} for only $2$ threads, as well as a result of~\cite{Chalupa18} for more than $2$ threads.
\item The first challenge in data-race prediction given a trace $\Trace$ is to choose the set $X$ of events of $\Trace$ over which to attempt to construct a correct reordering. 
Identifying such choices for $X$ is a significant challenge~\cite{Huang14,Roemer18,Pavlogiannis19}.
We establish non-trivial upper-bounds on the number of choices for $X$, 
and show that they are constantly many when certain parameters of $\Trace$ are constant.
\item Particularly for tree communication topologies, we show that a single choice for such $X$ suffices.
\end{compactenum}

Finally, we note that our notion of a predictable data race in a trace $\Trace$ requires as a witness a reordering $\Trace^*$ of $\Trace$ in which every read event reads from the same write event as in $\Trace$.
This guarantees that $\Trace^*$ is valid in any program that produced $\Trace$.
More permissive reorderings, e.g., requiring that every read event reads the same value, are also possible, and can capture potentially more races.
Our notion of witness reflects the most common practice in race-detection literature, where trace logging typically does not track the values.


\Paragraph{Related Work.}
Antoni Mazurkiewicz~\cite{Mazurkiewicz87,Aalbersberg88} used the notion of 
\emph{traces} to mathematically
model executions of concurrent programs. 
Bertoni et. al.~\cite{Bertoni89} studied various language-theoretic
questions about Mazurkiewicz traces.
The folklore results about the NP-hardness of race detection
are often attributed to Netzer and Miller~\cite{NetzerMiller90,NetzerMiller92,NetzerMiller89}.
However, the problem considered in their work differs in significant ways
from the problem of data-race prediction.
First, the notion of feasible executions in~\cite{NetzerMiller90}
(the counterpart of the notion of correct reorderings) requires that
any two conflicting events be ordered in the same way as the observed execution, and hence, are less permissive.
Next, the NP-hardness arises from the use of complex synchronization primitives
like \texttt{wait} and \texttt{signal}, which are more powerful 
than the primitives we study here (release/acquire of locks and read/write of registers).
The results due to Netzer and Miller, thus, do not apply to the problem of data-race prediction.
Gibbons and Korach~\cite{Gibbons97} establish NP-hardness for
a closely related problem in distributed computing --- verifying sequential consistency with a read mapping (VSC-rm).
Yet again, the problem is different than the problem of race prediction.
Complexity theoretic investigations have also been undertaken for
other problems in distributed computing like linearizability~\cite{HerlihyWing90,emmi17,Gibbons97}, serializability~\cite{Papadimitriou79}
and transactional consistency~\cite{biswas19}.
Hence, although there have been many theoretical results on related problems in concurrency,
none of them addresses dynamic data-race prediction.
Our work fills this gap.

Some proofs are relegated to the appendix.


\section{Preliminaries}\label{sec:preliminaries}


\subsection{Model}\label{subsec:model}

\Paragraph{General notation.}
Given a natural number $k$, let $[k]=\{1,\dots, k\}$.
Given a function $f:X\to Y$, we let $\Domain(f)=X$ and $\Image(f)=Y$.
Given two functions $f, g$, we write $f\subseteq g$ to denote that $\Domain(f)\subseteq \Domain(g)$ and for every $x\in \Domain(f)$ we have $f(x)=g(x)$.
Given a set $X'\subseteq \Domain(f)$, 
we denote by $f\Project X'$ the function with $\Domain(f\Project X') = X'$ and $f\Project X'\subseteq f$.

\Paragraph{Concurrent program.}
We consider a shared-memory concurrent program $\System$ that consists of $k$ 
threads $\{\Process_i\}_{i\in [k]}$, under sequential consistency semantics~\cite{Shasha88}.
For simplicity of presentation we assume no thread is created dynamically
and the set $\{\Process_i\}_{i\in [k]}$ is known a-priori.
Communication between threads occurs over a set of global variables $\Globals$, 
and synchronization over a set of locks $\Locks$ such that $\Globals \cap \Locks = \emptyset$.
We let $\Vars=\Globals\cup \Locks$ be the set of all variables of $\System$.
Each thread is deterministic, and performs a sequence of operations.
We are only interested in the operations that access a global variable or a lock, which are called \emph{events}. 
In particular, the allowed events are the following.
\begin{compactenum}
\item Given a global variable $x\in \Globals$, a thread can either \emph{write} to $x$ via an event $\Write(x)$ or \emph{read} from $x$ via an event $\Read(x)$.
\item Given a lock $\lk\in \Locks$, a thread can either \emph{acquire} $\lk$ via an event $\Acquire(\lk)$ or \emph{release} $\lk$ via an event $\Release(\lk)$.

\end{compactenum}
Each event is atomic, represented by a tuple $(a,b,c,d)$, where
\begin{compactenum}
\item $a\in\set{\Write, \Read, \Acquire, \Release}$ represents the type of the event (i.e., write, read, lock-acquire or lock-release event),
\item $b$ represents the variable or lock that the event accesses, 
\item $c$ is the thread of the event, and
\item $d$ is a unique identifier of the event. 
\end{compactenum}
Given an event $\Event$, we let $\Location{\Event}$ denote the global variable or lock that $\Event$ accesses.
We occasionally write $\Event(x)$ to denote an event $\Event$ with $\Location{\Event}=x$, while the thread and event id is often implied by the context.
We denote by $\SysWrites_{\Process}$ 
(resp. $\SysReads_{\Process}$, $\SysAcquires_{\Process}$, $\SysReleases_{\Process}$) 
the set of all write (resp. read, lock-acquire, lock-release) events 
that can be performed by thread $\Process$.
We let $\SysEvents_{\Process}=\SysWrites_{\Process}\cup \SysReads_{\Process}\cup \SysAcquires_{\Process}\cup \SysReleases_{\Process}$. 
We denote by $\SysEvents=\bigcup_{\Process} \SysEvents_{\Process}$, $\SysWrites=\bigcup_{\Process} \SysWrites_{\Process}$, $\SysReads=\bigcup_{\Process} \SysReads_{\Process}$, $\SysAcquires=\bigcup_{\Process} \SysAcquires_{\Process}$, $\SysReleases=\bigcup_{\Process} \SysReleases_{\Process}$ the events, write, read, lock-acquire and lock-release events of the program $\System$, respectively.
Given an event $\Event\in \SysEvents$, we denote by $\Proc{\Event}$ the thread of $\Event$.
Finally, given a set of events $X\subseteq \SysEvents$, we denote by $\Reads{X}$ (resp., $\Writes{X}$, $\Acquires{X}$, $\Releases{X}$) the set of read (resp., write, lock-acquire, lock-release) events of $X$.
For succinctness, we let $\WritesReads{X}=\Writes{X}\cup\Reads{X}$, $\ReadsReleases{X}=\Reads{X}\cup \Releases{X}$ and $\WritesAcquires{X}=\Writes{X}\cup \Acquires{X}$.
The semantics of $\System$ are the standard for sequential consistency~\cite{Shasha88}. 

\Paragraph{Conflicting events.}
Given two distinct events $\Event_1, \Event_2\in \SysEvents$, we say that $\Event_1$ and $\Event_2$ are \emph{conflicting}, denoted by $\Confl{\Event_1}{\Event_2}$, if 
(i)~$\Location{\Event_1}=\Location{\Event_2}$ (i.e., both events access the same global variable or the same lock) and
(ii)~$\{\Event_1, \Event_2 \}\cap \SysWrites\neq \emptyset$ or
$\{\Event_1, \Event_2 \}\cap \SysAcquires\neq \emptyset$
i.e., at least one of them is either a write event or a lock-acquire event.
We extend the notion of conflict to sets of events in the natural way:
two sets of events $X_1, X_2\subseteq \SysEvents$ 
are called \emph{conflicting}, denoted by $\Confl{X_1}{X_2}$
if $\exists (\Event_1, \Event_2) \in (X_1 \times X_2)$ such that $\Confl{\Event_1}{\Event_2}$.

\Paragraph{Event sequences.}
Let $\Trace$ be a sequence of events.
We denote by $\Events{\Trace}$ the set of events, by $\Locks(\Trace)$ the set of locks, and by $\Globals(\Trace)$ the set of global variables in $\Trace$.
We let $\Writes{\Trace}$ (resp., $\Reads{\Trace}$, $\Acquires{\Trace}$, $\Releases{\Trace}$) denote the set $\Writes{\Events{\Trace}}$ (resp., $\Reads{\Events{\Trace}}$, $\Acquires{\Events{\Trace}}$, $\Releases{\Events{\Trace}}$), i.e., it is the set of write (resp., read, lock-acquire, lock-release) events of $\Trace$.
Given two distinct events $\Event_1, \Event_2\in \Events{\Trace}$, 
we say that $\Event_1$ \emph{is earlier than}  $\Event_2$ in $\Trace$,
denoted by $\Event_1 <_{\Trace} \Event_2$ iff $\Event_1$ appears before $\Event_2$ in $\Trace$.
We say that $\Event_1$ is \emph{thread-ordered earlier than} $\Event_2$,
denoted $\Event_1<_{\TO(\Trace)}\Event_2$, 
when $\Event_1 <_{\Trace} \Event_2$ and $\Proc{\Event_1}=\Proc{\Event_2}$.
For events $\Event_1, \Event_2 \in \Events{\Trace}$,
we say $\Event_1 \leq_\Trace \Event_2$
(resp. $\Event_1 \leq_{\TO(\Trace)} \Event_2$)
if either $\Event_1 = \Event_2$ or $\Event_1 <_\Trace \Event_2$
(resp. $\Event_1 <_{\TO(\Trace)} \Event_2$).
We will often use $<_\TO$ (resp. $\leq_\TO$) in place of $<_{\TO(\Trace)}$ (resp. $\leq_{\TO(\Trace)}$)
when the trace $\Trace$ is clear from context.
Given a set of events $X\subseteq \SysEvents$, we denote by $\Trace\Project X$ the \emph{projection} of $\Trace$ onto $X$.
Given a thread $\Process_i$, we let $\Trace\Project\Process_i=\Trace\Project\SysEvents_{\Process_i}$.
Given two sequences $\Trace_1, \Trace_2$, we denote by $\Trace_1\circ \Trace_2$ their concatenation.

\Paragraph{Lock events.}
Given a sequence of events $\Trace$ and a lock-acquire event $\Acquire\in \Acquires{\Trace}$, 
we denote by $\Match{\Trace}{\Acquire}$ the earliest lock-release event
 $\Release\in \Releases{\Trace}$ such that $\Confl{\Release}{\Acquire}$ 
 and $\Acquire<_{\TO}\Release$, and let $\Match{\Trace}{\Acquire}=\bot$ if no such lock-release event exists.
If $\Match{\Trace}{\Acquire}\neq \bot$, we require that $\Proc{\Acquire}=\Proc{\Match{\Trace}{\Acquire}}$, i.e., the two lock events belong to the same thread.
Similarly, given a lock-release event $\Release\in \Releases{\Trace}$, we denote by $\Match{\Trace}{\Release}$ the latest acquire event $\Acquire\in \Acquires{\Trace}$ such that $\Match{\Trace}{\Acquire}=\Release$ and require that such a lock-acquire event always exists.
Given a lock-acquire event $\Acquire$, the \emph{critical section} $\CS{\Trace}{\Acquire}$ is the set of events $\Event$ such that 
(i)~$\Acquire<_{\TO}\Event$ and
(ii)~if $\Match{\Trace}{\Acquire}\neq \bot$, then $\Event<_{\TO}\Match{\Trace}{\Acquire}$.
For simplicity of presentation, we assume that locks are not re-entrant.
That is, for any two lock-acquire events 
$\Acquire_1, \Acquire_2$ with $\Confl{\Acquire_1}{\Acquire_2}$ and 
$\Acquire_1<_{\TO}\Acquire_2$,
we must have $\Match{\Trace}{\Acquire_1} <_{\TO} \Acquire_2$.
The \emph{lock-nesting depth} of $\Trace$ is the maximum number $\ell$ such that there exist distinct lock-acquire events
$\{\Acquire_i\}_{i=1}^{\ell}$ with 
(i)~$\Acquire_1<_{\TO}\Acquire_2<_{\TO}\dots<_{\TO}\Acquire_{\ell}$, and
(ii)~for all $i\in [\ell]$, if $\Match{\Trace}{\Acquire_i}\in \Events{\Trace}$ then $\Acquire_{\ell}<_{\TO}\Match{\Trace}{\Acquire_i}$.

\Paragraph{Traces and reads-from functions.}
An event sequence $\Trace$ is called a \emph{trace} if 
for any two lock-acquire events 
$\Acquire_1,\Acquire_2\in \Acquires{\Trace}$, if $\Location{\Acquire_1}=\Location{\Acquire_2}$ 
and $\Acquire_1<_{\Trace}\Acquire_2$,
then $\Release_1 =\Match{\Trace}{\Acquire_1}\in \Releases{\Trace}$ and $\Release_1<_{\Trace}\Acquire_2$.
A trace therefore ensures that locks
obey mutual exclusion, i.e., critical sections over the same lock
cannot overlap.

Given a trace $\Trace$, we define its \emph{reads-from function} $\Observation_{\Trace}:\Reads{\Trace}\to \Writes{\Trace}$ as follows:
$\Observation_{\Trace}(\Read)=\Write$ iff
$\Write<_{\Trace} \Read \text{ and } \forall \Write'\in \Writes{\Trace} \text{ with } \Confl{\Write}{\Write'}$, we have $\Write'<_{\Trace} \Read \Implies \Write'<_{\Trace} \Write$.
That is, $\Observation_{\Trace}$ maps every read event $\Read$ to the write event $\Write$ that $\Read$ observes in $\Trace$.
For simplicity, we assume that $\Trace$ starts with a write event to every location, hence $\Observation_{\Trace}$ is well-defined.
For notational convenience, we extend the reads-from function $\Observation_{\Trace}$ to lock-release events, such that, for any lock-release event $\Release\in \Releases{\Trace}$, we have $\Observation_{\Trace}(\Release)=\Match{\Trace}{\Release}$, 
i.e., $\Release$ observes its matching lock acquire event.

\Paragraph{Correct reordering, enabled events and predictable data races.}
A trace $\Trace^*$ is a correct reordering of trace
$\Trace$ if 
\begin{enumerate*}[label=(\roman*)]
\item $\Events{\Trace^*} \subseteq \Events{\Trace}$,
\item for every thread $\Process_i$,
we have that $\Trace^*\Project \Process_i$ is a prefix of $\Trace\Project \Process_i$, and
\item $\Observation_{\Trace^*}\subseteq \Observation_{\Trace}$,
i.e., the reads-from functions of $\Trace^*$ and $\Trace$ agree on their 
common read and lock-release events.
\end{enumerate*}
Given a trace $\Trace$, an event $\Event \in \Events{\Trace}$ and a correct reordering
$\Trace^*$ of $\Trace$, we say that $\Event$ is \emph{enabled}
in $\Trace^*$ if $e \not\in \Events{\Trace^*}$
and for every $e' \in \Events{\Trace}$ such that
$e' <_{\TO} e$, we have that $e' \in \Events{\Trace^*}$.
Given two conflicting events $\Event_1, \Event_2 \in \Events{\Trace}$
with $\Location{\Event_1} = \Location{\Event_2} \in \Globals$,
we say the pair $(\Event_1, \Event_2)$ is a predictable data race
of trace $\Trace$ if there is a correct reordering
$\Trace^*$ of $\Trace$ such that both $\Event_1$ and $\Event_2$ are enabled
in $\Trace^*$.
Finally, we say $\Trace$ has a predictable data race if there is a pair $(\Event_1, \Event_2)$
which is a predictable data race of $\Trace$.

Note that predictability of a race is defined with respect to 
a correct reordering in which every read event observes the same write event.
This requirement guarantees that the correct reordering is a valid trace of 
any concurrent program that produced the initial trace.
Hence, every such program is racy.
More permissive notions of predictability can also be defined, e.g., by requiring that, in a correct reordering, every read event reads the same value (possibly from a different write event).
This alternative definition would capture potentially more predictable races.
Our definition of correct reorderings reflects the most common practice in race-detection literature,
where trace logging typically does not track the values~\cite{Kini17,Mathur18,Smaragdakis12,Pavlogiannis19,Roemer18}.

\Paragraph{The communication topology.}
The trace $\Trace$ naturally induces a \emph{communication topology} 
graph $G=(V,E)$ where
\begin{enumerate*}[label=(\roman*)]
\item $V = \set{\Process_i}_i$ and
\item $E = \setpred{(\Process_i, \Process_j)}{i\neq j \text{ and } \Confl{\Events{\Process_i}}{\Events{\Process_j}}}$.
\end{enumerate*}
In words, we have one node in $G$ per thread, and there is  an edge between two distinct nodes if the corresponding threads execute conflicting events
(note that $G$ is undirected).
For simplicity, we assume that $G$ is connected.
In later sections, we will make a distinction between tree topologies (i.e., that do not contain cycles) and general topologies (that might contain cycles).
Common examples of tree topologies include stars (e.g., server-clients), pipelines, divide-and-conquer parallelism, and the special case of two threads.


\subsection{Problem Statement}\label{subsec:problem}

In the dynamic data-race prediction problem, 
we are given an observed trace $\Trace$, 
and the task is to identify whether $\Trace$ has a predictable data race.
In this work we focus on the following decision problem --- given a trace $\Trace$ and two (read or write) conflicting events $\Event_1, \Event_2\in \Events{\Trace}$, the task is to decide whether $(\Event_1, \Event_2)$ is a predictable data race of $\Trace$.
Clearly, having established the complexity of the decision problem, the general problem can be solved by answering the decision problem for all $O(n^2)$ pairs of conflicting variable access events of $\Trace$.
In the other direction, as the following lemma observes, detecting whether $\Trace$ has some predictable data race is no easier than detecting whether a given event pair of $\Trace$ constitutes a predictable data race.
We refer to \cref{sec:proofs_preliminaries} for the proof.

\smallskip
\begin{restatable}{lemma}{lemdecisiongivenpair}\label{lem:decision_given_pair}
Given a trace $\Trace$ of length $n$ and two events $\Event_1, \Event_2\in \Events{\Trace}$, we can construct a trace $\Trace'$ in $O(n)$ time so that $\Trace'$ has a predictable data race iff $(\Event_1, \Event_2)$ is a predictable data race of $\Trace$ .
\end{restatable}

To make the presentation simpler, 
we assume w.l.o.g that there are no open critical sections in $\Trace$,
i.e., every lock-acquire event $\Acquire$ is followed by a matching lock-release event $\Match{\Trace}{\Acquire}$.
Motivated by practical applications, we also study the complexity of dynamic data-race prediction parameterized by a notion of distance between the input trace $\Trace$ and the witness $\Trace^*$ that reveals the data race.

\Paragraph{Trace distances.}
Consider a trace $\Trace$ and a correct reordering $\Trace'$ of $\Trace$.
The set of \emph{reversals} between $\Trace$ and $\Trace'$ is defined as 
\begin{align*}
\Reversals(\Trace,\Trace')=& \setpred{
(\Write_1, \Write_2)\in \WritesAcquires{\Trace'}\times \WritesAcquires{\Trace'}}{\\
& \Confl{\Write_1}{\Write_2}
\text{ and }
\Write_1<_{\Trace}\Write_2
\text{ and }
\Write_2<_{\Trace'}\Write_1
}\ .
\end{align*}
In words, $\Reversals(\Trace,\Trace')$ contains the pairs of
conflicting write events or lock-acquire events, 
the order of which has been \emph{reversed} in $\Trace'$ 
when compared to $\Trace$.
The \emph{distance} of $\Trace'$ from $\Trace$ is defined as
$\Distance(\Trace,\Trace')=|\Reversals(\Trace,\Trace')|$.
Our notion of distance, thus, only counts the number of reversals
of conflicting write or lock-acquire events instead of 
counting reversals over all events 
(or even conflicting write-read events).

\Paragraph{Distance-bounded dynamic data race prediction.}
Consider a trace $\Trace$ and two events $\Event_1,\Event_2$ of $\Trace$.
Given an integer $\ell\geq 0$, the $\ell$-distance-bounded dynamic data-race 
prediction problem is the promise problem\footnote{
The promise problem~(\cite{promise_problem}) given languages $L_\True$ and $L_\False$ is
to design an algorithm $A$ such that $A(x) = \True$ for every $x \in L_\True$,
$A(x) = \False$ for every $x \in L_\False$, and all for all other 
inputs $x \not\in L_\True \cup L_\False$,
the output $A(x)$ of the algorithm is allowed to be any of $\True$ or $\False$.}
that allows for any answer ($\True/\False$) if $(\Event_1, \Event_2)$ 
is a predictable data race of $\Trace$ and every witness correct reordering $\Trace^*$ is such that $\Distance(\Trace,\Trace^*) > \ell$.

\subsection{Summary of Main Results}\label{subsec:summary}

Here we state the main results of this work,
and present the technical details in the later parts of the paper.

\subsubsection{The General Case}\label{subsec:general_case}

First, we study the complexity of the problem with respect to 
various parameters of the input trace.
These parameters are the number of threads, the number of variables, the lock-nesting depth, as well as the lock-dependence factor, which, intuitively, measures the amount of data flow between critical sections.
In the following, $<_\TOO$
(formal definition in~\cref{sec:ideals}),
is the smallest partial order that contains $<_\TO$,
and also orders read events after their corresponding observed write event
\big(i.e., $\Observation(\Read) <_\TOO \Read$ 
for every $\Read \in \ReadsReleases{\Events{\Trace}}$ \big).
%

\Paragraph{The lock-dependence factor.}
The \emph{lock-dependence graph} of a trace $\Trace$ is the graph $G_{\Trace}=(V_{\Trace},E_{\Trace})$ defined as follows. 
\begin{compactenum}
\item The set of vertices is $V_{\Trace}=\Acquires{\Trace}$, i.e., it is the set of lock-acquire events of $\Trace$.
\item The set of edges is such that $(\Acquire_1, \Acquire_2)\in E_{\Trace}$ if
\begin{enumerate*}[label=(\roman*)] 
\item $\Acquire_1\not<_{\TOO} \Acquire_2$, 
\item $\Acquire_1<_{\TOO} \Match{\Trace}{\Acquire_2}$, and
\item $\Match{\Trace}{\Acquire_1}\not<_{\TOO}\Match{\Trace}{\Acquire_2}$.
\end{enumerate*}
\end{compactenum}
Given a lock-acquire event $\Acquire\in V_{\Trace}$, let $A_{\Acquire}$ be the set of lock-acquire events that can reach $\Acquire$ in $G_{\Trace}$.
We define the \emph{lock dependence factor} of $\Trace$ as 
$
\max_{\Acquire\in V_{\Trace}}|A_{\Acquire} |
$.
We show the following theorem.

\smallskip
\begin{restatable}{theorem}{them:generalparameterized}\label{them:general_parameterized}
Consider a trace $\Trace$ of length $n$, $k$ threads, lock-nesting depth $\NestingDepth$, and lock-dependence factor $\LockFactor$.
The dynamic data-race prediction problem on $\Trace$ can be solved in $O(\alpha\cdot \beta)$ time, where
$\alpha=\min(n, k\cdot \NestingDepth \cdot \LockFactor)^{k-2}$ and
$\beta=k\cdot n^k$.
\end{restatable}

In particular, the problem is polynomial-time solvable for a fixed number of threads $k$.
In practice, the parameters $k$, $\NestingDepth$ and $\LockFactor$ behave as constants,
and in such cases our upper-bound becomes $O(n^k)$.
\cref{them:general_parameterized} naturally leads to two questions, namely 
\begin{enumerate*}[label=(\roman*)]
\item whether there is a polynomial-time algorithm for any $k$, and
\item\label{itm:fpt} if not, whether the problem is FPT with respect to the parameter $k$,
 i.e., can be solved in $O(f(k)\cdot n^{O(1)})$ time, for some function $f$.
\end{enumerate*}
Question~\ref{itm:fpt} is very relevant, as typically $k$ is several orders of magnitude smaller than $n$.
We complement \cref{them:general_w1hard} with the following lower-bound, which
answers both questions in negative.

\smallskip
\begin{restatable}{theorem}{themgeneralwhard}\label{them:general_w1hard}
The dynamic data-race prediction problem is W[1]-hard parameterized by the number of threads. 
\end{restatable}

\subsubsection{Tree Communication Topologies}\label{subsec:trees}

Next, we study the problem for tree communication topologies, such as pipelines and server-clients architectures.
We show the following theorem.

\smallskip
\begin{restatable}{theorem}{themtreedopologiesub}\label{them:tree_topologies_ub}
Let $\Trace$ be a trace over a tree communication topology with
$n$ events, $k$ threads and $d$ variables.
The dynamic data-race prediction problem for $\Trace$ can be solved in $O(k^2\cdot d\cdot n^2\cdot \log n)$ time.
\end{restatable}

Perhaps surprisingly, in sharp contrast to \cref{them:general_w1hard}, for tree topologies there exists an efficient algorithm where the degree of the polynomial is fixed and does not depend on any input parameter (e.g., number of threads).
Note that the dominating factor in this complexity is $n^2$, while $k$ and $d$ are typically much smaller.
Hence, the relevant theoretical question is whether the dependency on $n$ can be improved further.
We show that this is unlikely, by complementing \cref{them:tree_topologies_ub} with the following conditional lower-bound,
based on the Orthogonal Vectors conjecture~\cite{Bringmann19}.

\smallskip
\begin{restatable}{theorem}{themtreetopologieslb}\label{them:tree_topologies_lb}
Let $\Trace$ be a trace with $n$ events, $k\geq 2$ threads and $d\geq 9$ shared global variables with at least one lock.
There is no algorithm that solves the decision problem of dynamic data-race prediction for $\Trace$ in time $O(n^{2-\epsilon})$, for any $\epsilon>0$, unless the Orthogonal Vectors conjecture fails.
\end{restatable}
Since $k=2$ implies a tree communication topology, the result of \cref{them:tree_topologies_ub} is conditionally optimal, up-to poly-logarithmic factors, for a reasonable number of threads and variables (e.g., when $k,d=\log^{O(1)}(n)$).

\subsubsection{Witnesses in Small Distance}\label{subsec:small_distance}

Finally, we study the problem in more practical settings, namely, when
(i)~the number of threads, lock-nesting depth, lock-dependence factor of $\Trace$ are bounded, and
(ii)~we are searching for a witness at a small distance from $\Trace$.

\smallskip
\begin{restatable}{theorem}{themsmalldistance}\label{them:small_distance}
Fix a reversal bound $\ell\geq 0$.
Consider a trace $\Trace$ of length $n$ and constant number of threads, lock-nesting depth and lock-dependence factor.
The $\ell$-distance-bounded dynamic data-race prediction problem for $\Trace$ can be solved in $O(n)$ time.
\end{restatable}


\section{Trace Ideals}\label{sec:ideals}



\subsection{Partial Orders}\label{subsec:po}

\Paragraph{Partially ordered sets.}
A \emph{partially ordered set} (or \emph{poset}) 
is a pair $(X,P)$ where $X$ is a set of (write, read, lock-acquire, lock-release) 
events and $P$ is a reflexive, antisymmetric and transitive relation over $X$.
We will often 
write $\Event_1 \leq_P \Event_2$ to denote $(\Event_1, \Event_2) \in P$.
Given two events $\Event_1, \Event_2\in X$ we write $\Event_1 <_{P}\Event_2$ to denote that $\Event_1\leq_{P}\Event_2$ and $\Event_1\neq \Event_2$,
and write $\Event_1\Prec_{P}\Event_2$ to denote that $\Event_1<_{P} \Event_2$ and there exists no event $\Event$ such that 
$\Event_1<_{P}\Event<_{P}\Event_2$.
Given two distinct events $\Event_1,\Event_2\in X$, we say that $\Event_1$ and $\Event_2$ are \emph{unordered} by $P$, denoted by $\Unordered{\Event_1}{P}{\Event_2}$, if neither $\Event_1<_{P}\Event_2$ nor $\Event_2<_{P} \Event_1$.
We call an event $\Event\in X$ \emph{maximal} if there exists no $\Event'\in X$ such that $\Event<_{P}\Event'$.
Given a set $Y\subseteq X$, we denote by $P\Project Y$ the \emph{projection} of $P$ on $Y$,
i.e., we have $P\Project Y \subseteq Y\times Y$, and for all $\Event_1,\Event_2\in Y$, we have $\Event_1\leq_{P\Project Y} \Event_2$ iff $\Event_1\leq_{P} \Event_2$.
Given two posets $(X, P)$ and $(X,Q)$, we say that the partial order $Q$ \emph{refines} $P$, denoted by $Q\Refines P$, if
for every two events $\Event_1, \Event_2\in X$, if $\Event_1\leq_{P}\Event_2$ then $\Event_1\leq_{Q}\Event_2$.
If $Q$ refines $P$, we say that $P$ is \emph{weaker} than $Q$.
We denote by $Q\StrictRefines P$ the fact that $Q\Refines P$ and $P\not \Refines Q$.
A \emph{linearization} of $(X, P)$ is a total order over $X$ that refines $P$.
An \emph{order ideal} (or simply \emph{ideal}) of a poset $(X,P)$ is subset $Y\subseteq X$ such that for every two events $\Event_1\in Y$ and $\Event_2\in X$ with $\Event_2\leq_{P} \Event_1$, we have $\Event_2\in Y$.
An event $\Event$ is
\emph{executable} in ideal $Y$ if $Y\cup \{ \Event \}$ is also an ideal of $(X,P)$.
The number of threads and variables of a poset $(X,P)$ is the number of threads and variables of the events of $X$.

\Paragraph{Partially ordered sets with reads-from functions.}
A \emph{poset with a reads-from function} (or \emph{rf-poset}) is a tuple $(X,P,\Observation)$ where
(i)~$\Observation\colon \ReadsReleases{X} \to \WritesAcquires{X}$ is a reads-from function such that for all $\Read\in \ReadsReleases{X}$, we have $\Observation(\Read)\in \Writes{X}$ iff $\Read\in \Reads{X}$, and
(ii)~$(X,P)$ is a poset where for all $\Read\in \ReadsReleases{X}$ we have $\Observation(\Read)<_{P} \Read$.
Notation from posets is naturally lifted to rf-posets, e.g., 
an ideal of $\OPoset$ is an ideal of $(X,P)$.

\Paragraph{Thread-reads-from order and trace ideals.}
Given a trace $\Trace$,
the thread-reads-from order $\TOO(\Trace)$
(or simply $\TOO$ when $\Trace$ is clear from context)
is the weakest partial order over the set $\Events{\Trace}$ 
such that 
\begin{enumerate*}[label=(\roman*)]
\item $\TOO\Refines \TO$, and
\item $(\Events{\Trace}, \TOO, \Observation_{\Trace})$ is an rf-poset.
\end{enumerate*}
In particular, $\TOO$ is the transitive closure of  
$\big(\TO \cup \setpred{\Observation_{\Trace}(\Read)< \Read}{\Read\in\Reads{\Trace}}\big)$.
A \emph{trace ideal} of $\Trace$ is an ideal $X$ of the poset $(\Events{\Trace}, \TOO)$.
We say an event $\Event\in \Events{\Trace}\setminus X$ is \emph{enabled} in X
if for every $\Event' <_\TO \Event$,
we have $\Event' \in X$.
We call $X$ \emph{lock-feasible} if for every two lock-acquire events $\Acquire_1, \Acquire_2\in \Acquires{X}$ with $\Confl{\Acquire_1}{\Acquire_2}$,
we have $\Match{\Trace}{\Acquire_i}\in X$ for some $i\in[2]$.
We call $X$ \emph{feasible} if it is lock-feasible, and there exists a partial order $P$ over $X$ such that 
(i)~$P\Refines \TOO\Project X$ and
(ii)~for every pair of lock-acquire events $\Acquire_1, \Acquire_2\in \Acquires{X}$ with $\Confl{\Acquire_1}{\Acquire_2}$,
and $\Match{\Trace}{\Acquire_1}\not\in X$, we have $\Release_2<_{P} \Acquire_1$, where $\Release_2=\Match{\Trace}{\Acquire_2}$.
If $X$ is feasible, we define the \emph{canonical rf-poset} of $X$ as $(X, Q, \Observation_{\Trace}\Project X)$, where $Q$ is the weakest among all such partial orders $P$.
It is easy to see that $Q$ is well-defined, i.e., there exists at most one weakest partial order among all such partial orders $P$.

\Paragraph{The realizability problem of feasible trace ideals.}
The realizability problem for an rf-poset $\OPoset=(X,P,\Observation)$ 
asks whether there exists a linearization $\Trace^*$ of $P$ such that 
$\Observation_{\Trace^*} = \Observation$.
Given a trace $\Trace$ and a feasible trace ideal $X$ of $\Trace$, the realizability problem for $X$ is the realizability problem of the canonical rf-poset $(X,P,\Observation)$ of $X$.
The following remark relates the decision problem of dynamic race prediction in $\Trace$ with the realizability of trace ideals of $\Trace$.

\smallskip
\begin{remark}\label{rem:witness_trace}
If $\Trace^*$ is a witness of the realizability of $X$,
then $\Trace^*$ is a correct reordering of $\Trace$.
Two conflicting events $\Event_1,\Event_2\in \Events{\Trace}$ are a predictable data race of $\Trace$ iff there exists a realizable trace ideal $X$ of $\Trace$ such that $\Event_1,\Event_2$ are enabled in $X$.
\end{remark}

\Paragraph{Read pairs and triplets.}
For notational convenience, we introduce the notion of read pairs and read triplets.
Given an rf-poset $\OPoset=(X,P,\Observation)$, a \emph{read pair} (or \emph{pair} for short) of $\OPoset$ is a pair $(\Write, \Read)$ such that $\Read\in \ReadsReleases{X}$ and $\Write=\Observation(\Read)$ (note that $\Write \in X$).
A \emph{read triplet} (or \emph{triplet} for short) is a triplet $(\Write, \Read, \Write')$ such that
(i)~$(\Write, \Read)$ is a pair of $\OPoset$, 
(ii)~$\Write'\in X$, and
(iii)~$\Write'\neq \Write$ and $\Confl{\Read}{\Write'}$.
We denote by $\ReadPairs{\OPoset}$ and $\ReadTriplets{\OPoset}$ the set of pairs and triplets of $\OPoset$, respectively.

\Paragraph{Closed rf-posets.}
We call an rf-poset $\OPoset=(X,P,\Observation)$ \emph{closed} if 
for every triplet $(\Write, \Read, \Write')\in \ReadTriplets{\OPoset}$, we have
(i)~if $\Write' <_{P} \Read$ then $\Write'<_{P}\Write$, and
(ii)~if $\Write <_{P} \Write'$ then $\Read<_{P} \Write'$.
Given, an rf-poset $\OPoset=(X,P,\Observation)$, the \emph{closure} of $\OPoset$
is an rf-poset $\OPosetQ = (X,Q,\Observation)$ where $Q$ is the weakest partial order over $X$ such that $Q\Refines P$ and $\OPosetQ$ is closed.
If no such $Q$ exists, we let the closure of $\OPoset$ be $\bot$.
The closure is well-defined~\cite{Pavlogiannis19}.
The associated Closure problem is, given an rf-poset $\OPoset$, decide whether the closure of $\OPoset$ is not $\bot$.

\smallskip
\begin{remark}\label{rem:closure}
An rf-poset is realizable only if its closure exists and is realizable.
\end{remark}


\subsection{Bounds on the Number of Feasible Trace Ideals}\label{subsec:ideal_bounds}

\cref{rem:witness_trace} suggests that the dynamic data-race 
prediction problem for a trace $\Trace$ is reducible to deciding whether $\Trace$ has some realizable trace ideal.
In general, if $\Trace$ has length $n$ and $k$ threads, 
there exist $n^k$ possible trace ideals to test for realizability.
Here we derive another upper-bound on the number of such ideals that are sufficient to test, based on the number of threads of $\Trace$, its lock-nesting depth and its lock-dependence factor.
These parameters typically behave as constants in practice, and thus understanding the complexity of dynamic data race prediction in terms of these parameters is crucial.


\smallskip\noindent{\bf Causal cones.}
Given an event $\Event\in \Events{\Trace}$, 
the \emph{causal cone} $\Past_{\Trace}(\Event)$ of $\Event \in \Events{\Trace}$ 
is the smallest trace ideal $X$ of $\Trace$ so that $\Event$ is enabled in $X$.
In words, we construct $\Past_{\Trace}(\Event)$ by taking the $\TOO$-downwards closure of the thread-local predecessor $\Event'$ of $\Event$ (i.e., $\Event'\Prec_{\TO}\Event)$.
Given a non-empty set of events $S\subseteq \Events{\Trace}$, we define the causal cone of $S$ as 
$\Past_{\Trace}(S)=\bigcup_{\Event\in S}\Past_{\Trace}(\Event)$;
notice that $\Past_{\Trace}(S)$ is a trace ideal.

\smallskip\noindent{\bf Candidate ideal set.}
Given a set of events $X$, we denote by $\OpenAcquires(X)$ the set of lock-acquire events $\Acquire$ such that $\Match{\Trace}{\Acquire}\not \in X$.
Given two events $\Event_1, \Event_2\in \Events{\Trace}$, 
the \emph{candidate ideal set} $\CandidateSet_{\Trace}(\Event_1, \Event_2)$ of $\Event_1, \Event_2$ is the smallest set of trace ideals of $\Trace$ such that the following hold.
\begin{compactenum}
\item\label{item:candidate_set1} 
$\Past_{\Trace}(\{\Event_1, \Event_2\})\in \CandidateSet_{\Trace}(\Event_1, \Event_2)$.
\item\label{item:candidate_set2} 
Let $Y\in \CandidateSet_{\Trace}(\Event_1, \Event_2)$,
$\Acquire\in\OpenAcquires(Y)$, $\Release = \Match{\Trace}{\Acquire}$, 
and $Y' = \Past_{\Trace}(Y\cup\{\Release\})\cup \{\Release\}$.
If $\Event_1,\Event_2\not \in Y'$, then
$Y'\in \CandidateSet_{\Trace}(\Event_1, \Event_2)$.
\end{compactenum}
In light of \cref{rem:witness_trace}, we will decide whether $(\Event_1, \Event_2)$ is a predictable data race by deciding the realizability of ideals in the candidate ideal set. 
\cref{item:candidate_set2}  states that, as long as there is some ideal $Y$ in the candidate ideal set such that $Y$ leaves some critical section open, we construct another ideal $Y'\supset Y$ by choosing one such open critical section and closing it, and add $Y'$ in the candidate set as well.
Intuitively, the open critical section of $Y$ might deem $Y$ not realizable, while closing that critical section might make $Y'$ realizable.
Clearly, if $\Event_1\in Y'$ or $\Event_2\in Y'$, then the realizability of $Y'$ does not imply a data race on $\Event_1, \Event_2$ as one of the two events is not enabled in $Y'$ (\cref{rem:witness_trace}).
As the following lemma shows, in order to decide whether $(\Event_1, \Event_2)$ is a predictable data race of $\Trace$, it suffices to test for realizability all the ideals in $\CandidateSet_{\Trace}(\Event_1, \Event_2)$.


\smallskip
\begin{restatable}{lemma}{lemcandidateset}\label{lem:candidate_set}
$(\Event_1, \Event_2)$ is a predictable data race of $\Trace$ iff there exists a realizable ideal $X\in\CandidateSet_{\Trace}(\Event_1, \Event_2)$
such that $\Event_1, \Event_2\not \in X$.
\end{restatable}

The following lemma gives an upper-bound on $|\CandidateSet_{\Trace}(\Event_1, \Event_2)|$,
i.e., on the number of ideals we need to test for realizability.

\smallskip
\begin{restatable}{lemma}{lemnumcandidateideals}\label{lem:num_candidate_ideals}
We have $|\CandidateSet_{\Trace}(\Event_1, \Event_2)|\leq \min(n, \alpha)^{k-2}$, where
$\alpha=k\cdot \NestingDepth\cdot \LockFactor$, and
$k$ is the number of threads,
$\NestingDepth$ is the lock-nesting depth, and 
$\LockFactor$ is the lock-dependence factor of $\Trace$.
\end{restatable}

\section{The General Case}\label{sec:general}

In this section we address the general case of dynamic data-race prediction.
The section is organized in two parts, which present the formal details of \cref{them:general_parameterized} and \cref{them:general_w1hard}.

\subsection{Upper Bound}\label{subsec:constant_threads}

In this section we establish \cref{them:general_parameterized}.
Recall that, by \cref{lem:candidate_set}, the problem is reducible to detecting a realizable rf-poset in the candidate ideal set of the two events that are tested for a data-race.
Rf-poset realizability is known to be NP-complete~\cite{Gibbons97},
and solvable in polynomial time when the number of threads is bounded~\cite{Abdulla19}.
Here we establish more precise upper-bounds, based on the number of threads. In particular, we show the following.

\smallskip
\begin{restatable}{lemma}{lemoposetrealizability}\label{lem:oposet_realizability}
Rf-poset realizability can be solved in $O(k\cdot n^{k})$ time for an rf-poset of size $n$ and $k$ threads.
\end{restatable}

\Paragraph{Frontiers and extensions.}
Let $\OPoset=(X,P,\Observation)$ be an rf-poset, and
consider an ideal $Y$ of $\OPoset$.
The \emph{frontier} of $Y$, denoted $\Frontier{\OPoset}{Y}$,  is the set of pairs $(\Write, \Read)\in \ReadPairs{\OPoset}$ such that $\Write \in Y$ and $\Read\not \in Y$.
An event $\Event$ executable in $Y$ is said to \emph{extend} $Y$ if for every triplet $(\Write, \Read, \Event)\in \ReadTriplets{\OPoset}$, we have  $(\Write, \Read)\not\in \Frontier{\OPoset}{Y}$.
In this case, we say that $Y \cup \set{e}$ is an \emph{extension} of $Y$ via $\Event$.

\Paragraph{Ideal graphs and canonical traces.}
Let $\OPoset=(X,P,\Observation)$ be an rf-poset.
The \emph{ideal graph} of $\OPoset$, denoted $G_{\OPoset}=(V_{\OPoset}, E_{\OPoset})$ is a directed graph defined as follows.
\begin{compactenum}
\item $V_{\OPoset}$ is the set of ideals of $\OPoset$.
\item We have $(Y_1, Y_2)\in E_{\OPoset}$ iff $Y_2$ is an extension of $Y_1$.
\end{compactenum}
The \emph{ideal tree} of $\OPoset$, denoted $\Tree_{\OPoset}=(\TreeNodes_{\OPoset}, \TreeEdges_{\OPoset})$ is a (arbitrary) spanning tree of $G_{\OPoset}$ when restricted to nodes reachable from $\emptyset$.
We let $\emptyset$ be the root of $\Tree_{\OPoset}$.
Given an ideal $Y\in \TreeNodes_{\OPoset}$, we define the \emph{canonical trace} $\Trace_{Y}$ of $Y$ inductively, as follows.
If $Y=\emptyset$ then $\Trace_Y=\epsilon$. 
Otherwise, $Y$ has a parent $Y'$ in $\Tree_{\OPoset}$ such that $Y=Y'\cup \{ \Event \}$ for some event $\Event\in X$.
We define $\Trace_{Y}=\Trace_{Y'} \circ \Sequence{\Event}$.
\cref{lem:oposet_realizability} relies on the following lemmas.
We refer to \cref{subsec:proofs_constant_threads} for the proofs.

\smallskip
\begin{restatable}{lemma}{lemidealtree}\label{lem:ideal_tree}
We have $X\in \TreeNodes_{\OPoset}$ iff $\OPoset$ is realizable.
\end{restatable}

\smallskip
\begin{restatable}{lemma}{lemidealgraphsize}\label{lem:idealgraph_size}
The ideal graph $G_{\OPoset}$ has $O(n^k)$ nodes.
\end{restatable}



\begin{proof}[Proof of \cref{them:general_parameterized}.]
Consider a trace $\Trace$ and two conflicting events $\Event_1, \Event_2\in \WritesReads{\Trace}$.
By \cref{lem:candidate_set}, to decide whether $(\Event_1, \Event_2)$ is a predictable data race of $\Trace$, it suffices to iterate over all feasible trace ideals $X$ in the candidate ideal set $\CandidateSet_{\Trace}(\Event_1, \Event_2)$, and test whether $X$ is realizable.
By \cref{lem:num_candidate_ideals}, we have $|\CandidateSet_{\Trace}(\Event_1, \Event_2)|=O(\alpha)$, where $\alpha=\min(n, k\cdot \NestingDepth\cdot \LockFactor)^{k-2}$.
Finally, due to \cref{lem:oposet_realizability}, the realizability of every such ideal can be performed in $O(k\cdot n^{k})=O(\beta)$ time.
\end{proof}


\subsection{Hardness of Data Race Prediction}\label{subsec:w1_hardness}

Here we establish that the problem of dynamic data-race prediction is W[1]-hard when parameterized by the number of threads $k$.
Our proof is established in two steps.
In the first step, we show the following lemma.
\smallskip
\begin{restatable}{lemma}{lemoposetwhardness}\label{lem:oposet_w1hardness}
Rf-poset realizability parameterized by the number of threads $k$ is W[1]-hard. 
\end{restatable}

Rf-poset realizability is known to be NP-hard~\cite[Theorem~4.1]{Gibbons97},
and \cref{lem:oposet_w1hardness} strengthens that result by showing that the problem is even unlikely to be FPT.
In the second step, we show how the class of W[1]-hard instances of in \cref{lem:oposet_w1hardness} can be reduced to dynamic data-race prediction.

\Paragraph{Hardness of rf-poset realizability.}
Our reduction is from the $\INDSET(c)$ problem, which takes as input an undirected graph $G=(V,E)$ and asks whether $G$ has an independent set of size $c$.
$\INDSET(c)$ parameterized by $c$  is one of the canonical W[1]-hard problems~\cite{Rodney99}.

\begin{figure*}
\newcommand{\xdisposition}{0}
\newcommand{\ydisposition}{0}
\newcommand{\xtstep}{0.75}
\newcommand{\ytstep}{0.7}
\newcommand{\ybias}{-0.3 }
\newcommand{\xstep}{2.5}
\newcommand{\ystep}{-0.475}
\newcommand{\xtscale}{0.8}

\def \numevents{16.5}

\newcommand{\eventA}[4]{
\node[event, draw=black, fill=white] (A#1) at (#1*\xstep, #2*\ystep) {\footnotesize $#2(x_{#3})$};
}

\scalebox{0.95}{
\begin{tikzpicture}[thick,
pre/.style={<-,shorten >= 2pt, shorten <=2pt, very thick},
post/.style={->,shorten >= 3pt, shorten <=3pt,   thick},
seqtrace/.style={->, line width=2},
und/.style={very thick, draw=gray},
event/.style={rectangle, minimum height=0.8mm, minimum width=15mm,  line width=1pt, inner sep=0.5,},
virt/.style={circle,draw=black!50,fill=black!20, opacity=0}]
\footnotesize

\node[] (G) at (-1.4*\xstep,6*\ystep) {\large  $G$};
\node[draw=black, circle, very thick] (1) at (-1.4*\xstep,9*\ystep) {\normalsize $1$};
\node[draw=black, circle, very thick] (2) at (-1.4*\xstep,7.5*\ystep) {\normalsize $2$};
\node[draw=black, circle, very thick] (3) at (-1.4*\xstep,10.5*\ystep) {\normalsize $3$};

\draw[-, very thick] (1) to (2);
\draw[-, very thick] (1) to (3);

\node[] (S11) at (0*\xstep,0) {\normalsize $\SeqTrace_1$};
\node[] (S12) at (0*\xstep,\numevents * \ystep) {};
\node[] (S21) at (2*\xstep,0) {\normalsize $\SeqTrace_2$};
\node[] (S22) at (2*\xstep,\numevents * \ystep) {};
\node[] (S31) at (1*\xstep,0) {\normalsize $\SeqTrace_3$};
\node[] (S32) at (1*\xstep,\numevents * \ystep) {};
\node[] (S41) at (3*\xstep,0) {\normalsize $\SeqTrace_4$};
\node[] (S42) at (3*\xstep,\numevents * \ystep) {};
\node[] (S51) at (4*\xstep,0) {\normalsize $\SeqTrace_5$};
\node[] (S52) at (4*\xstep,\numevents * \ystep) {};
\node[] (S61) at (5*\xstep,0) {\normalsize $\SeqTrace_6$};
\node[] (S62) at (5*\xstep,\numevents * \ystep) {};

\draw[seqtrace] (S11) to (S12);
\draw[seqtrace] (S21) to (S22);
\draw[seqtrace] (S31) to (S32);
\draw[seqtrace] (S41) to (S42);
\draw[seqtrace] (S51) to (S52);
\draw[seqtrace] (S61) to (S62);

\node[event, draw=black, fill=white] (11) at (0*\xstep, 1*\ystep + 0*\ybias) {$\Acquire_1(\ell_{\{1,2\}})$};
\node[event, draw=black, fill=white] (12) at (0*\xstep, 2*\ystep + 0*\ybias) {$\Acquire_1(\ell_{\{1,3\}})$};
\node[event, draw=black, fill=white] (13) at (0*\xstep, 3*\ystep + 0*\ybias) {$\Write(s_1)$};
\node[event, draw=black, fill=white] (14) at (0*\xstep, 4*\ystep + 0*\ybias) {$\Read(z_1^1)$};
\node[event, draw=black, fill=white] (15) at (0*\xstep, 5*\ystep + 0*\ybias) {$\Release_1(\ell_{\{1,3\}})$};
\node[event, draw=black, fill=white] (16) at (0*\xstep, 6*\ystep + 0*\ybias) {$\Release_1(\ell_{\{1,3\}})$};
\node[event, draw=black, fill=white] (17) at (0*\xstep, 7*\ystep + 1*\ybias) {$\Acquire_1(\ell_{\{1,2\}})$};
\node[event, draw=black, fill=white, line width=1.75] (18) at (0*\xstep, 8*\ystep + 1*\ybias) {$\Write(y_1^2)$};
\node[event, draw=black, fill=white] (19) at (0*\xstep, 9*\ystep + 1*\ybias) {$\Read(z_1^2)$};
\node[event, draw=black, fill=white] (110) at (0*\xstep, 10*\ystep + 1*\ybias) {$\Release_1(\ell_{\{1,2\}})$};
\node[event, draw=black, fill=white] (111) at (0*\xstep, 11*\ystep + 2*\ybias) {$\Acquire_1(\ell_{\{1,3\}})$};
\node[event, draw=black, fill=white] (112) at (0*\xstep, 12*\ystep + 2*\ybias) {$\Write(y_1^3)$};
\node[event, draw=black, fill=white] (113) at (0*\xstep, 13*\ystep + 2*\ybias) {$\Read_1(x)$};
\node[event, draw=black, fill=white] (114) at (0*\xstep, 14*\ystep + 2*\ybias) {$\Release_1(\ell_{\{1,3\}})$};

\node[event, draw=black, fill=white] (21) at (2*\xstep, 1*\ystep + 0*\ybias) {$\Acquire_2(\ell_{\{1,2\}})$};
\node[event, draw=black, fill=white] (22) at (2*\xstep, 2*\ystep + 0*\ybias) {$\Acquire_2(\ell_{\{1,3\}})$};
\node[event, draw=black, fill=white] (23) at (2*\xstep, 3*\ystep + 0*\ybias) {$\Write(s_2)$};
\node[event, draw=black, fill=white] (24) at (2*\xstep, 4*\ystep + 0*\ybias) {$\Read(z_2^1)$};
\node[event, draw=black, fill=white] (25) at (2*\xstep, 5*\ystep + 0*\ybias) {$\Release_2(\ell_{\{1,3\}})$};
\node[event, draw=black, fill=white] (26) at (2*\xstep, 6*\ystep + 0*\ybias) {$\Release_2(\ell_{\{1,3\}})$};
\node[event, draw=black, fill=white] (27) at (2*\xstep, 7*\ystep + 1*\ybias) {$\Acquire_2(\ell_{\{1,2\}})$};
\node[event, draw=black, fill=white] (28) at (2*\xstep, 8*\ystep + 1*\ybias) {$\Write(y_2^2)$};
\node[event, draw=black, fill=white] (29) at (2*\xstep, 9*\ystep + 1*\ybias) {$\Read(z_2^2)$};
\node[event, draw=black, fill=white] (210) at (2*\xstep, 10*\ystep + 1*\ybias) {$\Release_2(\ell_{\{1,2\}})$};
\node[event, draw=black, fill=white] (211) at (2*\xstep, 11*\ystep + 2*\ybias) {$\Acquire_2(\ell_{\{1,3\}})$};
\node[event, draw=black, fill=white, line width=1.75] (212) at (2*\xstep, 12*\ystep + 2*\ybias) {$\Write(y_2^3)$};
\node[event, draw=black, fill=white] (213) at (2*\xstep, 13*\ystep + 2*\ybias) {$\Read_2(x)$};
\node[event, draw=black, fill=white] (214) at (2*\xstep, 14*\ystep + 2*\ybias) {$\Release_2(\ell_{\{1,3\}})$};

\node[event, draw=black, fill=white] (31) at (1*\xstep, 4*\ystep + 1*\ybias) {$\Acquire^1(\ell_1)$};
\node[event, draw=black, fill=white] (32) at (1*\xstep, 5*\ystep + 1*\ybias) {$\Write(z_1^1)$};
\node[event, draw=black, fill=white] (33) at (1*\xstep, 6*\ystep + 1*\ybias) {$\Read(y_1^2)$};
\node[event, draw=black, fill=white, line width=1.75] (34) at (1*\xstep, 7*\ystep + 1*\ybias) {$\Release^1(\ell_1)$};
\node[event, draw=black, fill=white] (35) at (1*\xstep, 10*\ystep + 1*\ybias) {$\Acquire^2(\ell_1)$};
\node[event, draw=black, fill=white] (36) at (1*\xstep, 11*\ystep + 1*\ybias) {$\Write(z_1^2)$};
\node[event, draw=black, fill=white] (37) at (1*\xstep, 12*\ystep + 1*\ybias) {$\Read(y_1^3)$};
\node[event, draw=black, fill=white] (38) at (1*\xstep, 13*\ystep + 1*\ybias) {$\Release^2(\ell_1)$};

\node[event, draw=black, fill=white] (41) at (3*\xstep, 4*\ystep + 1*\ybias) {$\Acquire^1(\ell_2)$};
\node[event, draw=black, fill=white] (42) at (3*\xstep, 5*\ystep + 1*\ybias) {$\Write(z_2^1)$};
\node[event, draw=black, fill=white] (43) at (3*\xstep, 6*\ystep + 1*\ybias) {$\Read(y_2^2)$};
\node[event, draw=black, fill=white] (44) at (3*\xstep, 7*\ystep + 1*\ybias) {$\Release^2(\ell_2)$};
\node[event, draw=black, fill=white] (45) at (3*\xstep, 10*\ystep + 1*\ybias) {$\Acquire^2(\ell_2)$};
\node[event, draw=black, fill=white] (46) at (3*\xstep, 11*\ystep + 1*\ybias) {$\Write(z_2^2)$};
\node[event, draw=black, fill=white] (47) at (3*\xstep, 12*\ystep + 1*\ybias) {$\Read(y_2^3)$};
\node[event, draw=black, fill=white, line width=1.75] (48) at (3*\xstep, 13*\ystep + 1*\ybias) {$\Release^2(\ell_2)$};

\node[event, draw=black, fill=white] (51) at (4*\xstep, 7*\ystep + -0.2*\ybias) {$\Write(x)$};

\node[event, draw=black, fill=white] (61) at (5*\xstep, 4*\ystep + 0.25*\ybias) {$\Read(s_1)$};
\node[event, draw=black, fill=white] (62) at (5*\xstep, 5*\ystep + 0.25*\ybias) {$\Read(s_2)$};
\node[event, draw=black, fill=white] (63) at (5*\xstep, 6*\ystep +  0.25*\ybias) {$\Acquire(\ell_1)$};
\node[event, draw=black, fill=white] (64) at (5*\xstep, 7*\ystep +  0.25*\ybias) {$\Acquire(\ell_2)$};
\node[event, draw=black, fill=white] (65) at (5*\xstep, 8*\ystep +  0.25*\ybias) {$\Read(x)$};
\node[event, draw=black, fill=white] (66) at (5*\xstep, 9*\ystep +  0.25*\ybias) {$\Release(\ell_2)$};
\node[event, draw=black, fill=white] (67) at (5*\xstep, 10*\ystep +  0.25*\ybias) {$\Release(\ell_1)$};

\draw[post, out=180, in=-0, looseness=-0.4] (32) to (14);
\draw[post, out=0, in=180, looseness=-0.4] (18) to (33);
\draw[post, out=180, in=-0, looseness=-0.4] (36) to (19);
\draw[post, out=0, in=180, looseness=-0.4] (112) to (37);

\draw[post, out=180, in=-0, looseness=-0.4] (42) to (24);
\draw[post, out=0, in=180, looseness=-0.4] (28) to (43);
\draw[post, out=180, in=-0, looseness=-0.4] (46) to (29);
\draw[post, out=0, in=180, looseness=-0.4] (212) to (47);

\draw[post] (13) to (61);
\draw[post, out=0, in=180, looseness=-0.4] (23) to (62);

\draw[post, out=0, in=180, looseness=-0.4] (51) to (65);
\draw[post, out=180, in=0, looseness=0.75] (65) to ($ (48) + (0,2.05*\ystep) $) to ($ (214) + (-0.25*\xstep,-0.25) $) to (113);
\draw[post, out=180, in=0, looseness=0.9] (65) to ($ (48) + (0,-0.5) $)  to  (213);


\end{tikzpicture}
}
\caption{
Illustration of the o-poset $\OPoset_G$ given a graph $G$ and independent-set size $c=2$.
Edges represent orderings in $P$.
}
\label{fig:w1hardness}
\end{figure*}
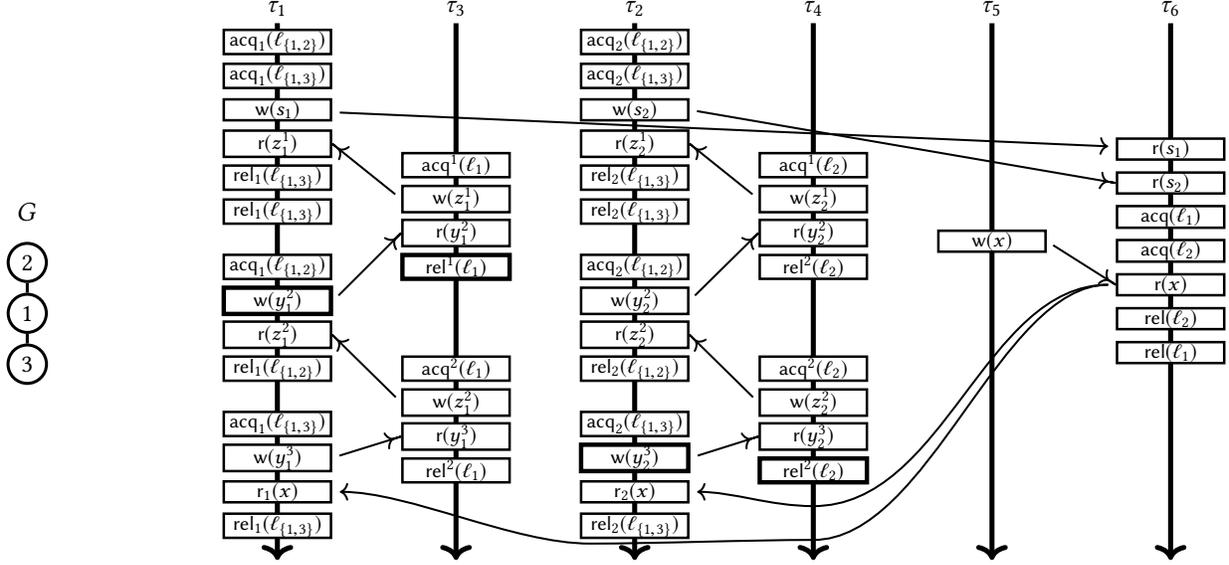

Given an input $G=(V,E)$ of $\INDSET(c)$ with $n=|V|$,
we construct an rf-poset $\OPoset_G=(X,P,\Observation)$ of size $O(c\cdot n)$ and $O(c)$ threads such that $\OPoset_G$ is realizable iff $G$ has an independent set of size  $c$.
We assume wlog that every node in $G$ has at least one neighbor,
otherwise, we can remove all $s$ such nodes and solve the problem for parameter $c'=c-s$.
The rf-poset $\OPoset_G$ consists of $k=2\cdot c + 2$ total orders $(X_i, \SeqTrace_i)$.
\cref{fig:w1hardness} provides an illustration.
In high level, for each $i\in [c]$, $\SeqTrace_i$ and $\SeqTrace_{c+i}$ are used to encode the $i$-th copy of $G$,
whereas the last two total orders are auxiliary.
Superscripts on the events and/or their variables refer to the node of $G$ that is encoded by those events.
Below we describe the events and certain orderings between them.
The partial order $P$ is the transitive closure of these orderings.
\begin{compactenum}
\item For $i=2\cdot c + 1$, $\SeqTrace_i$ consists of a single event
$
\SeqTrace_i= \Write(x)
$.
\item For $i=2\cdot c + 2$, we have
$\SeqTrace_i=\SeqSubTrace \Concat\SeqSubTraceRefl$, where
\begin{align*}
\SeqSubTrace &= \Read(s_1),\dots, \Read(s_c), \Acquire(\ell_1),\dots, \Acquire(\ell_c)
\qquad \text{and}\\
\SeqSubTraceRefl &=\Read(x), \Release(\ell_c),\dots, \Release(\ell_1)\ .
\end{align*}
\item For each $i\in [c]$, we have $\SeqTrace_i=\SeqTrace_i^1\Concat \SeqTrace_i^2\Concat \dots \Concat \SeqTrace_i^n$,
where each $\SeqTrace_i^j$ encodes node $j$ of $G$ and is defined as follows.
Let $\ov{\SeqTrace}^j_i=\SeqSubTrace^j_i \Concat \SeqSubTraceRefl^j_i$, where
\begin{align*}
\SeqSubTrace^j_i&=\Acquire_i(\ell_{\{ j, l_1 \}}), \dots , \Acquire_i(\ell_{\{ j, l_m \}})
\qquad\text{and}\\
\SeqSubTraceRefl^j_i&=\Release_i(\ell_{\{ j, l_m \}}), \dots , \Release_i(\ell_{\{ j, l_1 \}})
\end{align*}
where $l_1,\dots, l_m$ are the neighbors of $j$ in $G$.
For each $j\in [n]\setminus\{1, n \}$, the sequence $\SeqTrace_i^j$ is identical to $\ov{\SeqTrace}^j_i$, with the addition that the innermost critical section (i.e., between $\Acquire_i(\ell_{\{ j, l_m \}})$ and $\Release_i(\ell_{\{ j, l_m \}})$) contains
the sequence $\Write(y_{i}^{j}), \Read(z_{i}^{j})$.
The sequence $\SeqTrace_i^1$ is defined similarly, except that the innermost critical section contains the sequence $\Write(s_i), \Read(z_{i}^{1})$.
Finally, the sequence $\SeqTrace_i^n$ is defined similarly, except that the innermost critical section contains the sequence $\Write(y_{i}^{n}), \Read_i(x)$.
\item For each $i\in  [c]$, we have $\SeqTrace_{c+i}=\SeqTrace_{c+i}^1\Concat \SeqTrace_{c+i}^2\Concat \dots \Concat \SeqTrace_{c+i}^{n-1}$,
where $\SeqTrace_{c+i}^j = \Acquire^j(\ell_i), \Write(z_i^j), \Read(y_i^{j+1}),  \Release^j(\ell_i)$.
\end{compactenum}
Note that every memory location is written exactly once, hence the reads-from function $\Observation$ is defined implicitly.
In addition, for every read event $\Read$, we have $\Observation(\Read)<_{P}\Read$,
as well as $\Read(x)<_{P}\Read_i(x)$ for each $i \in [c]$.

\SubParagraph{Correctness.}
We now sketch the correctness of the construction, while we refer to \cref{subsec:proofs_w1_hardness} for the proof.
Assume that $\OPoset$ is realizable by a witness $\Trace$.
We say that $\Read(x)$ \emph{separates} a critical section in $\Trace$ if the lock-acquire (resp., lock-release) event of that critical section appears before (resp., after) $\Read(x)$ in $\Trace$.
The construction guarantees that, for each $i\in [c]$, $\Read(x)$ separates the critical sections of $\SeqTrace_i$ that encode some node $l_i$ of $G$.
By construction, these critical sections are on locks  $\ell_{\{l_i, v\}}$, where $v$ ranges over the neighbors of $l_i$ in $G$.
Hence, for any $i'\neq i$, the node $l_{i'}$ cannot be a neighbor of $l_i$, as this would imply that both critical sections on lock $\ell_{\{l_i, l_{i
'}\}}$ are opened before $\Read(x)$ and closed after $\Read(x)$ in $\Trace$, which clearly violates lock semantics.
Thus, an independent set $A=\{l_1,\dots, l_c\}$ of $G$ is formed by taking each $l_i$ to be the node of $G$, the critical sections of which belong to thread $\SeqTrace_i$ and are separated by $\Read(x)$ in $\Trace$.
On the other hand, if $G$ has an independent set $A=\{l_1,\dots, l_c\}$, a witness $\Trace$ that realizes $\OPoset$ can be constructed by separating the critical sections of the node $l_i$ in $\SeqTrace_i$, for each $i\in [c]$.

\Paragraph{Hardness of dynamic data-race prediction.}
Finally, we turn our attention to the hardness of dynamic data-race prediction.
Consider the $\INDSET(c)$ problem on a graph $G$ and the associated rf-poset $\OPoset_G=(X,P,\Observation)$ defined above.
We construct a trace $\Trace$ with $\Events{\Trace}=X$ and $\Observation_{\Trace}=\Observation$ such that
$(\Write(x), \Read(x))$ is a predictable data-race of $\Trace$ iff $\OPoset_G$ is realizable.
In particular, $\Trace$ consists of $2\cdot c + 2$ threads $\Process_i$, one for each total order $\SeqTrace_i$ of $\OPoset_G$.
We obtain $\Trace$ as
\[
\Trace=\SeqTrace_{2\cdot c + 1}\Concat \Trace_1\Concat \dots \Concat \Trace_c \Concat \SeqTrace_{2\cdot c + 2}\ ,
\]
where each $\Trace_i$ is an appropriate interleaving of the total orders $\SeqTrace_i$ and $\SeqTrace_{c+i}$ that respects the reads-from function $\Observation$.
In \cref{subsec:proofs_w1_hardness}, we conclude the proof of \cref{them:general_w1hard} by showing that $G$ has an independent set of size $c$ iff $(\Write(x), \Read(x))$ is a predictable data-race of $\Trace$.

\smallskip
\begin{remark}\label{rem:eth}
It is known that $\INDSET(c)$ cannot be solved in $f(c)\cdot n^{o(c)}$ time under ETH~\cite{Chen06}.
As our reduction to rf-poset realizability and dynamic data-race prediction uses $k=O(c)$ threads,
each of these problems does not have a $f(k)\cdot n^{o(k)}$-time algorithm under ETH.
\end{remark}

\section{Tree Communication Topologies}\label{sec:trees}

In this section we focus on the case where the input trace $\Trace$ constitutes a tree communication topology.
The section is organized in two parts, which present the formal details of \cref{them:tree_topologies_ub} and \cref{them:tree_topologies_lb}.


\subsection{An Efficient Algorithm for Tree Topologies}\label{subsec:tree_topologies_ub}

In this section we present the formal details of \cref{them:tree_topologies_ub}.
Recall that the \cref{them:general_parameterized} states an $O(\alpha\cdot \beta)$ upper-bound for dynamic data-race prediction, where $\beta$ is the complexity of deciding rf-poset realizability, and $\alpha$ is an upper-bound on the number of candidate ideals whose realizability we need to check.
For tree communication topologies, we obtain \cref{them:tree_topologies_ub}:
(i)~we show an improved upper-bound $\beta$ on the complexity of the realizability of trace ideals over tree topologies, and
(ii)~we show that it suffices to check the realizability of a single trace ideal (i.e., $\alpha=1$).
We start with point (i), and then proceed with (ii).

\Paragraph{Tree-inducible rf-posets.}
Let $(X,P)$ be a poset where $X\subseteq \Events{\Trace}$.
We call $(X,P)$ \emph{tree-inducible} if $X$ can be partitioned into $k$ sets
$\{X_i\}_{1\leq i\leq k}$ such that  the following conditions hold.
\begin{compactenum}
\item\label{item:tree_ind1} The graph $\Tree=([k], \setpred{ (i,j)}{\Confl{X_i}{X_j}})$ is a tree.
\item\label{item:tree_ind2} $P\Project X_{i}$ is a total order for each $i\in [k]$.
\item\label{item:tree_ind3} For every node $\ell\in[k]$ such that $\ell$ is 
an internal node in $\Tree$ and for every two connected components $C_1$, $C_2$ of $\Tree$ that are created after removing $\ell$ from $\Tree$, we have the following.
Consider two nodes $i\in C_1$ and $j\in C_2$ and
two events $\Event_1\in X_i$ and $\Event_2\in X_j$ such that $\Event_1<_{P} \Event_2$,
there exists some event $\Event\in X_{\ell}$ such that $\Event_1<_{P}\Event<_{P}\Event_2$.
\end{compactenum}
We call an rf-poset $(X,P,\Observation)$ tree-inducible if $(X,P)$ is tree-inducible.
The insight is that traces from tree communication topologies yield tree-inducible trace ideals.
Our motivation behind tree inducibility comes from the following lemma.

\smallskip
\begin{restatable}{lemma}{lemoposetrealizabilitytrees}\label{lem:oposet_realizability_trees}
Rf-poset realizability of tree-inducible rf-posets can be solved in $O(k^2\cdot d\cdot n^2\cdot \log n)$ time, for an rf-poset of size $n$, $k$ threads and $d$ variables.
\end{restatable}

The proof of \cref{lem:oposet_realizability_trees} is in two steps.
Recall the definition of closed rf-posets from \cref{subsec:po}.
First, we show that a tree-inducible, closed rf-poset is realizable (\cref{lem:tree_topologies}).
Second, we show that the closure of a tree-inducible rf-poset is also tree-inducible (\cref{lem:closure_tree}).

\smallskip
\begin{restatable}{lemma}{lemtreetopologies}\label{lem:tree_topologies}
Every closed, tree-inducible rf-poset is realizable.
\end{restatable}

Indeed, consider a tree-inducible, closed rf-poset $\OPoset=(X,P,\Observation)$.
The witness $\Trace$ realizing $\OPoset$ is obtained in two steps.
\begin{compactenum}
\item We construct a poset $(X,Q)$ with $Q\Refines P$ as follows. Initially, we let $Q$ be identical to $P$.
Let $\OPoset$ be tree-inducible to a tree $\Tree=([k], \setpred{ (i,j)}{\Confl{X_i}{X_j}})$.
We traverse $\Tree$ top-down, and for every node $i$ and child $j$ of $i$, for every two events $\Event_1\in X_i$ and  $\Event_2\in X_j$
with $\Confl{\Event_1}{\Event_2}$ and $\Event_2\not <_{P} \Event_1$, we order $\Event_1<_{Q}\Event_2$.
Finally, we transitively close $Q$.
\item We construct $\Trace$ by linearizing $(X,Q)$ arbitrarily.
\end{compactenum}

In \cref{subsec:proofs_tree_topologies_ub} we show that $\Trace$ is well-defined and realizes $\OPoset$.
The next lemma shows that tree-inducibility is preserved under taking closures (if the closure exists).

\smallskip
\begin{restatable}{lemma}{lemclosuretree}\label{lem:closure_tree}
Consider an rf-poset $\OPoset=(X,P,\Observation)$ and let $\OPosetQ=(X,Q,\Observation)$ be the closure of $\OPoset$.
If $\OPoset$ is tree-inducible then $\OPosetQ$ is also tree-inducible.
\end{restatable}

Since, by \cref{rem:closure}, an rf-poset is realizable only if its closure exists and is realizable,
\cref{lem:tree_topologies} and \cref{lem:closure_tree} allow to decide the realizability of an rf-poset by computing its closure.
The complexity of the algorithm comes from the complexity of deciding whether the closure exists.
We refer to \cref{subsec:proofs_tree_topologies_ub} for the full proof of \cref{lem:oposet_realizability_trees}.

Recall that \cref{lem:num_candidate_ideals} provides an upper-bound on the number of trace ideals of $\Trace$ that we need to examine for realizability in order to decide data-race prediction.
We now proceed with point (ii) towards \cref{them:tree_topologies_ub}, i.e., we show that for tree communication topologies, a single ideal suffices.
Our proof is based on the notion of lock causal cones below.

\Paragraph{Lock causal cones.}
Consider a trace $\Trace$ that defines a tree communication topology $G_{\Trace}=(V_{\Trace},E_{\Trace})$.
Given an event $\Event\in \Events{\Trace}$ the \emph{lock causal cone} $\LPast_{\Trace}(\Event)$ of $\Event$ is the set $X$ defined by the following process. 
Consider that $G_{\Trace}$ is rooted in $\Proc{\Event}$.
\begin{compactenum}
\item\label{item:lpast1} Initially $X$ contains all predecessors of $\Event$ in $(\Events{\Trace},\TO)$.
We perform a top-down traversal of $G_{\Trace}$, and consider a current thread $\Process_1$ visited by the traversal.
\item\label{item:lpast2} 
Let $\Process_2$ be the parent of $\Process_1$ in the traversal tree, and $\Event_2$ be the unique maximal event in $(X\Project \Process_2, \TO)$,
i.e., $\Event_2$ is the last event of thread $\Process_2$ that appears in $X$.
We insert in $X$ all events $\Event_1\in \Events{\Trace}\Project \Process_1$ such that $\Event_1<_{\TOO}\Event_2$.
\item\label{item:lpast3} While there exists some lock-acquire event $\Acquire_1\in X\Project \Process_1$ and there exists another lock-acquire event $\Acquire_2\in \OpenAcquires(X)$ with $\Confl{\Acquire_1}{\Acquire_2}$ and $\Proc{\Acquire_2}=\Process_2$,
we insert in $X$ all predecessors of $\Release_1$ in $(\Events{\Trace},\TO)$ (including $\Release_1$),
where $\Release_1=\Match{\Trace}{\Acquire_1}$.
\end{compactenum}
Observe that, by construction, $\LPast_{\Trace}(\Event)$ is a lock-feasible trace ideal of $\Trace$.
In addition, for any two events $\Event_1, \Event_2\in \Events{\Trace}$,
the set $\LPast_{\Trace}(\Event_1)\cup \LPast_{\Trace}(\Event_2)$ is an ideal of $\Trace$,
though not necessarily lock-feasible.
Our motivation behind lock causal cones comes from the following lemma.
Intuitively, we can decide a predictable data-race by deciding the realizability of the ideal that is the union of the two lock-causal cones.

\smallskip
\begin{restatable}{lemma}{lemtreeideal}\label{lem:tree_ideal}
Let $X=\LPast_{\Trace}(\Event_1)\cup \LPast_{\Trace}(\Event_2)$.
We have that $(\Event_1, \Event_2)$ is a predictable data race of $\Trace$ iff
(i)~$\{ \Event_1, \Event_2 \}\cap X=\emptyset$, and
(ii) ~$X$ is a realizable trace ideal of $\Trace$.
\end{restatable}

The $(\Leftarrow)$ direction of the lemma is straightforward.
We refer to \cref{subsec:proofs_tree_topologies_ub} for the $(\Rightarrow)$ direction.
Finally, \cref{them:tree_topologies_ub} follows immediately from \cref{lem:oposet_realizability_trees} and \cref{lem:tree_ideal}.

\begin{proof}[Proof of \cref{them:tree_topologies_ub}]
By \cref{lem:tree_ideal}, we have that $(\Event_1, \Event_2)$ is a predictable data race of $\Trace$ iff
$\{ \Event_1, \Event_2 \}\cap X= \emptyset $ and $X$ is realizable.
By \cref{lem:oposet_realizability_trees}, deciding the realizability of $X$ is done in $O(k^2\cdot d\cdot n^2\cdot \log n)$ time.
The desired result follows.
\end{proof}


\subsection{A Lower Bound for Two Threads}\label{subsec:tree_topologies_lb}

In this section we prove a conditional quadratic lower bound for dynamic data-race prediction for two threads.
Our proof is via a reduction from the Orthogonal Vectors problem.
To make it conceptually simpler, we present our reduction in two steps.
First, we show a fine-grained reduction from Orthogonal Vectors to the realizability of an rf-poset with $2$ threads and $7$ variables.
Afterwards, we show how the realizability problem for the rf-posets of the first step can be reduced to the decision problem of dynamic data race prediction with $2$ threads, $9$ variables and $1$ lock.

\Paragraph{The Orthogonal Vectors problem ($\OrthVec$).}
An instance of Orthogonal Vectors consists of two sets $A,B$, where each set contains $n/2$ binary vectors in $D$ dimensions.
The task is to determine whether there exists a pair of vectors $(a,b)\in A\times B$ that is orthogonal,
i.e., for all $i\in[D]$ we have $a[i]\cdot b[i]=0$.
There exist algorithms that solve the problem in $O(n^2\cdot D)$ and $O(2^{D}\cdot n)$ time, simply by computing the inner product of each pair $(a,b)\in A\times B$ and following a classic Four-Russians technique, respectively.
It is conjectured that there is no truly sub-quadratic algorithm for $\OrthVec$~\cite{Bringmann19}.

\smallskip
\begin{conjecture}[Orthogonal Vectors]\label{conj:orth_vec}
There is no algorithm for $\OrthVec$ that operates in $O(n^{2-\epsilon}\cdot D^{O(1)})$ time, for any $\epsilon>0$.
\end{conjecture}
It is also known that SETH implies the $\OrthVec$ conjecture~\cite{Williams05}.
We first relate $\OrthVec$ with rf-poset realizability.

\smallskip
\begin{restatable}{lemma}{lemclosureov}\label{lem:closure_ov}
Rf-poset realizability for an rf-poset with $2$ threads and $7$ variables has no $O(n^{2-\epsilon})$-time algorithm for any $\epsilon >0$, under the Orthogonal Vectors conjecture.
\end{restatable}

\Paragraph{Reduction from $\OrthVec$ to rf-poset realizability.}
For a fine-grained reduction from $\OrthVec$ to rf-poset realizability, consider an $\OrthVec$ instance $(A,B)$,
where $A=(a_j)_{1\leq j \leq n/2}$, $B=(b_l)_{1\leq l \leq n/2}$, and each $a_j, b_l\in \{0,1 \}^D$.
We will construct an rf-poset $\OPoset=(X,P, \Observation)$ 
with $2$ threads and $7$ variables such that the closure of $\OPoset$ exists 
iff there exists a pair of orthogonal vectors $(a,b)\in A\times B$.
Since $2$ threads define a tree-inducible rf-poset, \cref{rem:closure} and \cref{lem:tree_topologies} imply that $\OPoset$ is realizable iff $\OrthVec$ has a positive answer.
The set $X$ consists of two disjoint sets $X_A, X_B$, so that each is totally ordered in $P$.
For ease of presentation, we denote by $\SeqTrace_A, \SeqTrace_B$ the linear orders $(X_A, P \Project X_A)$ and $(X_B, P\Project X_B)$, respectively.
To develop some insight, we start with a high-level view of the construction, and then proceed with the details.

\smallskip\noindent{\em Overview of the construction.}
The linear orders $\SeqTrace_A$ and $\SeqTrace_B$ encode the vectors of $A$ and $B$, respectively.
Each of $\SeqTrace_A$ and $\SeqTrace_B$ consists of $n/2$ segments, so that the $i$-th (resp. $(n/2-i + 1)$-th) segment of $\SeqTrace_A$ (resp. $\SeqTrace_B$) encodes the contents of the $i$-th vector of $A$ (resp., $B$).
The two total orders are constructed with a closure computation in mind, which inserts event orderings in $P$ one-by-one.
In high-level, an ordering $\Event_1< \Event_2$ encodes the test of whether the bits in a specific coordinate $i$ of two vectors $a_j\in A$ and $b_l\in B$ have product $0$.
If yes, and moreover, $i<D$, then the closure conditions enforce a new ordering $\Event'_1< \Event'_2$, which encodes the test of the bits in coordinate $i+1$.
Otherwise $i=D$, and the closure has been computed and an orthogonal pair has been found.
On the other hand, if the bits in coordinate $i$ have product $1$, the two current vectors are not orthogonal, and the closure conditions enforce a new ordering $\Event''_1< \Event''_2$, which encodes the test of the first coordinate of the next pair of vectors.
The above is achieved using $7$ variables $\{x_i\}_{i\in [7]}$.

\smallskip\noindent{\em Formal construction.}
We now present the formal details of the construction (illustrated in \cref{fig:ov}).
The construction creates various events which have the form $\Event^{a_j}$ 
and $\Event^{b_l}$ when they are used at the vector level,
and have the form $\Event_i^{a_j}$ and $\Event_i^{b_l}$, 
where $i\in [D]$, when they are used at the coordinate level.
As a general rule, for each $j,l\in [n/2-1]$, we have 
$\Event^{a_j}<_{\SeqTrace_A}\Event^{a_{j+1}}$ and 
$\Event^{b_{l+1}}<_{\SeqTrace_B}\Event^{b_l}$,
both for events at the vector and at the coordinate level.
At the coordinate level, we also have $\Event^{a_j}_{i+1}<_{\SeqTrace_A}\Event^{a_{j}}_{i}$ and $\Event^{b_l}_{i+1}<_{\SeqTrace_B}\Event^{b_{l}}_{i}$.
For succinctness, we often write $\Event_1,\Event_2<\Event_3$ to denote $\Event_1<\Event_3$ and $\Event_2<\Event_3$.
We next describe the events and orderings between them.
The partial order $P$ is the transitive closure of these orderings.

\begin{figure}
\newcommand{\xdisposition}{0}
\newcommand{\ydisposition}{0}
\newcommand{\xtstep}{0.75}
\newcommand{\ytstep}{0.5}
\newcommand{\xstep}{2.1}
\newcommand{\ystep}{-0.5}
\newcommand{\xtscale}{0.8}
\newcommand{\ybias}{0.1}

\def \mycolone {blue}
\def \mycoltwo {green!70}
\def \mycolthree {yellow}
\def \mycolfour {red}
\def \mycolfive {orange}
\def \mycolsix {purple}
\def \mycolseven {magenta}

\def \mycolone {black}
\def \mycoltwo {black}
\def \mycolthree {black}
\def \mycolfour {black}
\def \mycolfive {black}
\def \mycolsix {black}
\def \mycolseven {black}

\def \numevents{26.3}

\newcommand{\eventA}[4]{
\node[event, draw=#4, fill=white] (A#1) at (0*\xstep, #1*\ystep) {\footnotesize $#2(x_{#3})$};
}
\newcommand{\eventB}[4]{
\node[event, draw=#4, fill=white] (B#1) at (1*\xstep, #1*\ystep) {\footnotesize $#2(x_{#3})$};
}
\scalebox{0.95}{
\begin{tikzpicture}[thick,
pre/.style={<-,shorten >= 2pt, shorten <=2pt, very thick},
post/.style={->,shorten >= 2pt, shorten <=2pt,  very thick},
seqtrace/.style={->, line width=2},
und/.style={very thick, draw=gray},
event/.style={rectangle, minimum height=4.4mm, minimum width=10mm,  line width=1pt, inner sep=0.1,},
virt/.style={circle,draw=black!50,fill=black!20, opacity=0}]

\node[] (SA1) at (0*\xstep,+\ybias) {$\SeqTrace_A$};
\node[] (SA2) at (0,\numevents * \ystep) {};
\node[] (SB1) at (1*\xstep,+\ybias) {$\SeqTrace_B$};
\node[] (SB2) at (1*\xstep,\numevents * \ystep) {};

\draw[seqtrace] (SA1) to (SA2);
\draw[seqtrace] (SB1) to (SB2);

%

\node[] at (-0.6*\xstep, 7*\ystep){
$a_1=
\begin{bmatrix}
0 \\
1
\end{bmatrix}
$
};

\node[] at (-0.6*\xstep, 19*\ystep){
$a_2=
\begin{bmatrix}
1 \\
0
\end{bmatrix}
$
};

\node[] at (1.6*\xstep, 7*\ystep){
$b_2=
\begin{bmatrix}
0 \\
1
\end{bmatrix}
$
};
\node[] at (1.6*\xstep+2, 7*\ystep -0.15) {$\leftarrow$ Coordinate 1};
\node[] at (1.6*\xstep+2, 7*\ystep +0.25) {$\leftarrow$ Coordinate 2};

\node[] at (1.6*\xstep, 19*\ystep){
$b_1=
\begin{bmatrix}
1 \\
1
\end{bmatrix}
$
};

\eventA{1}{\Write_2^{a_1}}{1}{\mycolone}
\eventA{2}{\Write_2^{a_1}}{2}{\mycoltwo}
\eventA{3}{\Write_2^{a_1}}{3}{\mycolthree}
\eventA{4}{\Write_2^{a_1}}{6}{\mycolsix}
\eventA{5}{\Read_2^{a_1}}{2}{\mycoltwo}
\eventA{6}{\Write_1^{a_1}}{2}{\mycoltwo}
\eventA{7}{\Write_1^{a_1}}{1}{\mycolone}
\eventA{8}{\Read_1^{a_1}}{6}{\mycolsix}
\eventA{9}{\Write^{a_1}}{4}{\mycolfour}
\eventA{10}{\Write^{a_1}}{5}{\mycolfive}
\eventA{11}{\Read_1^{a_1}}{2}{\mycoltwo}
\eventA{12}{\Write^{a_1}}{7}{\mycolseven}
\eventA{13}{\Read^{a_1}}{5}{\mycolfive}

\eventA{15}{\Write_2^{a_2}}{2}{\mycoltwo}
\eventA{16}{\Write_2^{a_2}}{1}{\mycolone}
\eventA{17}{\Write_2^{a_2}}{3}{\mycolthree}
\eventA{18}{\Write_2^{a_2}}{6}{\mycolsix}
\eventA{19}{\Read_2^{a_2}}{2}{\mycoltwo}
\eventA{20}{\Write_1^{a_2}}{1}{\mycolone}
\eventA{21}{\Write_1^{a_2}}{2}{\mycoltwo}
\eventA{22}{\Read_1^{a_2}}{6}{\mycolsix}
\eventA{23}{\Read^{a_2}}{7}{\mycolseven}
\eventA{24}{\Write^{a_2}}{4}{\mycolfour}
\eventA{25}{\Read_1^{a_2}}{2}{\mycoltwo}

\eventB{1}{\Write_2^{b_2}}{2}{\mycoltwo}
\eventB{2}{\Write_2^{b_2}}{1}{\mycolone}
\eventB{3}{\Write_2^{b_2}}{3}{\mycolthree}
\eventB{4}{\Read_2^{b_2}}{1}{\mycolone}
\eventB{5}{\Write_2^{b_2}}{6}{\mycolsix}
\eventB{6}{\Write_1^{b_2}}{1}{\mycolone}
\eventB{7}{\Write_1^{b_2}}{2}{\mycoltwo}
\eventB{8}{\Read_1^{b_2}}{3}{\mycolthree}
\eventB{9}{\Write^{b_2}}{4}{\mycolfour}
\eventB{10}{\Read_1^{b_2}}{1}{\mycolone}
\eventB{11}{\Write^{b_2}}{5}{\mycolfive}

\eventB{15}{\Write_2^{b_1}}{1}{\mycolone}
\eventB{16}{\Write_2^{b_1}}{2}{\mycoltwo}
\eventB{17}{\Write_2^{b_1}}{3}{\mycolthree}
\eventB{18}{\Read_2^{b_1}}{1}{\mycolone}
\eventB{19}{\Write_2^{b_1}}{6}{\mycolsix}
\eventB{20}{\Write_1^{b_1}}{1}{\mycolone}
\eventB{21}{\Write_1^{b_1}}{2}{\mycoltwo}
\eventB{22}{\Read_1^{b_1}}{3}{\mycolthree}
\eventB{23}{\Write^{b_1}}{7}{\mycolseven}
\eventB{24}{\Read^{b_1}}{4}{\mycolfour}
\eventB{25}{\Read_1^{b_1}}{1}{\mycolone}

\draw[post, out=0, in=180, looseness=0.15] (A7) to ($ (A7) + (+0.7,0) $)  to ($ (B25) + (-0.7,0) $) to (B25);
\draw[post, out=180, in=0, looseness=0.15] (B7) to ($ (B7) + (-0.7,0) $) to ($ (A25) + (+0.7,0) $) to (A25);

\end{tikzpicture}
}
\caption{
Illustration of the reduction of an $\OrthVec$ instance $(A=\{a_1, a_2\},B=\{b_1, b_2\})$ to the Closure problem of an $\OPoset$.
}
\label{fig:ov}
\end{figure}
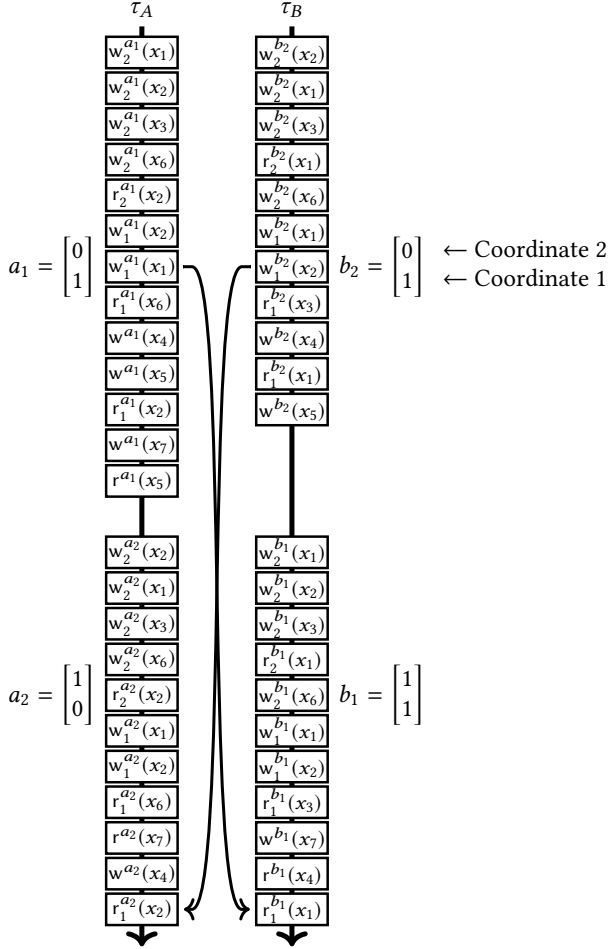

\SubParagraph{Events on $x_1$ and $x_2$.}
For every vector $a_j\in A$ and coordinate $i\in [D]$, we create three events
$\Write^{a_j}_i(x_1)$, $\Write^{a_j}_i(x_2)$ and $\Read^{a_j}_i(x_2)$.
We make $\Observation(\Read^{a_j}_i(x_2)) = \Write^{a_j}_i(x_2)$, and order 
\begin{align*}
\Write^{a_j}_i(x_1), \Write^{a_j}_i(x_2) <_{\SeqTrace_A} \Read^{a_j}_i(x_2)\ .
\numberthis\label{eq:x1x2_1}
\end{align*}
For every vector $b_l\in B$ and coordinate $i\in [D]$, we create three events
$\Write^{b_l}_i(x_1)$, $\Read^{b_l}_i(x_1)$ and $\Write^{b_l}_i(x_2)$.
We make $\Observation(\Read^{b_l}_i(x_1))=\Write^{b_l}_i(x_1)$, and order
\begin{align*}
\Write^{b_l}_i(x_1), \Write^{b_l}_i(x_2) <_{\SeqTrace_B}\Read^{b_l}_i(x_1)\ .
\numberthis\label{eq:x1x2_2}
\end{align*}
In addition, we order
\begin{align*}
\Write^{a_j}_i(x_2)<_{\SeqTrace_A}\Write^{a_j}_i(x_1) \quad &\text{iff}\quad a_j[i]=1
\qquad\text{and}\qquad\\
\Write^{b_l}_i(x_1)<_{\SeqTrace_B}\Write^{b_l}_i(x_2) \quad &\text{iff}\quad b_l[i]=1\ .
\numberthis\label{eq:x1x2_3}
\end{align*}
Observe that if 
$a_j[i]\cdot b_l[i]=1$, 
and if we order
$\Write^{a_j}_i(x_1)< \Write^{b_l}_i(x_1)$ 
then transitively
$\Write^{a_j}_i(x_2)< \Write^{b_l}_i(x_2)$ 
and hence by closure
$\Read^{a_j}_i(x_2)<\Write^{b_l}_i(x_2)$.

\SubParagraph{Events on $x_3$.}
Let $i\in [D-1]$ be a coordinate.
For every vector $a_j\in A$, we create an event $\Write^{a_j}_{i+1}(x_3)$, and order
\begin{align*}
\Write^{a_j}_{i+1}(x_1), \Write^{a_j}_{i+1}(x_2) <_{\SeqTrace_A} \Write^{a_j}_{i+1}(x_3) < _{\SeqTrace_A} \Write^{a_j}_{i}(x_1), \Write^{a_j}_{i}(x_2)\ .
\numberthis\label{eq:x3_1}
\end{align*}
For every vector $b_l\in B$, we create two events $\Write^{b_l}_{i+1}(x_3)$ and $\Read^{b_l}_{i}(x_3)$, and
make $\Observation(\Read^{b_l}_{i}(x_3)) = \Write^{b_l}_{i+1}(x_3)$.
We order
\begin{align*}
&\Write^{b_l}_{i+1}(x_1), \Write^{b_l}_{i+1}(x_2) <_{\SeqTrace_B} \Write^{b_l}_{i+1}(x_3)<_{\SeqTrace_B} \Read^{b_l}_{i+1}(x_1) 
\quad \text{and}\\
&\Write^{b_l}_{i}(x_1), \Write^{b_l}_{i}(x_2) <_{\SeqTrace_B} \Read^{b_l}_{i}(x_3)<_{\SeqTrace_B} \Read^{b_l}_{i}(x_1)\ .
\numberthis\label{eq:x3_2}
\end{align*}
Observe that if we order 
$\Write^{a_j}_i(x_1) < \Write^{b_l}_i(x_1)$ 
then we also have 
$\Write^{a_j}_{i+1}(x_3) < \Read^{b_l}_{i}(x_3)$, 
hence  by closure
$\Write^{a_j}_{i+1}(x_3) < \Write^{b_l}_{i+1}(x_3)$ 
and thus 
$\Write^{a_j}_{i+1}(x_1) <\Read^{b_l}_{i}(x_1)$.

\SubParagraph{Events on $x_6$.}
For every coordinate $i\in [D-1]$, we do as follows.
For every vector $a_j\in A$, we create two events $\Write^{a_j}_{i+1}(x_6)$ and $\Read^{a_j}_{i}(x_6)$,
and make $\Observation(\Read^{a_j}_{i}(x_6)) = \Write^{a_j}_{i+1}(x_6)$.
We order
\begin{align*}
&\Write^{a_j}_{i+1}(x_3) <_{\SeqTrace_A} \Write^{a_j}_{i+1}(x_6) <_{\SeqTrace_A} \Read^{a_j}_{i+1}(x_2)
\quad \text{and}\\
&\Write^{a_j}_{i}(x_1), \Write^{a_j}_{i}(x_2) <_{\SeqTrace_A} \Read^{a_j}_{i}(x_6)  <_{\SeqTrace_A} \Read^{a_j}_{i}(x_2)\ .
\numberthis\label{eq:x6_1}
\end{align*}
For every vector $b_l\in B$, we create one event $\Write^{b_l}_{i+1}(x_6)$, and order
\begin{align*}
\Read^{b_l}_{i+1}(x_1)<_{\SeqTrace_B}\Write^{b_l}_{i+1}(x_6) <_{\SeqTrace_B} \Write^{b_l}_{i}(x_1), \Write^{b_l}_{i}(x_2)\ .
\numberthis\label{eq:x6_2}
\end{align*}
Observe that if we order
$\Read^{a_j}_{i+1}(x_2)<\Write^{b_l}_{i+1}(x_2)$ ,
since $\Write_{i+1}^{b_l}(x_2)<_{\SeqTrace_B}\Read_{i+1}^{b_l}(x_1)$,
we also have  
$\Write^{a_j}_{i+1}(x_6)< \Write^{b_l}_{i+1}(x_6)$,  
hence by closure
$\Read^{a_j}_{i}(x_6) < \Write^{b_l}_{i+1}(x_6)$
and thus
$\Write^{a_j}_{i}(x_2) < \Write^{b_l}_{i}(x_2)$.

\SubParagraph{Events on $x_4$.}
For every vector $a_j\in A$, we create one event $\Write^{a_j}(x_4)$, and order
\begin{align*}
\Read^{a_j}_1(x_6) <_{\SeqTrace_A} \Write^{a_j}(x_4) <_{\SeqTrace_A} \Read^{a_j}_1(x_2)\ .
\numberthis\label{eq:x4_1}
\end{align*}
For every vector $b_l\in B$, with $l\in [n/2-1]$, we create two events $\Write^{b_{l+1}}(x_4)$ and $\Read^{b_l}(x_4)$, and make $\Observation(\Read^{b_l}(x_4))=\Write^{b_{l+1}}(x_4)$.
We order
\begin{align*}
&\Read^{b_l}_1(x_3)<_{\SeqTrace_B} \Read^{b_l}(x_4) <_{\SeqTrace_B} \Read^{b_l}_1(x_1)
\quad\text{and}\\
&\Read^{b_{l+1}}_1(x_3)<_{\SeqTrace_B} \Write^{b_{l+1}}(x_4) <_{\SeqTrace_B} \Read^{b_{l+1}}_1(x_1)\ .
\numberthis\label{eq:x4_2}
\end{align*}
Observe that if we order
$\Read^{a_j}_1(x_2)<\Write^{b_l}_1(x_2)$ 
then we also have
$\Write^{a_j}(x_4)<\Read^{b_l}(x_4)$ (since $\Write^{b_l}_1(x_2)<_{\SeqTrace_B} \Read^{b_l}_1(x_3)$ by \cref{eq:x3_2})
and thus by closure
$\Write^{a_j}(x_4)<\Write^{b_{l+1}}(x_4)$.

\SubParagraph{Events on $x_5$ and $x_7$.}
For every vector $a_j\in A$, with $j\in [n/2-1]$, we create two events $\Write^{a_j}(x_5)$ and $\Read^{a_j}(x_5)$, and make $\Observation(\Read^{a_j}(x_5))=\Write^{a_j}(x_5)$.
We also create two events $\Write^{a_j}(x_7)$ and $\Read^{a_{j+1}}(x_7)$, and make $\Observation(\Read^{a_{j+1}}(x_7))=\Write^{a_{j}}(x_7)$.
We order
\begin{align*}
&\Write^{a_j}(x_4) <_{\SeqTrace_A} \Write^{a_j}(x_5) <_{\SeqTrace_A} \Read^{a_j}_1(x_2) <_{\SeqTrace_A} \Write^{a_j}(x_7)<_{\SeqTrace_A} \Read^{a_j}(x_5) \text{ and}\\
&\Read^{a_{j+1}}_1(x_6) <_{\SeqTrace_A} \Read^{a_{j+1}}(x_7) <_{\SeqTrace_A} \Write^{a_{j+1}}(x_4)\ .
\numberthis\label{eq:x5x7_1}
\end{align*}
We also create two events $\Write^{b_{n/2}}(x_5)$ and $\Write^{b_{1}}(x_7)$, and order
\begin{align*}
&\Read_1^{b_{n/2}}(x_1) <_{\SeqTrace_B} \Write^{b_{n/2}}(x_5) <_{\SeqTrace_B} \Write_D^{b_{n/2}-1}(x_1), \Write_D^{b_{n/2}-1}(x_2) \quad \text{and}\\
&\Read^{b_1}_1(x_3) <_{\SeqTrace_B} \Write^{b_{1}}(x_7) <_{\SeqTrace_B} \Read^{b_1}(x_4)\ .
\numberthis\label{eq:x5x7_2}
\end{align*}
Observe that if we order
$\Read^{a_j}_1(x_2) < \Write^{b_{n/2}}_1(x_2)$
then we also have
$\Write^{a_j}(x_5)< \Write^{b_{n/2}}(x_5)$ (since  $\Write_1^{b_{n/2}}(x_2)<_{\SeqTrace_B}\Read_1^{b_{n/2}}(x_1)$ by a previous item)
and thus by closure
$\Read^{a_j}(x_5)< \Write^{b_{n/2}}(x_5)$.
But then also
$\Write^{a_j}(x_7)<\Write^{b_{1}}(x_7)$ (since $\Write^{b_{n/2}}(x_5)<_{\SeqTrace_B} \Write^{b_{1}}(x_7)$)
and thus by closure
$\Read^{a_{j+1}}(x_7)<\Write^{b_{1}}(x_7)$.

\SubParagraph{Final orderings.}
Finally, we order
\begin{align*}
&\Write^{b_{n/2}}(x_5) <_{\SeqTrace_B} \Write_{D}^{b_{n/2-1}}(x_1), \Write_{D}^{b_{n/2-1}}(x_2) \ . \numberthis\label{eq:final_1a}
\end{align*}
For each $j\in[n/2-1]$, we order
\begin{align*}
&\Read^{a_j}(x_5) <_{\SeqTrace_A} \Write_{D}^{a_{j+1}}(x_1), \Write_{D}^{a_{j+1}}(x_2)\ .
\end{align*}
For each $j\in[n/2-2]$, we order
\begin{align*}
&\Read^{b_{j+1}}(x_1) <_{\SeqTrace_B} \Write_{D}^{b_{j}}(x_1), \Write_{D}^{b_{j}}(x_2)
\numberthis\label{eq:final_1c}\ .
\end{align*}
We make two orderings across $\SeqTrace_A$ and $\SeqTrace_B$, namely
\begin{align*}
\Write^{a_1}_1(x_1) <_{P} \Read^{b_1}_1(x_1)
\quad\text{and}\quad
\Write^{b_{n/2}}_1(x_2) <_{P} \Read^{a_{n/2}}_1(x_2) \ .
\numberthis\label{eq:final_2}
\end{align*}

\SubParagraph{Correctness.}
Observe that we have used $7$ variables, while $|X_A|+|X_B|=O(n\cdot  D)$, 
and the reduction can be easily computed in linear time.
We refer to \cref{subsec:proofs_tree_topologies_lb} for the full proofs of the correctness of the above construction.
This concludes \cref{lem:closure_ov}, as any algorithm for rf-poset realizability on the above instances that runs in $O((n\cdot D)^{2-\epsilon})$ time also solves $\OrthVec$ in $O(n^{2-\epsilon}\cdot D^{O(1)})$ time.
Although the full proof is rather technical, the correctness is conceptually straightforward.
We illustrate the key idea on the example of \cref{fig:ov}, where we perform closure operations by inserting new orderings in a partial order $<$.

By construction, we have $\Write_1^{a_1}(x_1)<\Read_1^{b_1}(x_1)$, which signifies testing the first coordinate of vectors $a_1$ and $b_1$.
Note that $a_1[1]\cdot b_1[1]=1$, for which our encoding guarantees that eventually $\Read_1^{a_1}(x_2)<\Write_1^{b_1}(x_2)$.
Indeed, since $\Write_1^{a_1}(x_1)<\Read_1^{b_1}(x_1)$, by closure we also have  $\Write_1^{a_1}(x_1)<\Write_1^{b_1}(x_1)$.
In turn, this leads to $\Write_1^{a_1}(x_2)<\Write_1^{b_1}(x_2)$, and by closure,
we also have $\Read_1^{a_1}(x_2)<\Write_1^{b_1}(x_2)$.
This leads to $\Write^{a_1}(x_4)< \Read^{b_1}(x_4)$, and by closure,
we have $\Write^{a_1}(x_4)< \Write^{b_2}(x_4)$.
This last ordering leads to $\Write_1^{a_1}(x_1)<\Read_1^{b_2}(x_1)$, which signifies testing the first coordinate of vectors $a_1$ and $b_2$,
i.e., moving with the next vector of $B$.

The process for $a_1$ and $b_2$ is similar to $a_1$ and $b_1$, as the two vectors are found not orthogonal already in the first coordinate.
As previously, we eventually arrive at $\Read_1^{a_1}(x_2)<\Write_1^{b_2}(x_2)$.
Note that this leads to $\Write^{a_1}(x_5)<\Write^{b_2}(x_5)$, and by closure, 
we have $\Read^{a_1}(x_5)<\Write^{b_2}(x_5)$.
In turn, this leads to $\Write^{a_1}(x_7)<\Write^{b_1}(x_7)$, and by closure,
we have $\Read^{a_2}(x_7)< \Write^{b_1}(x_7)$.
This last ordering leads to $\Write_1^{a_2}(x_1)<\Read_1^{b_1}(x_1)$, which signifies testing the first coordinate of vectors $a_2$ and $b_1$,
i.e., moving with the next vector of $A$ and the first vector of $B$.

The process for $a_2$ and $b_1$ is initially different than before, as $a_2[1]\cdot b_1[1]=0$, i.e., the test on the first coordinate does not deem $a_2$ and $b_1$ not orthogonal.
By closure, the ordering $\Write_1^{a_2}(x_1)<\Read_1^{b_1}(x_1)$ leads to
$\Write_1^{a_2}(x_1)<\Write_1^{b_1}(x_1)$.
This leads to $\Write_2^{a_2}(x_3)<\Read_1^{b_1}(x_3)$, and by closure,
we have $\Write_2^{a_2}(x_3)<\Write_2^{b_1}(x_3)$.
This leads to $\Write_2^{a_2}(x_1)<\Read_2^{b_1}(x_1)$, which signifies testing the second coordinate of vectors $a_2$ and $b_1$.
As $a_2[2]\cdot b_2[2]=1$, the two vectors are discovered as non-orthogonal, which is captured by an eventual ordering
$\Read_2^{a_2}(x_2)<\Write_2^{b_1}(x_2)$.
This ordering which witnesses non-orthogonality is propagated downwards to the first coordinate, i.e., $\Read_1^{a_2}(x_2)<\Write_1^{b_1}(x_2)$.
This propagation is made by events on variable $x_6$.
Indeed, first note that, as $\Read_2^{a_2}(x_2)<\Write_2^{b_1}(x_2)$, we also have
$\Write_2^{a_2}(x_6)<\Write_2^{b_1}(x_6)$, and by closure,
we have $\Read_1^{a_2}(x_6)<\Write_2^{b_1}(x_6)$.
This leads to $\Write_1^{a_2}(x_2)<\Write_1^{b_1}(x_2)$, and by closure,
$\Read_1^{a_2}(x_2)<\Write_1^{b_1}(x_2)$, which marks the two vectors as non-orthogonal.
This leads to $\Write^{a_2}(x_4)< \Read^{b_1}(x_4)$, and by closure,
we have $\Write^{a_2}(x_4)< \Write^{b_2}(x_4)$.
This last ordering leads to $\Write_1^{a_2}(x_1)<\Read_1^{b_2}(x_1)$, which signifies testing the first coordinate of vectors $a_2$ and $b_2$,
i.e., moving with the next vector of $B$.

The process for $a_2$ and $b_2$ is initially similar to the previous case, as $a_2[1]\cdot b_2[1]=0$.
However, because we also have $a_2[2]\cdot b_2[2]=0$, we will \emph{not} order $\Read_2^{a_2}(x_2)<\Write_2^{b_2}(x_2)$,
and the closure will terminate after ordering $\Write_2^{a_2}(x_1)<\Write_2^{b_2}(x_1)$.
Since no cyclical orderings were introduce, the closure of $\OPoset$ exists, and by \cref{lem:tree_topologies}, $\OPoset$ is realizable.
Finally, observe that \emph{if} we eventually had $\Read_1^{a_2}(x_2)<\Write_1^{b_2}(x_2)$ (signifying that $a_2$ and $b_2$ are not orthogonal, hence there is no orthogonal pair in $A\times B$), this would create a cycle with the ordering $\Write_1^{b_2}(x_2)<_{P}\Read_1^{a_2}(x_2)$, and by \cref{rem:closure}, $\OPoset$ would \emph{not} be realizable.

\Paragraph{Reduction to dynamic data-race prediction.}
Consider an instance of the rf-poset $\OPoset=(X,P,\Observation)$ constructed in the above reduction, and we construct a trace $\Trace$ and two events $\Event_1,\Event_2\in \Events{\Trace}$ such that $\OPoset$ is realizable iff $(\Event_1, \Event_2)$ is a predictable data race of $\Trace$.
The trace $\Trace$ consists of two threads $\Process_A,\Process_B$
and two local traces $\SeqTrace'_A$ and $\SeqTrace'_B$ such that $\SeqTrace'_A$ and $\SeqTrace'_B$ contain the events of $\Process_A$ and $\Process_B$, respectively.
Each of $\SeqTrace'_A$ and $\SeqTrace'_B$ is identical to $\SeqTrace_A$ and $\SeqTrace_B$ of $\OPoset$, respectively, with some additional events inserted in it.
In particular, besides the variables $x_i$, $i\in[7]$ that appear in the events of $X$, we introduce one variable $y$ and one lock $\ell$.
For the event set, we have
\begin{align*}
\Events{\Trace} = &X \cup \{ \Write(y), \Read(y) \} \cup\{ \Acquire_A(\ell), \Release_A(\ell) \} \cup\\
&\{ \Acquire_B(\ell), \Release_B(\ell) \} \cup \{ \Write(z), \Read(z) \}\ .
\end{align*}
The  local traces $\SeqTrace'_A$ and $\SeqTrace'_B$ are constructed as follows.
\begin{compactenum}
\item For $\SeqTrace'_A$, we insert an empty critical section $\Acquire_A(\ell), \Release_A(\ell)$ right after $\Write_1^{a_1}(x_1)$.
Additionally, we insert the read event $\Read(y)$ right before $\Read_1^{a_{n/2}}(x_2)$,
and the event $\Read(z)$ as the last event of $\SeqTrace'_A$.
\item For $\SeqTrace'_B$, we insert the write event $\Write(y)$ right after $\Write^{b_{n/2}}_1 (x_2)$.
Additionally, we insert $\Write(z)$ right after $\Read_1^{b_1}(x_1)$, and surround these two events with $\Acquire_B(\ell), \Release_B(\ell)$.
\end{compactenum}
Finally, we obtain $\Trace$ as $\Trace=\SeqTrace'_B \circ \SeqTrace'_A$, i.e., the two local traces are executed sequentially and there is no context switching.
The task is to decide whether $(\Write(z), \Read(z))$ is a predictable data race of $\Trace$.
We refer to \cref{subsec:proofs_tree_topologies_lb} for the correctness of the construction, which concludes \cref{them:tree_topologies_lb}.


\section{Witnesses in Small Distance}\label{sec:small_distance}

The results in the previous sections neglect information provided by the input trace $\Trace$ about constructing a correct reordering that witnesses the data race.
Indeed, our hardness results show that, in the worst case, the orderings in $\Trace$ provide no help.
However, in practice when a data race exists, a witness trace $\Trace^*$ can be constructed that is similar to $\Trace$.
In fact, virtually all practical techniques predict data races by constructing $\Trace^*$ to be very similar to $\Trace$ (e.g., \cite{Kini17,Roemer18,Flanagan09,Smaragdakis12,Pavlogiannis19}).

\Paragraph{The distance-bounded realizability problem of feasible trace ideals.}
Given a natural number $\ell$, a trace $\Trace$ and a feasible trace ideal $X$ of $\Trace$, the solution
to the \emph{$\ell$-distance-bounded realizability} problem is $\False$ if $X$ is not realizable,
$\True$ if there is a witness $\Trace^*$ that realizes $X$ such that
$\Distance(\Trace,\Trace^*) \leq \ell$, 
and can be any answer ($\True/\False$) if $X$ is realizable but any witness $\Trace^*$ 
that realizes $X$ is such that $\Distance(\Trace,\Trace^*) > \ell$.
We remark that this formulation is that of a \emph{promise problem}~\cite{promise_problem}.
We are interested in the case where $\ell=O(1)$.
There exists a straightforward algorithm that operates in $O(|X|^{2\cdot \ell})$ time.
The algorithm iterates over all possible subsets of pairs of conflicting write and lock-acquire events that have size at most $\ell$, and tries all possible combinations of conflicting-write reversals in that set.
\cref{them:small_distance} is based on the following lemma, which states that the problem can be solved much faster when $k$ is also constant.

\smallskip
\begin{restatable}{lemma}{lemreversalrealizability}\label{lem:reversal_realizability}
Consider a natural number $\ell$, a trace $\Trace$ over $n$ events and $k$ threads,
and a feasible trace ideal $X$ of $\Trace$.
The $\ell$-distance-bounded realizability problem for $X$ can be solved in $O(k^{\ell+O(1)}\cdot n)$ time.
\end{restatable}


\smallskip
\begin{proof}[Proof of \cref{them:small_distance}.]
By \cref{lem:reversal_realizability}, given a trace ideal $X$ of $\Trace$, we can solve the $\ell$-distance-bounded realizability problem for $X$ in $O(n)$ time.
The proof then follows by \cref{lem:candidate_set} and \cref{lem:num_candidate_ideals}, as to decide whether $(\Event_1, \Event_2)$ is a predictable data race of $\Trace$, it suffices to examine $O(1)$ trace ideals of $\Trace$.
\end{proof}

In the remaining of this section we prove \cref{lem:reversal_realizability}.
We first define the notion of read extensions of graphs.
Afterwards, we present the algorithm for the lemma, and show its correctness and complexity.

\Paragraph{Read extensions.}
Consider a digraph $G=(X,E)$ where $X$ is a set of events.
Given two events $\Event_1, \Event_2\in G$, we write $\Event_1\Path\Event_2$ to denote that $\Event_2$ is reachable from $\Event_1$.
We call $G$ \emph{write-ordered} if for every two distinct conflicting write or lock-acquire events $\Write_1, \Write_2\in \WritesAcquires{X}$, we have $\Write_1\Path \Write_2$ or $\Write_2\Path \Write_1$ in $G$.
Given an acyclic write-ordered graph $G_1=(X,E_1)$,
the \emph{read extension} of $G_1$ is the digraph $G_2=(X,E_2)$ where $E_2=E_1\cup A \cup B$,
where the sets $A$ and $B$ are defined as follows.
\begin{align*}
A=&\setpred{ (\Read, \Write)\in \ReadsReleases{X}\times \WritesAcquires{X}}{\Confl{\Read}{\Write} \text{ and } (\Observation_{\Trace}(\Read), \Write) \in E_1}\ ,\\
B=&\setpred{ (\Write, \Read)\in \WritesAcquires{X} \times \ReadsReleases{X}}{\Confl{\Read}{\Write} \text{ and } (\Write, \Observation_{\Trace}(\Read)) \in E_1}\ .
\end{align*}

\Paragraph{A fast algorithm for distance-bounded rf-poset realizability.}
Let $\OPoset=(X,P,\Observation)$ be the canonical rf-poset of $X$,
and the task is to decide the realizability of $\OPoset$ with $\ell$ reversals.
We describe a recursive algorithm for solving the problem for some rf-poset $\OPosetQ=(X,Q,\Observation)$  with $\ell'$ reversals, for some $\ell'\leq \ell$, where initially $Q=P$ and $\ell'=\ell$.

\SubParagraph{Algorithm and correctness.}
Consider the set
\begin{align*}
C=&\setpred{ (\Write_1, \Write_2)\in \WritesAcquires{X}\times \WritesAcquires{X}}{\Confl{\Write_1}{\Write_2}
\quad\text{and}\\
&\Unordered{\Write_1}{Q}{\Write_2} \text{ and } \Write_1<_{\Trace}\Write_2}\ .
\end{align*}
We construct a graph $G_1=(X,E_1)$, where $E_1=(\TOO\Project X) \cup C$.
Note that $G_1$ is write-ordered. 
If it is acyclic, we construct the read extension $G_2$ of $G_1$.
Observe that if $G_2$ is acyclic then any linearization $\Trace^*$ of $G$ realizes $\OPosetQ$, hence we are done.
Now consider that either $G_1$ or $G_2$ is not acyclic,
and let $G=G_1$ if $G_1$ is not acyclic, otherwise $G=G_2$.
Given a cycle $\Cycle$ of $G$, represented as a collection of edges, 
define the set of \emph{cross-edges} of $\Cycle$ as $\Cycle \setminus Q$.
Note that, since there are $k$ threads, $G$ has a cycle with $\leq k$ cross edges.
In addition, any trace $\Trace^*$ that realizes $\OPosetQ$ must linearize 
an rf-poset $(X,Q_a, \Observation)$
 where $a=(\Event_1, \Event_2)$ ranges over the cross-edges of $\Cycle.$
In particular, we take $Q_a=Q\cup \{ b \}$, where
\begin{align*}
b=
\begin{cases}
(\Event_2, \Event_1), & \text{ if } a\in \WritesAcquires{X}\times \WritesAcquires{X}\\
(\Observation(\Event_2), \Event_1), & \text{ if } a \in \WritesAcquires{X} \times \ReadsReleases{X}\\
(\Event_2, \Observation(\Event_1)), & \text{ if } a \in \ReadsReleases{X} \times \WritesAcquires{X}\ .
\end{cases}
\end{align*}

Observe that any such choice of $b$ reverses the order of two conflicting write events or lock-acquire events in $\Trace$.
Since there are $\leq k$ cross edges in $\Cycle$, there are $\leq k$ such choices for $Q_a$.
Repeating the same process recursively for the rf-poset $(X, Q_a, \Observation)$ for $\ell'-1$ levels solves the $\ell'$-distance-bounded realizability problem for $\OPosetQ$.
Since initially $\ell'=\ell$ and $Q=P$, this process solves the same problem for $\OPoset$ and thus for $X$.

\SubParagraph{Complexity.}
The recursion tree above has branching $\leq k$ and depth $\leq \ell$, hence there will be at most $k^{\ell}$ recursive instances.
In \cref{sec:proofs_small_distance}, we provide some lower-level algorithmic details which show that each instance can be solved in $O(k^{O(1)}\cdot n)$ time.
The main idea is that each of the graphs $G_1$ and $G_2$ have a sparse transitive reduction~\cite{Aho72} of size $O(k\cdot n)$, 
and thus each graph can be analyzed in $O(k\cdot n)$ time.


\section{Conclusion}\label{sec:conclusion}

In this work, we have studied the complexity of dynamic data-race prediction,
and have drawn a rich complexity landscape depending on various parameters of the input trace.
Our main results indicate that the problem is in polynomial time when the number of threads is bounded,
however, it is unlikely to be FPT wrt this parameter.
On the other hand, we have shown that the problem can be solved in, essentially, quadratic time, when the communication topology is acyclic.
We have also proved a quadratic lower bound for this case, which shows that our algorithm for tree communication topologies is optimal.
Finally, motivated by practical techniques, we have shown that a distance-bounded version of data-race prediction can be solved in linear time under mild assumptions on the input parameters.

\begin{acks}
We thank anonymous reviewers for their constructive feedback on an earlier draft of this manuscript.
\cref{rem:eth} is due to anonymous reviewer.
Umang Mathur is partially supported by a Google PhD Fellowship.
Mahesh Viswanathan was partially supported by grant NSF SHF 1901069.
\end{acks}

%

\bibliography{bibliography}


\begin{thebibliography}{43}


\ifx \showCODEN    \undefined \def \showCODEN     #1{\unskip}     \fi
\ifx \showDOI      \undefined \def \showDOI       #1{#1}\fi
\ifx \showISBNx    \undefined \def \showISBNx     #1{\unskip}     \fi
\ifx \showISBNxiii \undefined \def \showISBNxiii  #1{\unskip}     \fi
\ifx \showISSN     \undefined \def \showISSN      #1{\unskip}     \fi
\ifx \showLCCN     \undefined \def \showLCCN      #1{\unskip}     \fi
\ifx \shownote     \undefined \def \shownote      #1{#1}          \fi
\ifx \showarticletitle \undefined \def \showarticletitle #1{#1}   \fi
\ifx \showURL      \undefined \def \showURL       {\relax}        \fi
\providecommand\bibfield[2]{#2}
\providecommand\bibinfo[2]{#2}
\providecommand\natexlab[1]{#1}
\providecommand\showeprint[2][]{arXiv:#2}

\bibitem[\protect\citeauthoryear{Aalbersberg}{Aalbersberg}{1988}]%
        {Aalbersberg88}
\bibfield{author}{\bibinfo{person}{I.~J. Aalbersberg}.}
  \bibinfo{year}{1988}\natexlab{}.
\newblock \showarticletitle{Theory of Traces}.
\newblock \bibinfo{journal}{{\em Theor. Comput. Sci.\/}} \bibinfo{volume}{60},
  \bibinfo{number}{1} (\bibinfo{date}{Sept.} \bibinfo{year}{1988}),
  \bibinfo{pages}{1–82}.
\newblock
\showISSN{0304-3975}
\showDOI{%
\url{https://doi.org/10.1016/0304-3975(88)90051-5}}


\bibitem[\protect\citeauthoryear{Abdulla, Atig, Jonsson, L\r{a}ng, Ngo, and
  Sagonas}{Abdulla et~al\mbox{.}}{2019}]%
        {Abdulla19}
\bibfield{author}{\bibinfo{person}{Parosh~Aziz Abdulla},
  \bibinfo{person}{Mohamed~Faouzi Atig}, \bibinfo{person}{Bengt Jonsson},
  \bibinfo{person}{Magnus L\r{a}ng}, \bibinfo{person}{Tuan~Phong Ngo}, {and}
  \bibinfo{person}{Konstantinos Sagonas}.} \bibinfo{year}{2019}\natexlab{}.
\newblock \showarticletitle{Optimal Stateless Model Checking for Reads-from
  Equivalence under Sequential Consistency}.
\newblock \bibinfo{journal}{{\em Proc. ACM Program. Lang.\/}}
  \bibinfo{volume}{3}, \bibinfo{number}{OOPSLA}, Article
  \bibinfo{articleno}{Article 150} (\bibinfo{date}{Oct.} \bibinfo{year}{2019}),
  \bibinfo{numpages}{29}~pages.
\newblock
\showDOI{%
\url{https://doi.org/10.1145/3360576}}


\bibitem[\protect\citeauthoryear{Aho, Garey, and Ullman}{Aho
  et~al\mbox{.}}{1972}]%
        {Aho72}
\bibfield{author}{\bibinfo{person}{A.~V. Aho}, \bibinfo{person}{M.~R. Garey},
  {and} \bibinfo{person}{J.~D. Ullman}.} \bibinfo{year}{1972}\natexlab{}.
\newblock \showarticletitle{The Transitive Reduction of a Directed Graph}.
\newblock \bibinfo{journal}{{\it SIAM J. Comput.}} \bibinfo{volume}{1},
  \bibinfo{number}{2} (\bibinfo{year}{1972}), \bibinfo{pages}{131--137}.
\newblock
\showDOI{%
\url{https://doi.org/10.1137/0201008}}


\bibitem[\protect\citeauthoryear{Bertoni, Mauri, and Sabadini}{Bertoni
  et~al\mbox{.}}{1989}]%
        {Bertoni89}
\bibfield{author}{\bibinfo{person}{A. Bertoni}, \bibinfo{person}{G. Mauri},
  {and} \bibinfo{person}{N. Sabadini}.} \bibinfo{year}{1989}\natexlab{}.
\newblock \showarticletitle{Membership Problems for Regular and Context-Free
  Trace Languages}.
\newblock \bibinfo{journal}{{\em Inf. Comput.\/}} \bibinfo{volume}{82},
  \bibinfo{number}{2} (\bibinfo{date}{Aug.} \bibinfo{year}{1989}),
  \bibinfo{pages}{135–150}.
\newblock
\showISSN{0890-5401}
\showDOI{%
\url{https://doi.org/10.1016/0890-5401(89)90051-5}}


\bibitem[\protect\citeauthoryear{Biswas and Enea}{Biswas and Enea}{2019}]%
        {biswas19}
\bibfield{author}{\bibinfo{person}{Ranadeep Biswas} {and}
  \bibinfo{person}{Constantin Enea}.} \bibinfo{year}{2019}\natexlab{}.
\newblock \showarticletitle{On the Complexity of Checking Transactional
  Consistency}.
\newblock \bibinfo{journal}{{\em Proc. ACM Program. Lang.\/}}
  \bibinfo{volume}{3}, \bibinfo{number}{OOPSLA}, Article
  \bibinfo{articleno}{Article 165} (\bibinfo{date}{Oct.} \bibinfo{year}{2019}),
  \bibinfo{numpages}{28}~pages.
\newblock
\showDOI{%
\url{https://doi.org/10.1145/3360591}}


\bibitem[\protect\citeauthoryear{Boehm}{Boehm}{2011}]%
        {boehmbenign2011}
\bibfield{author}{\bibinfo{person}{Hans-J. Boehm}.}
  \bibinfo{year}{2011}\natexlab{}.
\newblock \showarticletitle{How to Miscompile Programs with “Benign” Data
  Races}. In \bibinfo{booktitle}{{\em Proceedings of the 3rd USENIX Conference
  on Hot Topic in Parallelism}} {\em (\bibinfo{series}{HotPar’11})}.
  \bibinfo{publisher}{USENIX Association}, \bibinfo{address}{USA},
  \bibinfo{pages}{3}.
\newblock


\bibitem[\protect\citeauthoryear{Boehm}{Boehm}{2012}]%
        {evil2012}
\bibfield{author}{\bibinfo{person}{Hans-J. Boehm}.}
  \bibinfo{year}{2012}\natexlab{}.
\newblock \showarticletitle{Position Paper: Nondeterminism is Unavoidable, but
  Data Races Are Pure Evil}. In \bibinfo{booktitle}{{\em Proceedings of the
  2012 ACM Workshop on Relaxing Synchronization for Multicore and Manycore
  Scalability}} {\em (\bibinfo{series}{RACES ’12})}.
  \bibinfo{publisher}{Association for Computing Machinery},
  \bibinfo{address}{New York, NY, USA}, \bibinfo{pages}{9–14}.
\newblock
\showISBNx{9781450316323}
\showDOI{%
\url{https://doi.org/10.1145/2414729.2414732}}


\bibitem[\protect\citeauthoryear{Bringmann}{Bringmann}{2019}]%
        {Bringmann19}
\bibfield{author}{\bibinfo{person}{Karl Bringmann}.}
  \bibinfo{year}{2019}\natexlab{}.
\newblock \showarticletitle{{Fine-Grained Complexity Theory (Tutorial)}}. In
  \bibinfo{booktitle}{{\em 36th International Symposium on Theoretical Aspects
  of Computer Science (STACS 2019)}} {\em (\bibinfo{series}{Leibniz
  International Proceedings in Informatics (LIPIcs)})},
  \bibfield{editor}{\bibinfo{person}{Rolf Niedermeier} {and}
  \bibinfo{person}{Christophe Paul}} (Eds.), Vol.~\bibinfo{volume}{126}.
  \bibinfo{publisher}{Schloss Dagstuhl--Leibniz-Zentrum fuer Informatik},
  \bibinfo{address}{Dagstuhl, Germany}, \bibinfo{pages}{4:1--4:7}.
\newblock
\showISBNx{978-3-95977-100-9}
\showISSN{1868-8969}
\showDOI{%
\url{https://doi.org/10.4230/LIPIcs.STACS.2019.4}}


\bibitem[\protect\citeauthoryear{Chalupa, Chatterjee, Pavlogiannis, Vaidya, and
  Sinha}{Chalupa et~al\mbox{.}}{2018}]%
        {Chalupa18}
\bibfield{author}{\bibinfo{person}{Marek Chalupa}, \bibinfo{person}{Krishnendu
  Chatterjee}, \bibinfo{person}{Andreas Pavlogiannis}, \bibinfo{person}{Kapil
  Vaidya}, {and} \bibinfo{person}{Nishant Sinha}.}
  \bibinfo{year}{2018}\natexlab{}.
\newblock \showarticletitle{Data-centric Dynamic Partial Order Reduction}. In
  \bibinfo{booktitle}{{\em Proceedings of the 45rd Annual ACM SIGPLAN-SIGACT
  Symposium on Principles of Programming Languages}} {\em
  (\bibinfo{series}{POPL '18})}.
\newblock


\bibitem[\protect\citeauthoryear{Chen and Ro\c{s}u}{Chen and Ro\c{s}u}{2007}]%
        {Chen07}
\bibfield{author}{\bibinfo{person}{Feng Chen} {and} \bibinfo{person}{Grigore
  Ro\c{s}u}.} \bibinfo{year}{2007}\natexlab{}.
\newblock \showarticletitle{Parametric and Sliced Causality}. In
  \bibinfo{booktitle}{{\em Proceedings of the 19th International Conference on
  Computer Aided Verification}} {\em (\bibinfo{series}{CAV'07})}.
  \bibinfo{publisher}{Springer-Verlag}, \bibinfo{address}{Berlin, Heidelberg},
  \bibinfo{pages}{240--253}.
\newblock
\showISBNx{978-3-540-73367-6}
\showURL{%
\url{http://dl.acm.org/citation.cfm?id=1770351.1770387}}


\bibitem[\protect\citeauthoryear{Chen, Huang, Kanj, and Xia}{Chen
  et~al\mbox{.}}{2006}]%
        {Chen06}
\bibfield{author}{\bibinfo{person}{Jianer Chen}, \bibinfo{person}{Xiuzhen
  Huang}, \bibinfo{person}{Iyad~A. Kanj}, {and} \bibinfo{person}{Ge Xia}.}
  \bibinfo{year}{2006}\natexlab{}.
\newblock \showarticletitle{Strong computational lower bounds via parameterized
  complexity}.
\newblock \bibinfo{journal}{{\it J. Comput. System Sci.}} \bibinfo{volume}{72},
  \bibinfo{number}{8} (\bibinfo{year}{2006}), \bibinfo{pages}{1346 -- 1367}.
\newblock
\showISSN{0022-0000}
\showDOI{%
\url{https://doi.org/10.1016/j.jcss.2006.04.007}}


\bibitem[\protect\citeauthoryear{Downey and Fellows}{Downey and
  Fellows}{1999}]%
        {Rodney99}
\bibfield{author}{\bibinfo{person}{Rodney~G. Downey} {and}
  \bibinfo{person}{Michael~R. Fellows}.} \bibinfo{year}{1999}\natexlab{}.
\newblock \bibinfo{booktitle}{{\em Parameterized Complexity}}.
\newblock \bibinfo{publisher}{Springer}.
\newblock
\showISBNx{978-1-4612-6798-0}
\showDOI{%
\url{https://doi.org/10.1007/978-1-4612-0515-9}}


\bibitem[\protect\citeauthoryear{Emmi and Enea}{Emmi and Enea}{2017}]%
        {emmi17}
\bibfield{author}{\bibinfo{person}{Michael Emmi} {and}
  \bibinfo{person}{Constantin Enea}.} \bibinfo{year}{2017}\natexlab{}.
\newblock \showarticletitle{Sound, Complete, and Tractable Linearizability
  Monitoring for Concurrent Collections}.
\newblock \bibinfo{journal}{{\em Proc. ACM Program. Lang.\/}}
  \bibinfo{volume}{2}, \bibinfo{number}{POPL}, Article
  \bibinfo{articleno}{Article 25} (\bibinfo{date}{Dec.} \bibinfo{year}{2017}),
  \bibinfo{numpages}{27}~pages.
\newblock
\showDOI{%
\url{https://doi.org/10.1145/3158113}}


\bibitem[\protect\citeauthoryear{Even, Selman, and Yacobi}{Even
  et~al\mbox{.}}{1984}]%
        {promise_problem}
\bibfield{author}{\bibinfo{person}{Shimon Even}, \bibinfo{person}{Alan~L.
  Selman}, {and} \bibinfo{person}{Yacov Yacobi}.}
  \bibinfo{year}{1984}\natexlab{}.
\newblock \showarticletitle{The Complexity of Promise Problems with
  Applications to Public-Key Cryptography}.
\newblock \bibinfo{journal}{{\em Inf. Control\/}} \bibinfo{volume}{61},
  \bibinfo{number}{2} (\bibinfo{date}{May} \bibinfo{year}{1984}),
  \bibinfo{pages}{159–173}.
\newblock
\showISSN{0019-9958}
\showDOI{%
\url{https://doi.org/10.1016/S0019-9958(84)80056-X}}


\bibitem[\protect\citeauthoryear{Flanagan and Freund}{Flanagan and
  Freund}{2009}]%
        {Flanagan09}
\bibfield{author}{\bibinfo{person}{Cormac Flanagan} {and}
  \bibinfo{person}{Stephen~N. Freund}.} \bibinfo{year}{2009}\natexlab{}.
\newblock \showarticletitle{FastTrack: Efficient and Precise Dynamic Race
  Detection}. In \bibinfo{booktitle}{{\em Proceedings of the 30th ACM SIGPLAN
  Conference on Programming Language Design and Implementation}} {\em
  (\bibinfo{series}{PLDI '09})}. \bibinfo{publisher}{ACM},
  \bibinfo{address}{New York, NY, USA}, \bibinfo{pages}{121--133}.
\newblock
\showISBNx{978-1-60558-392-1}
\showDOI{%
\url{https://doi.org/10.1145/1542476.1542490}}


\bibitem[\protect\citeauthoryear{Gibbons and Korach}{Gibbons and
  Korach}{1997}]%
        {Gibbons97}
\bibfield{author}{\bibinfo{person}{Phillip~B. Gibbons} {and}
  \bibinfo{person}{Ephraim Korach}.} \bibinfo{year}{1997}\natexlab{}.
\newblock \showarticletitle{Testing Shared Memories}.
\newblock \bibinfo{journal}{{\em SIAM J. Comput.\/}} \bibinfo{volume}{26},
  \bibinfo{number}{4} (\bibinfo{date}{Aug.} \bibinfo{year}{1997}),
  \bibinfo{pages}{1208--1244}.
\newblock
\showISSN{0097-5397}
\showDOI{%
\url{https://doi.org/10.1137/S0097539794279614}}


\bibitem[\protect\citeauthoryear{Herlihy and Wing}{Herlihy and Wing}{1990}]%
        {HerlihyWing90}
\bibfield{author}{\bibinfo{person}{Maurice~P. Herlihy} {and}
  \bibinfo{person}{Jeannette~M. Wing}.} \bibinfo{year}{1990}\natexlab{}.
\newblock \showarticletitle{Linearizability: A Correctness Condition for
  Concurrent Objects}.
\newblock \bibinfo{journal}{{\em ACM Trans. Program. Lang. Syst.\/}}
  \bibinfo{volume}{12}, \bibinfo{number}{3} (\bibinfo{date}{July}
  \bibinfo{year}{1990}), \bibinfo{pages}{463–492}.
\newblock
\showISSN{0164-0925}
\showDOI{%
\url{https://doi.org/10.1145/78969.78972}}


\bibitem[\protect\citeauthoryear{Huang, Meredith, and Rosu}{Huang
  et~al\mbox{.}}{2014}]%
        {Huang14}
\bibfield{author}{\bibinfo{person}{Jeff Huang}, \bibinfo{person}{Patrick~O'Neil
  Meredith}, {and} \bibinfo{person}{Grigore Rosu}.}
  \bibinfo{year}{2014}\natexlab{}.
\newblock \showarticletitle{Maximal Sound Predictive Race Detection with
  Control Flow Abstraction}. In \bibinfo{booktitle}{{\em Proceedings of the
  35th ACM SIGPLAN Conference on Programming Language Design and
  Implementation}} {\em (\bibinfo{series}{PLDI '14})}.
  \bibinfo{publisher}{ACM}, \bibinfo{address}{New York, NY, USA},
  \bibinfo{pages}{337--348}.
\newblock
\showISBNx{978-1-4503-2784-8}
\showDOI{%
\url{https://doi.org/10.1145/2594291.2594315}}


\bibitem[\protect\citeauthoryear{Kasikci, Zamfir, and Candea}{Kasikci
  et~al\mbox{.}}{2013}]%
        {racemob2013}
\bibfield{author}{\bibinfo{person}{Baris Kasikci}, \bibinfo{person}{Cristian
  Zamfir}, {and} \bibinfo{person}{George Candea}.}
  \bibinfo{year}{2013}\natexlab{}.
\newblock \showarticletitle{RaceMob: Crowdsourced Data Race Detection}. In
  \bibinfo{booktitle}{{\em Proceedings of the Twenty-Fourth ACM Symposium on
  Operating Systems Principles}} {\em (\bibinfo{series}{SOSP ’13})}.
  \bibinfo{publisher}{Association for Computing Machinery},
  \bibinfo{address}{New York, NY, USA}, \bibinfo{pages}{406–422}.
\newblock
\showISBNx{9781450323888}
\showDOI{%
\url{https://doi.org/10.1145/2517349.2522736}}


\bibitem[\protect\citeauthoryear{Kini, Mathur, and Viswanathan}{Kini
  et~al\mbox{.}}{2017}]%
        {Kini17}
\bibfield{author}{\bibinfo{person}{Dileep Kini}, \bibinfo{person}{Umang
  Mathur}, {and} \bibinfo{person}{Mahesh Viswanathan}.}
  \bibinfo{year}{2017}\natexlab{}.
\newblock \showarticletitle{Dynamic Race Prediction in Linear Time}. In
  \bibinfo{booktitle}{{\em Proceedings of the 38th ACM SIGPLAN Conference on
  Programming Language Design and Implementation}} {\em (\bibinfo{series}{PLDI
  2017})}. \bibinfo{publisher}{ACM}, \bibinfo{address}{New York, NY, USA},
  \bibinfo{pages}{157--170}.
\newblock
\showISBNx{978-1-4503-4988-8}
\showDOI{%
\url{https://doi.org/10.1145/3062341.3062374}}


\bibitem[\protect\citeauthoryear{Lamport}{Lamport}{1978}]%
        {Lamport78}
\bibfield{author}{\bibinfo{person}{Leslie Lamport}.}
  \bibinfo{year}{1978}\natexlab{}.
\newblock \showarticletitle{Time, Clocks, and the Ordering of Events in a
  Distributed System}.
\newblock \bibinfo{journal}{{\em Commun. ACM\/}} \bibinfo{volume}{21},
  \bibinfo{number}{7} (\bibinfo{date}{July} \bibinfo{year}{1978}),
  \bibinfo{pages}{558--565}.
\newblock
\showISSN{0001-0782}
\showDOI{%
\url{https://doi.org/10.1145/359545.359563}}


\bibitem[\protect\citeauthoryear{Mathur, Kini, and Viswanathan}{Mathur
  et~al\mbox{.}}{2018}]%
        {Mathur18}
\bibfield{author}{\bibinfo{person}{Umang Mathur}, \bibinfo{person}{Dileep
  Kini}, {and} \bibinfo{person}{Mahesh Viswanathan}.}
  \bibinfo{year}{2018}\natexlab{}.
\newblock \showarticletitle{What Happens-after the First Race? Enhancing the
  Predictive Power of Happens-before Based Dynamic Race Detection}.
\newblock \bibinfo{journal}{{\em Proc. ACM Program. Lang.\/}}
  \bibinfo{volume}{2}, \bibinfo{number}{OOPSLA}, Article
  \bibinfo{articleno}{145} (\bibinfo{date}{Oct.} \bibinfo{year}{2018}),
  \bibinfo{numpages}{29}~pages.
\newblock
\showISSN{2475-1421}
\showDOI{%
\url{https://doi.org/10.1145/3276515}}


\bibitem[\protect\citeauthoryear{Mattern}{Mattern}{1989}]%
        {Mattern89}
\bibfield{author}{\bibinfo{person}{Friedemann Mattern}.}
  \bibinfo{year}{1989}\natexlab{}.
\newblock \showarticletitle{Virtual Time and Global States of Distributed
  Systems}.
\newblock In \bibinfo{booktitle}{{\em Parallel and Distributed Algorithms:
  proceedings of the International Workshop on Parallel \& Distributed
  Algorithms}}, \bibfield{editor}{\bibinfo{person}{M.~Cosnard et. al.}} (Ed.).
  \bibinfo{publisher}{Elsevier Science Publishers B. V.},
  \bibinfo{pages}{215--226}.
\newblock


\bibitem[\protect\citeauthoryear{Mazurkiewicz}{Mazurkiewicz}{1987}]%
        {Mazurkiewicz87}
\bibfield{author}{\bibinfo{person}{A Mazurkiewicz}.}
  \bibinfo{year}{1987}\natexlab{}.
\newblock \showarticletitle{Trace Theory}. In \bibinfo{booktitle}{{\em Advances
  in Petri Nets 1986, Part II on Petri Nets: Applications and Relationships to
  Other Models of Concurrency}}. \bibinfo{publisher}{Springer-Verlag New York,
  Inc.}, \bibinfo{pages}{279--324}.
\newblock


\bibitem[\protect\citeauthoryear{Naik, Aiken, and Whaley}{Naik
  et~al\mbox{.}}{2006}]%
        {Naik06}
\bibfield{author}{\bibinfo{person}{Mayur Naik}, \bibinfo{person}{Alex Aiken},
  {and} \bibinfo{person}{John Whaley}.} \bibinfo{year}{2006}\natexlab{}.
\newblock \showarticletitle{Effective Static Race Detection for Java}. In
  \bibinfo{booktitle}{{\em Proceedings of the 27th ACM SIGPLAN Conference on
  Programming Language Design and Implementation}} {\em (\bibinfo{series}{PLDI
  '06})}. \bibinfo{publisher}{ACM}, \bibinfo{address}{New York, NY, USA},
  \bibinfo{pages}{308--319}.
\newblock
\showISBNx{1-59593-320-4}
\showDOI{%
\url{https://doi.org/10.1145/1133981.1134018}}


\bibitem[\protect\citeauthoryear{Narayanasamy, Wang, Tigani, Edwards, and
  Calder}{Narayanasamy et~al\mbox{.}}{2007}]%
        {Narayanasamy2007}
\bibfield{author}{\bibinfo{person}{Satish Narayanasamy},
  \bibinfo{person}{Zhenghao Wang}, \bibinfo{person}{Jordan Tigani},
  \bibinfo{person}{Andrew Edwards}, {and} \bibinfo{person}{Brad Calder}.}
  \bibinfo{year}{2007}\natexlab{}.
\newblock \showarticletitle{Automatically Classifying Benign and Harmful Data
  Races Using Replay Analysis}. In \bibinfo{booktitle}{{\em Proceedings of the
  28th ACM SIGPLAN Conference on Programming Language Design and
  Implementation}} {\em (\bibinfo{series}{PLDI ’07})}.
  \bibinfo{publisher}{Association for Computing Machinery},
  \bibinfo{address}{New York, NY, USA}, \bibinfo{pages}{22–31}.
\newblock
\showISBNx{9781595936332}
\showDOI{%
\url{https://doi.org/10.1145/1250734.1250738}}


\bibitem[\protect\citeauthoryear{Netzer and Miller}{Netzer and Miller}{1990}]%
        {NetzerMiller90}
\bibfield{author}{\bibinfo{person}{Robert~H.B. Netzer} {and}
  \bibinfo{person}{Barton~P. Miller}.} \bibinfo{year}{1990}\natexlab{}.
\newblock \showarticletitle{On the Complexity of Event Ordering for
  Shared-Memory Parallel Program Executions}. In \bibinfo{booktitle}{{\em In
  Proceedings of the 1990 International Conference on Parallel Processing}}.
  \bibinfo{pages}{93--97}.
\newblock


\bibitem[\protect\citeauthoryear{Netzer and Miller}{Netzer and Miller}{1992}]%
        {NetzerMiller92}
\bibfield{author}{\bibinfo{person}{Robert H.~B. Netzer} {and}
  \bibinfo{person}{Barton~P. Miller}.} \bibinfo{year}{1992}\natexlab{}.
\newblock \showarticletitle{What Are Race Conditions? Some Issues and
  Formalizations}.
\newblock \bibinfo{journal}{{\em ACM Lett. Program. Lang. Syst.\/}}
  \bibinfo{volume}{1}, \bibinfo{number}{1} (\bibinfo{date}{March}
  \bibinfo{year}{1992}), \bibinfo{pages}{74–88}.
\newblock
\showISSN{1057-4514}
\showDOI{%
\url{https://doi.org/10.1145/130616.130623}}


\bibitem[\protect\citeauthoryear{Netzer and Miller}{Netzer and Miller}{1989}]%
        {NetzerMiller89}
\bibfield{author}{\bibinfo{person}{Robert~Netzer Netzer} {and}
  \bibinfo{person}{Barton~P. Miller}.} \bibinfo{year}{1989}\natexlab{}.
\newblock \showarticletitle{Detecting Data Races in Parallel Program
  Executions}. In \bibinfo{booktitle}{{\em In Advances in Languages and
  Compilers for Parallel Computing, 1990 Workshop}}. \bibinfo{publisher}{MIT
  Press}, \bibinfo{pages}{109--129}.
\newblock


\bibitem[\protect\citeauthoryear{Papadimitriou}{Papadimitriou}{1979}]%
        {Papadimitriou79}
\bibfield{author}{\bibinfo{person}{Christos~H. Papadimitriou}.}
  \bibinfo{year}{1979}\natexlab{}.
\newblock \showarticletitle{The Serializability of Concurrent Database
  Updates}.
\newblock \bibinfo{journal}{{\em J. ACM\/}} \bibinfo{volume}{26},
  \bibinfo{number}{4} (\bibinfo{date}{Oct.} \bibinfo{year}{1979}),
  \bibinfo{pages}{631–653}.
\newblock
\showISSN{0004-5411}
\showDOI{%
\url{https://doi.org/10.1145/322154.322158}}


\bibitem[\protect\citeauthoryear{Pavlogiannis}{Pavlogiannis}{2019}]%
        {Pavlogiannis19}
\bibfield{author}{\bibinfo{person}{Andreas Pavlogiannis}.}
  \bibinfo{year}{2019}\natexlab{}.
\newblock \showarticletitle{Fast, Sound, and Effectively Complete Dynamic Race
  Prediction}.
\newblock \bibinfo{journal}{{\em Proc. ACM Program. Lang.\/}}
  \bibinfo{volume}{4}, \bibinfo{number}{POPL}, Article
  \bibinfo{articleno}{Article 17} (\bibinfo{date}{Dec.} \bibinfo{year}{2019}),
  \bibinfo{numpages}{29}~pages.
\newblock
\showDOI{%
\url{https://doi.org/10.1145/3371085}}


\bibitem[\protect\citeauthoryear{Pozniansky and Schuster}{Pozniansky and
  Schuster}{2003}]%
        {Pozniansky03}
\bibfield{author}{\bibinfo{person}{Eli Pozniansky} {and} \bibinfo{person}{Assaf
  Schuster}.} \bibinfo{year}{2003}\natexlab{}.
\newblock \showarticletitle{Efficient On-the-fly Data Race Detection in
  Multithreaded C++ Programs}.
\newblock \bibinfo{journal}{{\em SIGPLAN Not.\/}} \bibinfo{volume}{38},
  \bibinfo{number}{10} (\bibinfo{date}{June} \bibinfo{year}{2003}),
  \bibinfo{pages}{179--190}.
\newblock
\showISSN{0362-1340}
\showDOI{%
\url{https://doi.org/10.1145/966049.781529}}


\bibitem[\protect\citeauthoryear{Pratikakis, Foster, and Hicks}{Pratikakis
  et~al\mbox{.}}{2011}]%
        {Pratikakis11}
\bibfield{author}{\bibinfo{person}{Polyvios Pratikakis},
  \bibinfo{person}{Jeffrey~S. Foster}, {and} \bibinfo{person}{Michael Hicks}.}
  \bibinfo{year}{2011}\natexlab{}.
\newblock \showarticletitle{LOCKSMITH: Practical Static Race Detection for C}.
\newblock \bibinfo{journal}{{\em ACM Trans. Program. Lang. Syst.\/}}
  \bibinfo{volume}{33}, \bibinfo{number}{1}, Article \bibinfo{articleno}{3}
  (\bibinfo{date}{Jan.} \bibinfo{year}{2011}), \bibinfo{numpages}{55}~pages.
\newblock
\showISSN{0164-0925}
\showDOI{%
\url{https://doi.org/10.1145/1889997.1890000}}


\bibitem[\protect\citeauthoryear{Roemer, Gen\c{c}, and Bond}{Roemer
  et~al\mbox{.}}{2018}]%
        {Roemer18}
\bibfield{author}{\bibinfo{person}{Jake Roemer}, \bibinfo{person}{Kaan
  Gen\c{c}}, {and} \bibinfo{person}{Michael~D. Bond}.}
  \bibinfo{year}{2018}\natexlab{}.
\newblock \showarticletitle{High-coverage, Unbounded Sound Predictive Race
  Detection}. In \bibinfo{booktitle}{{\em Proceedings of the 39th ACM SIGPLAN
  Conference on Programming Language Design and Implementation}} {\em
  (\bibinfo{series}{PLDI 2018})}. \bibinfo{publisher}{ACM},
  \bibinfo{address}{New York, NY, USA}, \bibinfo{pages}{374--389}.
\newblock
\showISBNx{978-1-4503-5698-5}
\showDOI{%
\url{https://doi.org/10.1145/3192366.3192385}}


\bibitem[\protect\citeauthoryear{Said, Wang, Yang, and Sakallah}{Said
  et~al\mbox{.}}{2011}]%
        {Said11}
\bibfield{author}{\bibinfo{person}{Mahmoud Said}, \bibinfo{person}{Chao Wang},
  \bibinfo{person}{Zijiang Yang}, {and} \bibinfo{person}{Karem Sakallah}.}
  \bibinfo{year}{2011}\natexlab{}.
\newblock \showarticletitle{Generating Data Race Witnesses by an SMT-based
  Analysis}. In \bibinfo{booktitle}{{\em Proceedings of the Third International
  Conference on NASA Formal Methods}} {\em (\bibinfo{series}{NFM'11})}.
  \bibinfo{publisher}{Springer-Verlag}, \bibinfo{address}{Berlin, Heidelberg},
  \bibinfo{pages}{313--327}.
\newblock
\showISBNx{978-3-642-20397-8}
\showURL{%
\url{http://dl.acm.org/citation.cfm?id=1986308.1986334}}


\bibitem[\protect\citeauthoryear{Savage, Burrows, Nelson, Sobalvarro, and
  Anderson}{Savage et~al\mbox{.}}{1997}]%
        {Savage97}
\bibfield{author}{\bibinfo{person}{Stefan Savage}, \bibinfo{person}{Michael
  Burrows}, \bibinfo{person}{Greg Nelson}, \bibinfo{person}{Patrick
  Sobalvarro}, {and} \bibinfo{person}{Thomas Anderson}.}
  \bibinfo{year}{1997}\natexlab{}.
\newblock \showarticletitle{Eraser: A Dynamic Data Race Detector for
  Multithreaded Programs}.
\newblock \bibinfo{journal}{{\em ACM Trans. Comput. Syst.\/}}
  \bibinfo{volume}{15}, \bibinfo{number}{4} (\bibinfo{date}{Nov.}
  \bibinfo{year}{1997}), \bibinfo{pages}{391--411}.
\newblock
\showISSN{0734-2071}
\showDOI{%
\url{https://doi.org/10.1145/265924.265927}}


\bibitem[\protect\citeauthoryear{Sen, Ro{\c{s}}u, and Agha}{Sen
  et~al\mbox{.}}{2005}]%
        {sen2005}
\bibfield{author}{\bibinfo{person}{Koushik Sen}, \bibinfo{person}{Grigore
  Ro{\c{s}}u}, {and} \bibinfo{person}{Gul Agha}.}
  \bibinfo{year}{2005}\natexlab{}.
\newblock \showarticletitle{Detecting Errors in Multithreaded Programs by
  Generalized Predictive Analysis of Executions}. In \bibinfo{booktitle}{{\em
  Formal Methods for Open Object-Based Distributed Systems}},
  \bibfield{editor}{\bibinfo{person}{Martin Steffen} {and}
  \bibinfo{person}{Gianluigi Zavattaro}} (Eds.). \bibinfo{publisher}{Springer
  Berlin Heidelberg}, \bibinfo{address}{Berlin, Heidelberg},
  \bibinfo{pages}{211--226}.
\newblock


\bibitem[\protect\citeauthoryear{Serebryany and Iskhodzhanov}{Serebryany and
  Iskhodzhanov}{2009}]%
        {tsan2009}
\bibfield{author}{\bibinfo{person}{Konstantin Serebryany} {and}
  \bibinfo{person}{Timur Iskhodzhanov}.} \bibinfo{year}{2009}\natexlab{}.
\newblock \showarticletitle{ThreadSanitizer: Data Race Detection in Practice}.
  In \bibinfo{booktitle}{{\em Proceedings of the Workshop on Binary
  Instrumentation and Applications}} {\em (\bibinfo{series}{WBIA ’09})}.
  \bibinfo{publisher}{Association for Computing Machinery},
  \bibinfo{address}{New York, NY, USA}, \bibinfo{pages}{62–71}.
\newblock
\showISBNx{9781605587936}
\showDOI{%
\url{https://doi.org/10.1145/1791194.1791203}}


\bibitem[\protect\citeauthoryear{Shasha and Snir}{Shasha and Snir}{1988}]%
        {Shasha88}
\bibfield{author}{\bibinfo{person}{Dennis Shasha} {and} \bibinfo{person}{Marc
  Snir}.} \bibinfo{year}{1988}\natexlab{}.
\newblock \showarticletitle{Efficient and Correct Execution of Parallel
  Programs That Share Memory}.
\newblock \bibinfo{journal}{{\em ACM Trans. Program. Lang. Syst.\/}}
  \bibinfo{volume}{10}, \bibinfo{number}{2} (\bibinfo{date}{April}
  \bibinfo{year}{1988}), \bibinfo{pages}{282--312}.
\newblock
\showISSN{0164-0925}
\showDOI{%
\url{https://doi.org/10.1145/42190.42277}}


\bibitem[\protect\citeauthoryear{Smaragdakis, Evans, Sadowski, Yi, and
  Flanagan}{Smaragdakis et~al\mbox{.}}{2012}]%
        {Smaragdakis12}
\bibfield{author}{\bibinfo{person}{Yannis Smaragdakis}, \bibinfo{person}{Jacob
  Evans}, \bibinfo{person}{Caitlin Sadowski}, \bibinfo{person}{Jaeheon Yi},
  {and} \bibinfo{person}{Cormac Flanagan}.} \bibinfo{year}{2012}\natexlab{}.
\newblock \showarticletitle{Sound Predictive Race Detection in Polynomial
  Time}. In \bibinfo{booktitle}{{\em Proceedings of the 39th Annual ACM
  SIGPLAN-SIGACT Symposium on Principles of Programming Languages}} {\em
  (\bibinfo{series}{POPL '12})}. \bibinfo{publisher}{ACM},
  \bibinfo{address}{New York, NY, USA}, \bibinfo{pages}{387--400}.
\newblock
\showISBNx{978-1-4503-1083-3}
\showDOI{%
\url{https://doi.org/10.1145/2103656.2103702}}


\bibitem[\protect\citeauthoryear{Wang, Kundu, Ganai, and Gupta}{Wang
  et~al\mbox{.}}{2009}]%
        {Wang2009}
\bibfield{author}{\bibinfo{person}{Chao Wang}, \bibinfo{person}{Sudipta Kundu},
  \bibinfo{person}{Malay Ganai}, {and} \bibinfo{person}{Aarti Gupta}.}
  \bibinfo{year}{2009}\natexlab{}.
\newblock \showarticletitle{Symbolic Predictive Analysis for Concurrent
  Programs}. In \bibinfo{booktitle}{{\em Proceedings of the 2Nd World Congress
  on Formal Methods}} {\em (\bibinfo{series}{FM '09})}.
  \bibinfo{publisher}{Springer-Verlag}, \bibinfo{address}{Berlin, Heidelberg},
  \bibinfo{pages}{256--272}.
\newblock
\showISBNx{978-3-642-05088-6}
\showDOI{%
\url{https://doi.org/10.1007/978-3-642-05089-3_17}}


\bibitem[\protect\citeauthoryear{Williams}{Williams}{2005}]%
        {Williams05}
\bibfield{author}{\bibinfo{person}{Ryan Williams}.}
  \bibinfo{year}{2005}\natexlab{}.
\newblock \showarticletitle{A New Algorithm for Optimal 2-constraint
  Satisfaction and Its Implications}.
\newblock \bibinfo{journal}{{\em Theor. Comput. Sci.\/}} \bibinfo{volume}{348},
  \bibinfo{number}{2} (\bibinfo{date}{Dec.} \bibinfo{year}{2005}),
  \bibinfo{pages}{357--365}.
\newblock
\showISSN{0304-3975}
\showDOI{%
\url{https://doi.org/10.1016/j.tcs.2005.09.023}}


\bibitem[\protect\citeauthoryear{Zhivich and Cunningham}{Zhivich and
  Cunningham}{2009}]%
        {SoftwareErrors2009}
\bibfield{author}{\bibinfo{person}{M. Zhivich} {and} \bibinfo{person}{R.~K.
  Cunningham}.} \bibinfo{year}{2009}\natexlab{}.
\newblock \showarticletitle{The Real Cost of Software Errors}.
\newblock \bibinfo{journal}{{\em IEEE Security and Privacy\/}}
  \bibinfo{volume}{7}, \bibinfo{number}{2} (\bibinfo{date}{March}
  \bibinfo{year}{2009}), \bibinfo{pages}{87–90}.
\newblock
\showISSN{1540-7993}
\showDOI{%
\url{https://doi.org/10.1109/MSP.2009.56}}


\end{thebibliography}

\clearpage
\appendix

\section{Details of {\cref{sec:preliminaries}}}\label{sec:proofs_preliminaries}

In this section we provide the proof of \cref{lem:decision_given_pair}.

\smallskip
\lemdecisiongivenpair*
\begin{proof}
We outline the construction of $\Trace'$.
We introduce two locks $\ell_1, \ell_2$, and for each $i\in [2]$, we surround  $\Event_i$ with the lock $\ell_i$,
i.e., we replace $\Event_i$ with $\Acquire(\ell_i), \Event_i, \Release(\ell_i)$.
For every other event $\Event\in \Events{\Trace} \cap \WritesReads{\Events{\Trace}}\setminus\{ \Event_1, \Event_2 \}$,
we replace $\Event$ with the following sequence
\[
\Acquire(\ell_1), \Acquire(\ell_2), \Event, \Release(\ell_2), \Release(\ell_1)\ .
\] 
Observe that the resulting sequence $\Trace'$ is a valid trace.
It is easy to see that $(\Event_1, \Event_2)$ 
can be the only predictable data race of $\Trace'$,
and any correct reordering of $\Trace$ that witnesses the data race $(\Event_1, \Event_2)$
can be transformed to a correct reordering of $\Trace'$ that witnesses the same data race, and vice versa.

The desired result follows.
\end{proof}


\section{Details of {\cref{sec:ideals}}}\label{sec:proofs_ideals}

Here we prove \cref{lem:candidate_set} and \cref{lem:num_candidate_ideals}.

\smallskip
\lemcandidateset*
\begin{proof}
The ($\Leftarrow$) direction of the statement is straightforward, and here we focus on the ($\Rightarrow$) direction.
Let $\Trace^*$ be a correct reordering that witnesses $(\Event_1, \Event_2)$, and $X^*=\Events{\Trace^*}$.
We show that there exists an ideal $X\in \CandidateSet_{\Trace}(\Event_1, \Event_2)$ such that 
(i)~$X\subseteq X^*$, and
(ii)~$\OpenAcquires(X)\subseteq \OpenAcquires(X^*)$.
Observe that the two conditions imply the lemma:
(i)~since $X\subseteq X^*$, we have that $\Event_1, \Event_2\not \in X$, while both events are enabled in $X$, and
(ii)~since $X\subseteq X^*$ and $\OpenAcquires(X)\subseteq \OpenAcquires(X^*)$, we have that $\Trace^*\Project X$ is a correct reordering of $\Trace$, and hence $X$ is realizable.

Consider any ideal $Y\in \CandidateSet_{\Trace}(\Event_1, \Event_2)$ such that $Y\subseteq X^*$.
Clearly, at least one such $Y$ exists, by taking $Y=\Past_{\Trace}(\{\Event_1, \Event_2\})$ and noticing that $Y\subseteq X^*$.
If $\OpenAcquires(Y)\not \subseteq \OpenAcquires(X^*)$, there exists some lock-acquire event $\Acquire\in \OpenAcquires(Y)$ such that $\Release\in X^*$, where $\Release=\Match{\Trace}{\Acquire}$.
But then $Y'\in \CandidateSet_{\Trace}(\Event_1, \Event_2)$, where $Y'=Y\cup \Past_{\Trace}(\{ \Release \} )\cup \{\Release\}$.
Note that $Y'\subseteq X^*$, and repeat the process for $Y=Y'$.
Since $Y\subset Y'$, this process can be repeated at most $n$ times, thus at some point we have chosen an ideal $Y\in \CandidateSet_{\Trace}(\Event_1, \Event_2)$ with the desired properties of $X$.

The desired result follows.
\end{proof}

\smallskip
\lemnumcandidateideals*
\begin{proof}
Let $Z=\CandidateSet_{\Trace}(\Event_1, \Event_2)$ and $X=\Past_{\Trace}(\{\Event_1, \Event_2\})$.
For an ideal $Y\in Z\setminus\{ X \}$, let $\Release_Y$ be the lock-release event that lead to $Y\in Z$ according to \cref{item:candidate_set2} of the definition of $\CandidateSet_{\Trace}(\Event_1, \Event_2)$, and $\Acquire_Y=\Match{\Trace}{\Release_{Y}}$.
In addition, we call the ideal $Y'\in Z$ with $\Acquire_Y\in \OpenAcquires(Y')$ that lead to adding $Y\in Z$
the parent of $Y$.
We define inductively 
$A_X=\emptyset$, and 
$A_Y=A_{Y'}\cup \{ \Acquire_Y \}$, where $Y'$ is the parent of $Y$.
Note that every ideal $Y\in Z$ is uniquely characterized by $A_Y$.

Let $G_{\Trace}$ be the lock-dependence graph of $\Trace$.
We show by induction that for every $Y\in Z\setminus \{ X \}$ there exists a lock-acquire event $\Acquire\in \OpenAcquires(X)$ such that $\Acquire$ is reachable from $\Acquire_Y$ in $G_{\Trace}$.
Let $Y_1$ be the smallest (wrt set inclusion) ancestor of $Y$ such that $\Acquire_Y\in \OpenAcquires(Y_1)$.
The statement holds if $Y_1=X$, by taking $\Acquire=\Acquire_Y$.
Otherwise, let $Y_2$ the parent of $Y_1$, and we have $\Acquire_{Y_1}\in \OpenAcquires(Y_2)$.
Note that
(i)~$\Acquire_Y\not <_{\TOO} \Acquire_{Y_1}$ (since $\Acquire_Y\not \in Y_2$), 
(ii)~$\Acquire_Y<_{\TOO} \Release_{Y_1}$ (since $\Acquire_Y\in \Past_{\Trace}(\Release_{Y_1})$), and
(iii)~$\Release_Y\not<_{\TOO} \Release_{Y_1}$ (since $\Release_Y\not \in Y_1$).
It follows that $(\Acquire_Y, \Acquire_{Y_1})$ is an edge in $G_{\Trace}$.
By the induction hypothesis, we have that there exists some lock-acquire event $\Acquire\in \OpenAcquires(X)$ that is reachable from $\Acquire_{Y_1}$ in $G_{\Trace}$.
Hence $\Acquire$ is reachable by $\Acquire_Y$ in $G_{\Trace}$, as desired.

Now, let $A=\bigcup_{Y\in Z} A_Y$, and by the previous paragraph, for every lock-acquire event $\Acquire'\in A$, there exists a lock-acquire event $\Acquire\in \OpenAcquires(X)$ that is reachable from $\Acquire'$ in $G_{\Trace}$.
Hence, $|A|\leq \OpenAcquires(X) \cdot \LockFactor$.
In addition, since there are $k$ threads and the lock-nesting depth is $\NestingDepth$, we have  $|\OpenAcquires(X)|\leq k\cdot \NestingDepth$,
thus $|A|\leq k\cdot \NestingDepth \cdot \LockFactor=\alpha$.
Moreover, we trivially have $|A|\leq n$.

Finally, since 
(i)~we have $k$ threads, and 
(ii)~in each ideal $Y\in Z$ the events $\Event_1$  and $\Event_2$ are enabled, 
we have $|Z|\leq |A|^{k-2}\leq \min(n,\alpha)^{k-2}$ such ideals.

The desired result follows.
\end{proof}


\section{Details of {\cref{sec:general}}}\label{sec:proofs_general}

\subsection{Details of {\cref{subsec:constant_threads}}}\label{subsec:proofs_constant_threads}

Here we provide formal proofs to \cref{lem:ideal_tree}, \cref{lem:idealgraph_size} and, using these, \cref{lem:oposet_realizability}.

\smallskip
\lemidealtree*
\begin{proof}
We prove each direction separately.

\noindent{\em ($\Rightarrow$).}
We prove by induction that every canonical trace $\Trace_Y$ of $\Tree_{\OPoset}$ realizes $Y$.
The statement clearly holds if $Y=\epsilon$.
Otherwise, let $\Trace_{Y}=\Trace_{Y'} \circ \Sequence{\Event}$, i.e., $Y$ is the extension of $Y'$ by the event $\Event$, and by the induction hypothesis we have that $\Trace_{Y'}$ realizes $Y'$.
The statement clearly holds if $\Event\in \WritesAcquires{\Trace}$, so we focus on the case where $\Event\in \ReadsReleases{\Trace}$.
Consider the pair $(\Write, \Event)\in \ReadPairs{X}$, and note that $\Write <_{\TOO}\Event$ and thus $\Write \in Y'$.
It remains to argue that for every triplet $(\Write, \Event, \Write')\in \ReadTriplets{\OPoset}$, 
if $\Write'\in Y'$ then $\Write'<_{\Trace_{Y'}} \Write$.
Assume towards contradiction otherwise, thus there exist two ancestors $Y_1,Y_2\in V^\Tree_{\OPoset}$ of $Y'$ such that
(i)~$Y_2=Y_1\cup\{\Write'\}$ and
(ii)~$\Write\in Y_1$.
In that case we have $(\Write, \Read)\in \Frontier{\OPoset}{Y_1}$, hence $\Write'$ could not have been executable in $Y_1$, a contradiction.

\noindent{\em ($\Leftarrow$).}
Let $\Trace^*$ be a witness trace that realizes $\OPoset$.
We argue that for every prefix $\Trace'$ of $\Trace^*$, there exists a node $X'\in \TreeNodes_{\OPoset}$ such that $X'=\Events{\Trace'}$.
The proof is by induction on the prefixes $\Trace'$.
The statement clearly holds if $\Trace'=\epsilon$, by taking $X'$ to be the root of $\Tree_{\OPoset}$.
Otherwise, let $\Trace'=\Trace'', \Event$, for some event $\Event\in X$,
and by the induction hypothesis, there exists $X''\in\TreeNodes_{\OPoset}$ such that $X''=\Events{\Trace''}$.
It suffices to argue that $\Event$ extends $X''$.
The statement is trivial if $\Event\in \ReadsReleases{X}$, hence we focus on the case where $\Event\in \WritesAcquires{X}$.
It suffices to argue that for every pair $(\Write, \Read)\in \Frontier{\OPoset}{X''}$ we have $(\Write, \Read, \Event)\not\in \ReadTriplets{\OPoset}$.
Consider any such triplet, and since $\Trace^*$ is a witness of the realizability of $X$, we have $\Read\in \Events{\Trace^*}$.
By the induction hypothesis, we have $\Write\in \Events{\Trace''}$ and $\Read\not \in \Events{\Trace''}$.
But then $\Observation_{\Trace^*}(\Read)\neq \Write$, a contradiction.

The desired result follows.
\end{proof}

\smallskip
\lemidealgraphsize*
\begin{proof}
The statement follows directly from the definition of poset ideals and the fact that the poset $(X,P)$ has width at most $k$.
\end{proof}

We are now ready to prove \cref{lem:oposet_realizability}.

\smallskip
\lemoposetrealizability*
\begin{proof}
Consider an rf-poset $\OPoset$ of size $n$, $k$ threads and $d$ variables.
By \cref{lem:ideal_tree}, to decide the realizability of $\OPoset$ it suffices to construct the ideal graph $G_{\OPoset}$ and test whether $X\in V_{\OPoset}$.
By \cref{lem:idealgraph_size}, $G_{\OPoset}$ has $O(n^k)$ nodes,
and since $\OPoset$ has $k$ threads, there each node of $G_X$ has $\leq k$ outgoing edges.
Hence $G_X$ can be constructed in $O(k\cdot n^k)=O(\beta)$ time.

The desired result follows.
\end{proof}

\subsection{Details of {\cref{subsec:w1_hardness}}}\label{subsec:proofs_w1_hardness}

Here we prove formally \cref{lem:oposet_w1hardness} and, using this, \cref{them:general_w1hard}.

\smallskip
\lemoposetwhardness*
\begin{proof}
We show that $\OPoset_G$ is realizable iff $G$ has an independent set of size $c$.
We prove each direction separately.

\noindent{($\Rightarrow$)}
Let $\Trace$ be a trace that realizes $\OPoset_G$.
For each $i\in [c]$, let $m_i$ be the maximum integer $j$ 
such that the event $\Read(z_{i}^{j})$ has a predecessor $\Event_i$ in $\SeqTrace_i$
with  $\Event_i<_{\Trace}\Read(x)$.
Note that for each such $i$ we have $\Write(s_i)<_\Trace\Read(x)$, and since $\Write(s_i)$ is the predecessor of $\Read(z_{i}^{1})$ in $\SeqTrace_i$, 
the index $m_i$ is well-defined.
We argue that $A=\{l_i\}_{i \in [c]}$ is an independent set of $G$, where $l_i=m_i$ if $\Read(x)<_{\Trace}\Read(z_{i}^{m_i})$ and $l_i=m_i+1$ otherwise.
We note that in the second case, $l_i=n$, as, otherwise, we would also have $\Acquire_{l_i}(\ell_i)<_{\Trace} \Read(x)$,
and since $\Read(x)$ is protected by lock $\ell_i$, we would also have that $\Release_{l_i}(\ell_i)<_{\Trace} \Read(x)$.
But then, $\Write(y_1^{l_1+1})<_{\Trace} \Read(x)$, which would contradict our choice of $m_i$.

Indeed, consider any distinct $i_1,i_2\in [c]$ and any neighbor $v_1$ and $v_2$ of $m_{i_1}$ and $m_{i_2}$, respectively.
Note that our choice of $A = \set{l_{i}}_{i\in[c]}$ implies that
both
$\Acquire_{i_1}(\ell_{\{l_{i_1}, v_1\}}) <_{\Trace} \Read(x)$
and
$\Acquire_{i_2}(\ell_{\{l_{i_2}, v_2\}}) <_{\Trace} \Read(x)$.
It suffices to argue that both
$\Read(x)<_{\Trace} \Release_{i_1}(\ell_{\{l_{i_1}, v_1\}})$
and
$\Read(x)<_{\Trace} \Release_{i_2}(\ell_{\{l_{i_2}, v_2\}} )$,
as then we have that $(l_{i_1}, l_{i_2})\not \in E$, which concludes that $A$ is an independent set of $G$.

Assume towards contradiction that $\Release_{i_1}(\ell_{\{l_{i_1}, v_1\}})<_{\Trace} \Read(x)$.
Clearly $l_{i_1}< n$, as $\Read(x)<_{P} \Release_{i_1}(\ell_{\{n, v_1\}})$.
But then, $\Acquire^{l_{i_1}}(\ell_{i_1})<_{\Trace} \Release_{i_1}(\ell_{\{l_{i_1}, v_1\}})$ and thus 
$\Acquire^{l_{i_1}}(\ell_{i_1}) <_{\Trace} \Read(x)$.
Since $\Read(x)$ appears in a critical section on lock $\ell_{i_1}$, we have 
$\Release^{l_{i_1}}(\ell_1) <_{\Trace} \Read(x)$.
But then $\Event_{i_1}<_{\Trace} \Read(x)$, where $\Event_{i_1}$ is the predecessor of $\Read(z_{i}^{l_1+1})$ in $\SeqTrace_{i_1}$, 
which contradicts our definition of $l_1$.
Hence, $\Read(x)<_{\Trace} \Release_{i_1}(\ell_{\{l_{i_1}, v_1\}}), \Release_{i_2}(\ell_{\{l_{i_2}, v_2\}} )$,
and thus $A$ is an independent set of $G$.

\noindent{($\Leftarrow$)}
Let $A$ be an independent set of $G$ of size $c$, and $l_1, \dots, l_c$ some arbitrary ordering of $A$.
For each $i\in [c]$ let $Y_i$ be the following events of threads $i$ and $c+i$.
\begin{compactenum}
\item All strict predecessors of $\Event_i$ in $\SeqTrace_i$, where  $\Event_i=\Read(z_i^{l_i})$ if $l_i<n$ and $\Event_i=\Read_i(x)$ otherwise.
\item If $l_i>1$, all predecessors of $\Release^{l_i-1}(\ell_i)$ in $\SeqTrace_{c+i}$ (including $\Release^{l_i-1}(\ell_i)$). 
\end{compactenum}
For each $i\in [c]$, let 
\[
Z_i = \setpred{\Event\in Y_i}{\exists \Acquire\in \OpenAcquires(Y_i) \text{ s.t. } \Acquire\leq_{P} \Event}\ ,
\]
i.e., $Z_i$ contains all events of $\SeqTrace_i$ and $\SeqTrace_{c+i}$ that succeed some lock-acquire event that is open in $Y_i$.
We argue that for any two distinct $i_1, i_2\in [c]$, any two lock-acquire events $\Acquire_1\in Z_{i_1}$ and $\Acquire_2\in Z_{i_2}$ access a different lock.
Note that the statement follows easily if one $\Acquire_1$ or $\Acquire_2$ belongs to $\SeqTrace_{c+i}$, for some $i\in[c]$,
as the locks accessed by each such total ordered are only also accessed by $\SeqTrace_{2\cdot c + 2}$.
Hence, we focus on the case where each of $\Acquire_1$ and $\Acquire_2$ belong to $\SeqTrace_{i}$, for some $i\in[c]$.
Assume towards contradiction otherwise, hence there exist distinct $i_1, i_2\in [c]$ such that
(i)~$\Acquire_{i_1}(\ell_{\{ l_{i_1}, l_{i_2} \}})\in Z_{i_1}$ and $\Acquire_{i_2}(\ell_{\{ l_{i_1}, l_{i_2} \}})\in Z_{i_2}$, while
(ii)~$\Release_{i_1}(\ell_{\{ l_{i_1}, l_{i_2} \}})\not \in Z_{i_1}$ and $\Release_{i_2}(\ell_{\{ l_{i_1}, l_{i_2} \}})\not \in Z_{i_2}$
It follows that $(l_{i_1}, l_{i_2})\in E$, which contradicts the fact that $A$ is an independent set of $G$.

We now construct a trace $\Trace$ that realizes $\OPoset_G$ in five phases, where initially, we have $\Trace=\epsilon$.
\begin{compactenum}
\item\label{item:phase1} {\em Phase~1:} For each $i\in [c]$, we linearize the partial order $P\Project(Y_i\setminus Z_i)$ arbitrarily, and append it to $\Trace$.
\item\label{item:phase2} {\em Phase~2:} For each $i\in [c]$, we linearize the partial order $P\Project Z_i$ arbitrarily, and append it to $\Trace$.
\item\label{item:phase3} {\em Phase~3:} We append to $\Trace$ the sequence $\Trace_1\Concat \Write(x), \Read(x)\Concat \Trace_2$,
where $\Trace_1$ and $\Trace_2$ are sequences over the events of $\SeqTrace_{2\cdot c + 2}$, as follows.
\begin{align*}
\Trace_1&=\Read(s_1),\dots, \Read(s_c), \Acquire(\ell_1),\dots \Acquire(\ell_c)
\qquad \text{and}\\
\Trace_2&= \Release(\ell_c),\dots, \Release(\ell_1)
\end{align*}
\item\label{item:phase4} {\em Phase~4:} For each $i\in [c]$, we linearize the partial order $P\Project(S_i\setminus Y_i)$, where $S_i$ is the smallest ideal of $(X,P)$ that contains the matching lock-release events of all lock-acquire events that are open in $Y_i$
\item\label{item:phase5} {\em Phase~5:} For each $i\in [c]$ we linearize $P$ over the remaining events of $\SeqTrace_i$ and $\SeqTrace_{c+i}$ arbitrarily, and append them to $\Trace$.
\end{compactenum}

Finally, we argue that $\Trace$ is a valid witness trace.
It is straightforward to verify that $\Trace$ is a linearization of $P$.
Moreover, since every memory location is written exactly once in $X$,
for every read event $\Read \in \Reads{X}$ we have $\Observation_{\Trace}(\Read)=\Observation(\Read)$.
It remains to argue that $\Trace$ respects the critical sections of $\OPoset_G$.
Observe that the only phase in which we interleave open critical sections between total orders $\SeqTrace_i$ that access the same lock is in Phase~\ref{item:phase2}.
However, as we have shown, for every  two distinct $i_1, i_2\in [c]$, any two lock-acquire events $\Acquire_1\in Z_{i_1}$ and $\Acquire_2\in Z_{i_2}$ access a different lock.
It follows that $\Trace$ respects the critical sections of $\OPoset_G$.

The desired result follows.
\end{proof}

We are now ready to prove \cref{them:general_w1hard}.

\smallskip
\themgeneralwhard*
\begin{proof}
We show that $(\Write(x), \Read(x))$ is a predictable data race of $\Trace$ iff $\OPoset_G$ is realizable.
If $\OPoset_G$ is realizable, then $(\Write(x), \Read(x))$ is a predictable data-race of $\Trace$ witnessed by 
the witness $\Trace^*$ of the realizability of $\OPoset_G$, as constructed in the proof of \cref{lem:oposet_w1hardness} (direction $\Leftarrow$), restricted to the events $\{\Event\in X\colon \Event<_{\Trace^*} \Write(x) \}$.

For the inverse direction, let $\Trace^*$ be a correct reordering of $\Trace$ that witnesses the data race $(\Write(x), \Read(x))$.
We construct a trace $\Trace'$ that realizes $\OPoset_G$ as 
\[
\Trace'=\Trace^*\Concat \Write(x), \Read(x), \Release(\ell_c), \dots, \Release(\ell_1) \Concat \Trace_4 \Concat \Trace_5\ ,
\]
where $\Trace_4$ and $\Trace_5$ are analogous to the linearizations in phase~\ref{item:phase4} and Phase~\ref{item:phase5}, respectively,
of the construction in the proof of \cref{lem:oposet_w1hardness}.
In particular, let $\Trace''$ be the prefix of $\Trace'$ until the event $\Release(\ell_1)$.
\begin{compactenum}
\item We construct $\Trace_4$ as follows.
For each $i\in [c]$, we linearize the set $S_i\setminus \Events{\Trace''}$ and append it to $\Trace_4$, where $S_i$ is the smallest ideal of $X$ that contains the matching lock-release events of all lock-acquire events that are open in $\Trace''$.
It is easy to see that each $S_i$ contains the event $\Read_i(x)$, hence $\Trace'$ respects the orderings $\Read(x)<_{P} \Read_i(x)$.
Indeed, as $\Observation(\Read_i(x))=\Write(x)$ and $\Write(x)\not \in \Trace^*$, we have that $\Read_i(x)\not \in \Trace^*$, and hence by definition, $\Read_i(x)\in S_i$.
\item We construct $\Trace_5$ as follows.
For each $i\in [c]$ we linearize the remaining events of $\SeqTrace_i$ and $\SeqTrace_{c+i}$ arbitrarily, and append them to $\Trace_5$.
\end{compactenum}
The correctness is established similarly to the proof of \cref{lem:oposet_w1hardness}.

The desired result follows.
\end{proof}


\section{Details of {\cref{sec:trees}}}\label{sec:proofs_trees}

\subsection{Details of {\cref{subsec:tree_topologies_ub}}}\label{subsec:proofs_tree_topologies_ub}

Here we provide details of \cref{subsec:tree_topologies_ub}.
We start with the proof of \cref{lem:tree_topologies}.
First, consider the construction of witness trace $\Trace$, by linearizing the poset $(X,Q)$.
We have the following lemma, which states that $(X,Q)$ is well-defined.

\smallskip
\begin{restatable}{lemma}{lemqpo}\label{lem:q_po}
$(X,Q)$ is a poset.
\end{restatable}
\begin{proof}
Assume towards contradiction otherwise.
Consider the process in which we insert the orderings in $Q$ in sequence, and examine the first pair of conflicting events $(\Event_1, \Event_2)$ such that we try to order $\Event_1<_{Q}\Event_2$ whereas it already holds that $\Event_2<_{Q}\Event_1$.
For $j\in [2]$, let $i_j$ be such that $\Event_j\in X_{i_j}$, and it follows that $i_1$ is the parent of $i_2$ in $\Tree$.
Observe that, by the construction of $Q$, we have $\Unordered{\Event_1}{P}{\Event_2}$.
Consider the sequence of orderings 
\[
\Event_2=\ov{\Event}_1\Prec_{Q} \ov{\Event}_2 \Prec_{Q} \dots \Prec_{Q} \ov{\Event}_{\ell}=\Event_1
\]
that witnesses (transitively) that $\Event_2<_{Q}\Event_1$.
Since $\OPoset$ is tree-inducible, there exists some $j\in[\ell]$ such that $\ov{\Event}_{j}\in X_{i_1}$ and $\Event_2<_{P}\ov{\Event}_{j}$.
Since $X_{i_1}$ is totally ordered in $P$, the events $\ov{\Event}_j$ and $\Event_1$ are ordered in $P$.
Clearly $\Event_1<_{P}\ov{\Event}_j$, otherwise we would have $\Event_2<_{P} \Event_1$.
But then we already have $\ov{\Event}_j<_{Q}\Event_1$, and hence a cycle already exists in $Q$, 
contradicting our assumption about $(e_1, e_2)$ being the first pair where a cycle is encountered.
The desired result follows.
\end{proof}

As the above lemma establishes that $\Trace$ is well-defined,
we can prove \cref{lem:tree_topologies} by showing that $\Trace$ is a witness of the realizability of $\OPoset$.

\smallskip
\lemtreetopologies*
\begin{proof}
Consider any  triplet $(\Write, \Read, \Write')\in \ReadTriplets{\OPoset}$ and we argue that
(i)~$\Write<_{\Trace}\Read$ and
(ii)~if $\Write<_{\Trace} \Write'$ then $\Read<_{\Trace}\Write'$.
For (i), by the definition of rf-posets, we have $\Write<_{P}\Read$ and thus $\Write<_{\Trace}\Read$.
We now turn our attention to (ii).
If $\Ordered{\Write'}{P}{\Read}$ or $\Ordered{\Write'}{P}{\Write}$, since $\OPoset$ is closed, we have either $\Write'<_{P} \Write$ or $\Read<_{P}\Write'$, and hence $\Write'<_{Q} \Write$ or $\Read<_{Q}\Write'$.
Since $\Trace$ is a linearization of $Q$, the first case leads to a contradiction, whereas the second case leads to $\Read<_{\Trace}\Write'$, as desired.
Now assume that $\Unordered{\Write'}{P}{\Read}$ and $\Unordered{\Write'}{P}{\Write}$, and let $i\in[k]$ be such that $\Read\in X_i$.
Since $\OPoset$ is tree-inducible, it follows that $\Write\in X_i$ and $\Write'\in X_j$ for some $i\neq j$, and such that $(i,j)$ is an edge of $\Tree$.
Since $\Write<_{\Trace}\Write'$, we have $\Write<_{Q}\Write'$, and thus, by construction, t $j$ is a child of $i$ in $\Tree$.
Again, by construction, we have  $\Read<_{Q} \Write$. 
Since $\Trace$ is a linearization of $Q$, we have that $\Observation_{\Trace}(\Read)\neq \Write'$.

The desired result follows.
\end{proof}

Next, we prove \cref{lem:closure_tree}, which states that tree-inducibility is preserved under taking closures.

\smallskip
\lemclosuretree*
\begin{proof}
Let $\Tree=([k], \setpred{(i,j)}{\Confl{X_i}{X_j}})$ be a rooted tree such that $\OPoset$ is tree-inducible to $\Tree$,
and we argue that $\OPosetQ$ is also tree-inducible to $\Tree$.
Since $Q\Refines P$, clearly the condition~\ref{item:tree_ind1} and condition~\ref{item:tree_ind2} of tree-inducibility are met.
Condition~\ref{item:tree_ind3} follows from the fact that
(i)~$Q$ is the transitive closure of $P\cup Q'$, where $Q'$ is a relation between conflicting events, and
(ii)~since $\Tree$ is a tree, for any two conflicting events $\Event_1, \Event_2$ we have either $\Event_1, \Event_2\in X_i$ for some $i\in[k]$, or $\Event_1\in X_i$ and $\Event_2\in X_j$ for some $i,j\in[k]$ such that $(i,j)$ is an edge of $\Tree$.

The desired result follows.
\end{proof}

Now we can conclude the proof of \cref{lem:oposet_realizability_trees}, which is step~(i) towards the proof of \cref{them:tree_topologies_ub}.

\smallskip
\lemoposetrealizabilitytrees*
\begin{proof}
Consider a tree-inducible rf-poset $\OPoset$.
In $O(k^2\cdot d\cdot n^2\cdot \log n)$ time, we can decide whether the closure of $\OPoset=(X,P,\Observation)$ exists~\cite{Pavlogiannis19}.
By \cref{rem:closure}, if the closure does not exist, $\OPoset$ is not realizable.
On the other hand, if the closure exists, denote it by $\OPosetQ=(X,Q,\Observation)$.
By \cref{lem:closure_tree}, $\OPosetQ$ is tree-inducible, and by \cref{lem:tree_topologies}, it is realizable by a witness $\Trace^*$.
Since $Q\Refines P$, we have that $\Trace^*\Refines P$, and thus $\Trace^*$ also realizes $\OPoset$.

The desired result follows.
\end{proof}

We now turn our attention to \cref{lem:tree_ideal} which forms step~(ii) towards \cref{them:tree_topologies_ub}.
We start with the following lemma, which, in high level,
guarantees that if $\Trace^*$ is a trace in which $\Event$ is enabled,
then $\Trace^*\Project \LPast_{\Trace}(\Event)$ also has this property.

\smallskip
\begin{restatable}{lemma}{lemlpastevents}\label{lem:lpast_events}
Let $X=\LPast_{\Trace}(\Event)$, and $\Trace^*$ be any correct reordering of $\Trace$ in which $\Event$ is enabled.
The following assertions hold.
\begin{compactenum}
\item\label{item:lpast_events1} $X\subseteq \Events{\Trace^*}$.
\item\label{item:lpast_events2} Consider any thread $\Process_1\neq \Proc{\Event}$, and thread $\Process_2$ such that
(i)~$X\Project \Process_1 \neq \emptyset$, and
(ii)~$\Process_2$ is a parent of $\Process_1$ in the tree topology $G_{\Trace}$ rooted at $\Proc{\Event}$.
Let $\Event_2$ be the unique maximal event of $X\Project \Process_2$ in $\Trace^*$,
and $\Event_1$ any event in $X\Project \Process_1$.
We have that $\Event_1<_{\Trace^*} \Event_2$.
\end{compactenum}
\end{restatable}
\begin{proof}
The proof is by induction on the steps of the process that constructs $\LPast_{\Trace}(\Event)$.
Clearly, both statements hold after \cref{item:lpast1} of the process has been executed.
Similarly, both statements hold easily after each time \cref{item:lpast2}  of the process has been executed.
We now proceed with \cref{item:lpast3} of the process.
For each $i\in[2]$, consider the lock-acquire events $\Acquire_i\in X$ as identified in this step, and $\Release_i=\Match{\Trace}{\Acquire_i}$.
Observe that $\Event_2$ appears in the critical section of $\Acquire_2$, while by the induction hypothesis, we have $\Acquire_1,\Acquire_2\in \Events{\Trace^*}$.
We distinguish the step that led to $\Acquire_1\in X$.
\begin{compactenum}
\item[Step~\ref{item:lpast2}.] 
Then $\Acquire_1<_{\TOO} \Event_2$, and thus
(i)~all predecessors of $\Release_1$ (including $\Release_1$) appear in $\Trace^*$, and 
(ii)~$\Release_1<_{\Trace*} \Acquire_2$, hence $\Release_1<_{\Trace^*} \Event_2$,
and thus $\Event_1<_{\Trace^*}\Event_2$ for all predecessors of $\Release_1$ inserted to $X$.
\item[Step~\ref{item:lpast3}.]
For each $i\in[2]$, consider the lock-acquire events $\Acquire'_i$ as identified in that step, and $\Release'_i=\Match{\Trace}{\Acquire'_i}$.
Note that $\Acquire_1<_{\TO}\Release'_1$,
and since, by the induction hypothesis, we have $\Release'_1<_{\Trace^*} \Event_2$, we obtain $\Acquire_1<_{\Trace^*}\Event_2$.
Thus, again,
(i)~all predecessors of $\Release_1$ (including $\Release_1$) appear in $\Trace^*$, and 
(ii)~$\Release_1<_{\Trace*} \Acquire_2$, hence $\Release_1<_{\Trace^*} \Event_2$,
and thus $\Event_1<_{\Trace^*}\Event_2$ for all predecessors of $\Release_1$ inserted to $X$.
\end{compactenum}

The desired result follows.
\end{proof}

Next we have a technical lemma, which will allow us to conclude that 
if $(\Event_1, \Event_2)$ is a predictable data race of $\Trace$,
the ideal $\LPast_{\Trace}(\Event_1)\cup \LPast_{\Trace}(\Event_2)$ is lock-feasible.

\smallskip
\begin{restatable}{lemma}{lemlpastdouble}\label{lem:lpast_double}
For any two conflicting events $\Event_1, \Event_2$, consider the set $X=\LPast_{\Trace}(\Event_1)\cup \LPast_{\Trace}(\Event_2)$.
For any two conflicting events $\Event'_1, \Event'_2\in X$ such that $\Proc{\Event'_i}\neq \Proc{\Event_1}$ for each $i\in [2]$,
we have that $\Event'_1, \Event'_2\in \LPast_{\Trace}(\Event_i)$, for some $i\in[2]$.
\end{restatable}
\begin{proof}
We assume that $\Proc{\Event_1}\neq \Proc{\Event_2}$, as the statement clearly holds otherwise.
Assume w.l.o.g. that $\Event'_i\in \LPast_{\Trace}(\Event_i)$ for each $i\in [2]$, and consider the tree 
$\Tree=([k], \setpred{ (i,j)}{\Confl{\Events{\Trace}\Project \Process_i}{\Events{\Trace}\Project \Process_j}})$.
Let $i_1, i_2, i'_1, i'_2\in [k]$ be such that $\Proc{\Event_1}=\Process_{i_1}$, $\Proc{\Event_2}=\Process_{i_2}$, $\Proc{\Event'_1}=\Process_{i'_1}$ and $\Proc{\Event'_2}=\Process_{i'_2}$, and since $\Confl{\Event_1}{\Event_2}$, we have that $(i_1,i_2)$ is an edge of $\Tree$.
Consider the two components $C_1, C_2$ that are created in $\Tree$ by removing the edge $(i_1, i_2)$, such that $i_1\in C_1$ and $i_2\in C_2$.
Since $\Confl{\Event'_1}{\Event'_2}$ and $i'_1, i'_2\neq i_1$,  we have that $i'_1, i'_2\in C_1$ or $i'_1, i'_2\in C_2$.
We only consider  $i'_1, i'_2\in C_2$, as the other case is similar.

Since $\Event'_1\in \LPast_{\Trace}(\Event_1)$, there exists event $\Event \in \LPast_{\Trace}(\Event_1)$ such that $\Proc{\Event}=\Process_{i_2}$ and either $\Event'_1=\Event$ or  $\Event'_1\in\LPast_{\Trace}(\Event)$.
If $\Event<_{\TO}\Event_2$, then $\LPast_{\Trace}(\Event)\subseteq  \LPast_{\Trace}(\Event_2)$ and thus $\Event'_1\in \LPast_{\Trace}(\Event_2)$.
Otherwise, $\Event_2\in \LPast_{\Trace}(\Event_1)$, hence $\LPast_{\Trace}(\Event_2)\subseteq \LPast_{\Trace}(\Event_1)$ and thus $\Event'_2\in \LPast_{\Trace}(\Event_1)$.

The desired result follows.
\end{proof}

We are now ready to prove \cref{lem:tree_ideal}.

\smallskip
\lemtreeideal*
\begin{proof}
For the $(\Leftarrow)$ direction, notice that $\Event_1$ and $\Event_2$ are enabled in $X$,
and since $X$ is realizable, we have that  $(\Event_1, \Event_2)$ is a predictable data race of $\Trace$.
We now focus on the $(\Rightarrow)$ direction.

Let $\Trace^*$ be a witness of the data race $(\Event_1, \Event_2)$.
As $\Event_1$ and $\Event_2$ are enabled in $\Trace^*$, we have $\{\Event_1, \Event_2 \}\cap \Events{\Trace^*}=\emptyset$.
By \cref{lem:lpast_events}, we have $X\subseteq \Events{\Trace^*}$, and thus $\{\Event_1, \Event_2 \}\cap X=\emptyset$.

We now argue that $\Trace'=\Trace^*\Project X$ realizes $X$.
Since $X\subseteq \Events{\Trace^*}$ and $X$ is a trace ideal, it only remains to show that $\Trace'$ respects the critical sections.
Consider any two lock-acquire events $\Acquire_1,\Acquire_2\in X$ such that $\Confl{\Acquire_1}{\Acquire_2}$ and $\Acquire_1<_{\Trace^*}\Acquire_2$, and we argue that $\Release_1\in X$.
where $\Release_1=\Match{\Trace}{\Acquire_1}$.
Let $Y=\Events{\Trace}\Project \{\Proc{\Event_1},\Proc{\Event_2}\}$, and observe that
$X\Project \{\Proc{\Event_1},\Proc{\Event_2}\} = Y$,
thus the statement is true if $\{\Proc{\Acquire_1}, \Proc{\Acquire_2}\}\subseteq \{\Proc{\Event_1}, \Proc{\Event_2} \}$.
Otherwise, we have $\Proc{\Acquire_j}\neq \Proc{\Event_i}$ for some $i\in[2]$ and each $j\in[2]$.
By \cref{lem:lpast_double}, we have that $\Acquire_1, \Acquire_2\in \LPast_{\Trace}(\Event_l)$ for some $l\in [2]$.
Assume towards contradiction that $\Release_1\not \in X$.
By the definition of $\LPast_{\Trace}(\Event_l)$, we have that $\Proc{\Acquire_1}$ is a parent of $\Proc{\Acquire_2}$ in the tree $G_{\Trace}$ rooted at $\Proc{\Event_l}$.
By \cref{lem:lpast_events}, we have $\Acquire_2<_{\Trace^*} \Event$, where $\Event$ is the last event of $X\Project \Proc{\Acquire_1}$ in $\Trace^*$. 
Note that $\Event$ belongs to the critical section of $\Acquire_1$, hence we must have $\Acquire_2<_{\Trace^*} \Acquire_1$, a contradiction.

The desired result follows.
\end{proof}

\subsection{Details of {\cref{subsec:tree_topologies_lb}}}\label{subsec:proofs_tree_topologies_lb}

Here we provide the proof of \cref{lem:closure_ov}.
We first develop some helpful notation.

Given integers $j,l\in [n/2]$, we denote by $X_{j}^{l}$ the events of $X_A$ and $X_B$ that correspond to vectors $a_j$ and $b_l$,
that is, events that have superscript $a_j$ or $b_l$, with the following exception:~if the observation $\Observation(\Read)$ of a read event $\Read$ is not in $X_{j}^{l}$ then $\Read$ is also not in $X_{j}^{l}$.
The notation carries over to partial orders $Q_{j}^{l}$ and reads-from functions $\Observation_{j}^{l}$.
We also use inequalities for the subscripts and superscripts of $X$ (e.g., $X_{j}^{\leq n/2}$) to denote events that corresponds to vectors $a_{j'}$ and $b_{l'}$ for $j'$ and $l'$ that satisfy the inequalities.

Consider a partial order $Q'$ over a subset $Y\subseteq X$ such that $Q'\Refines \SeqTrace_A\Project Y$ and $Q'\Refines \SeqTrace_B\Project Y$.
Given two events $\Event_1, \Event_2\in Y$ such that $\{\Event_1, \Event_2\}\not \subseteq X_A$ and $\{\Event_1, \Event_2\}\not \subseteq X_B$
(i.e., the events belong to different sets among $X_A$ and $X_B$) we say that a $Q'$ has a \emph{cross edge} $\Event_1<_{Q'}\Event_2$ to mean that it has the ordering $\Event_1<_{Q'}\Event_2$ and the orderings that are introduced transitively through it.

Before the final proof of \cref{lem:closure_ov}, we present some technical, but conceptually simple, lemmas.
Consider the two total orders $\SeqTrace_A$ and $\SeqTrace_B$ without any cross edges.
The first lemma reasons at the coordinate level of two vectors $a_j$ and $b_l$.
It states that if we start with a cross edge $\Write^{a_j}_1(x_1)< \Read^{b_l}_1(x_1)$,
the closure rules eventually lead to an ordering $\Read^{a_j}_1(x_2)< \Write^{b_l}_1(x_2)$ iff $a_j$ and $b_l$ are not orthogonal.

\smallskip
\begin{restatable}{lemma}{lemovproof1}\label{lem:ovproof1}
For any $j,l\in [n/2]$, consider the rf-poset $\OPosetQ_{j}^{l}=(X_{j}^{l}, Q_{j}^{l}, \Observation_{j}^{l})$,
where $Q_{j}^{l}$ has a single cross edge $\Write^{a_j}_1(x_1)<_{Q_{i}^{j}} \Read^{b_l}_1(x_1)$.
Let  $\OPosetS_{j}^{l}=(X_{j}^{l}, S_{j}^{l}, \Observation_{j}^{l})$ be the closure of $\OPosetQ_{j}^{l}$.
We have $\Read^{a_j}_1(x_2)<_{S_{j}^{l}} \Write^{b_l}_1(x_2)$ iff $a_j$ and $b_l$ are not orthogonal.
\end{restatable}
\begin{proof}[Proof (sketch)]
First, observe that since $\Write^{a_j}_1(x_1)<_{Q_{j}^{l}} \Read^{b_l}_1(x_1)$, by closure, we will have
$\Write^{a_j}_1(x_1)<_{S_{j}^{l}} \Write^{b_l}_1(x_1)$.
By a simple induction on the events on variables $x_1$ and $x_3$, we have that 
$\Write^{a_j}_i(x_1)<_{S_{j}^{l}} \Write^{b_l}_i(x_1)$ for each $i\in[D]$.
In turn, one of these orderings leads to
$\Write^{a_j}_i(x_2)<_{S_{j}^{l}}\Write^{b_l}_i(x_2)$ and thus
$\Read^{a_j}_i(x_2)<_{S_{j}^{l}}\Write^{b_l}_i(x_2)$
iff $a_j[i]\cdot b_l[i]=1$.
By a simple induction on the events on variables $x_2$ and $x_6$, we have that
$\Read^{a_j}_1(x_2)<_{S_{j}^{l}} \Write^{b_l}_1(x_2)$
iff $a_j$ and $b_l$ are not orthogonal.
\end{proof}

The next lemma reasons at the vector $B$ level.
Consider any fixed $j\in[n/2]$.
The lemma states that if we order $\Write^{a_1}_1(x_1)< \Read^{b_1}_1(x_1)$,
if $a_j$ is not orthogonal to $b_{l'}$ for any $l' \leq n/2-1$,
the closure rules eventually lead to $\Write^{a_j}_1(x_1)< \Read^{b_{n/2}}_1(x_1)$.
On the other hand, if there exists a smallest $l_1$ such that $a_j$ is orthogonal to $b_{l_1}$,
the closure rules will stop inserting cross edges between events of vectors $a_j$ and all $b_{l'}$, for $l'\geq l_1+1$.

\begin{restatable}{lemma}{lemovproof2}\label{lem:ovproof2}
For any $j\in[n/2]$ and $l\leq n/2-1$, consider the rf-poset $\OPosetQ_{j}^{\leq l+1}=(X_{j}^{\leq l+1}, Q_{j}^{\leq l+1}, \Observation_{j}^{\leq l+1})$,
where $Q_{j}^{\leq l+1}$ has a single cross edge $\Write^{a_j}_1(x_1)<_{Q_{j}^{\leq l+1}} \Read^{b_1}_1(x_1)$.
Let  $\OPosetS_{j}^{\leq l+1}=(X_{j}^{\leq l+1}, S_{j}^{\leq l+1}, \Observation_{j}^{\leq l+1})$ be the closure of $\OPosetQ_{j}^{\leq l+1}$.
The following assertions hold.
\begin{compactenum}
\item\label{item:ovproof2A} If $a_j$ is not orthogonal to $b_{l'}$ for any $l'\in [l]$, then $\Write^{a_j}_1(x_1)<_{S_{j}^{\leq l+1}} \Read^{b_{l+1}}_1(x_1)$.
\item\label{item:ovproof2B} If $a_j$ is orthogonal to $b_{l'}$, for some $l'\in [l]$, then there are no cross edges in $(X_{j}^{l_1+1\leq l'\leq l+1}, S_{j}^{l_1+1\leq l'\leq l+1})$,
where $l_1$ is the smallest $l'$ such that $a_j$ is orthogonal to $b_{l'}$.
\end{compactenum}
\end{restatable}
\begin{proof}[Proof (sketch)]
The proof is by induction on $l$.
For the base case $(l=1)$, 
consider the rf-poset $\OPosetQ_{j}^{1}$, and by \cref{lem:ovproof1}, 
we have $\Read^{a_j}_1(x_2)<_{S_{j}^{1}} \Write^{b_1}_1(x_2)$ iff $a_j$ and $b_1$ are not orthogonal.
Observe that if $\Read^{a_j}_1(x_2)<_{S_{j}^{1}} \Write^{b_1}_1(x_2)$, we also have
$\Read^{a_j}_1(x_2)<_{S_{j}^{\leq 2}} \Write^{b_1}_1(x_2)$, and thus
$\Write^{a_j}(x_4)<_{S_{j}^{\leq 2}}\Read^{b_1}(x_4)$.
Then, by closure we have $\Write^{a_j}(x_4)<_{S_{j}^{\leq 2}}\Write^{b_2}(x_4)$ and thus transitively
$\Write^{a_j}_2(x_1)<_{S_{j}^{\leq 2}} \Read^{b_2}_1(x_1)$.
On the other hand, it can be easily seen that if 
$\Read^{a_j}_1(x_2)\not <_{S_{j}^{1}} \Write^{b_1}_1(x_2)$
then the cross edges in $(X_{j}^{\leq 2}, S_{j}^{\leq 2})$ are precisely the cross edges in $(X_{j}^{\leq 1}, S_{j}^{\leq 1})$,
and thus there are no cross edges in $(X_{j}^{2},S_{j}^{2})$.

Now assume that the claim holds for $l$, and we argue that it holds for $l+1$.
By the induction hypothesis, we have that if $a_j$ is not orthogonal to $b_{l'}$ for any $l'\in [l]$, then
$\Write^{a_j}_1(x_1)<_{S_{j}^{\leq l+1}} \Read^{b_{l+1}}_1(x_1)$
and thus
$\Write^{a_j}_1(x_1)<_{S_{j}^{\leq l+2}} \Read^{b_{l+1}}_1(x_1)$.
A similar analysis to the base case shows that if $a_j$ is not orthogonal to $b_{l+1}$ then
$\Write^{a_j}_1(x_1)<_{S_{j}^{\leq l+2}} \Read^{b_{l+2}}_1(x_1)$.
On the other hand, if $a_j$ is orthogonal to $b_{l+1}$, it can be easily seen that the cross edges in $(X_{j}^{\leq l+2}, S_{j}^{\leq l+2})$ are precisely the cross edges in $(X_{j}^{\leq l+1}, S_{j}^{\leq l+1})$,
and thus there are no cross edges in $(X_{j}^{l_1+1}, S_{j}^{l_1+1})$, where $l_1=l+1$.
Finally, if $a_j$ is orthogonal to $b_{l'}$ for some $l'\in[l]$, the statement follows easily by the induction hypothesis.

The desired result follows.
\end{proof}

The next lemma reasons at the vector $A$ level.
The lemma states that if we order $\Write^{a_1}_1(x_1)< \Read^{b_1}_1(x_1)$,
if $a_{j'}$ is not orthogonal to $b_{l'}$ for any $j'\leq n/2-1$ and $l'\leq n/2$,
the closure rules eventually lead to $\Write^{a_{n/2}}_1(x_1)< \Read^{b_{n/2}}_1(x_1)$.
On the other hand, if there exists a smallest $j_1$ such that  $a_{j_1}$ is orthogonal to $b_{l'}$, for some $l'\in [n/2]$,
the closure rules will stop inserting cross edges between events of all vectors $a_{j'}$ and all vectors $b_{l}$, for $j'\geq j_1+1$ and $l\in [n/2]$.

\begin{restatable}{lemma}{lemovproof3}\label{lem:ovproof3}
For any $j\in [n/2-1]$, consider the rf-poset $\OPosetQ_{\leq j+1}^{\leq n/2}=(X_{\leq j+1}^{\leq n/2}, Q_{\leq j+1}^{\leq n/2}, \Observation_{\leq j+1}^{\leq n/2})$,
where $Q_{\leq j+1}^{\leq n/2}$ has a single cross edge $\Write^{a_1}_1(x_1)<_{Q_{\leq j+1}^{\leq n/2}} \Read^{b_1}_1(x_1)$.
Let  $\OPosetS_{\leq j+1}^{\leq n/2}=(X_{\leq j+1}^{\leq n/2}, S_{\leq j+1}^{\leq n/2}, \Observation_{\leq j+1}^{\leq n/2})$ be the closure of $\OPosetQ_{\leq j+1}^{\leq n/2}$.
The following assertions hold.
\begin{compactenum}
\item\label{item:ovproof3A} If $a_{j'}$ is not orthogonal to $b_{l'}$ for any $j'\in[j]$ and  $l'\in [n/2]$, then $\Write^{a_{j+1}}_1(x_1)<_{S_{\leq {j+1}}^{\leq n/2}} \Write^{b_{1}}_1(x_1)$.
\item\label{item:ovproof3B} If $a_{j'}$ is orthogonal to $b_{l'}$ for some $j'\in[j]$ and $l'\in[n/2]$, then there are no cross edges in $(X_{j_1+1\leq j'\leq j+1}^{\leq n/2}, S_{j_1+1\leq j'\leq j+1}^{\leq n/2})$, where $j_1$  is the smallest $j'$ such that $a_{j'}$ is orthogonal to $b_{l'}$, for some $l'\in [n/2]$.
\end{compactenum}
\end{restatable}
\begin{proof}[Proof (sketch)]
The proof is by induction on $j$.
For the base case $(j=1)$, consider the rf-poset $\OPosetQ_{1}^{\leq n/2}=(X_{1}^{\leq n/2},Q_{1}^{\leq n/2},\Observation_{1}^{\leq n/2})$, and by \cref{lem:ovproof2}, we have that 
if $a_1$ is not orthogonal to $b_{l'}$ for any $l'\in [n/2-1]$, then $\Write^{a_j}_1(x_1)<_{S_{j}^{\leq n/2}} \Read^{b_{n/2}}_1(x_1)$.
Consider the rf-poset $\OPosetQ_{1}^{n/2}$, and by \cref{lem:ovproof1}, 
we have $\Read^{a_1}_1(x_2)<_{S_{1}^{n/2}} \Write^{b_{n/2}}_1(x_2)$ iff $a_1$ and $b_{n/2}$ are not orthogonal.
Observe that if 
$\Read^{a_1}_1(x_2)<_{S_{1}^{n/2}} \Write^{b_{n/2}}_1(x_2)$, we also have 
$\Read^{a_1}_1(x_2)<_{S_{1}^{\leq n/2}} \Write^{b_{n/2}}_1(x_2)$, and thus
$\Read^{a_1}_1(x_2)<_{S_{\leq 2}^{\leq n/2}} \Write^{b_{n/2}}_1(x_2)$, hence, transitively,
$\Write^{a_1}(x_5)<_{S_{\leq 2}^{\leq n/2}}\Write^{b_{n/2}}(x_5)$.
Hence, by closure, we have 
$\Read^{a_1}(x_5)<_{S_{\leq 2}^{\leq n/2}}\Write^{b_{n/2}}(x_5)$ and thus transitively
$\Write^{a_1}(x_7)<_{S_{\leq 2}^{\leq n/2}}\Write^{b_{n/2}}(x_7)$.
Then, by closure we have
$\Read^{a_2}(x_7)<_{S_{\leq 2}^{\leq n/2}}\Write^{b_{n/2}}(x_7)$ and thus transitively
$\Write^{a_2}_1(x_1)<_{S_{\leq 2}^{\leq n/2}} \Read^{b_{1}}_1(x_1)$.
On the other hand, it can be easily seen that if 
$\Read^{a_1}_1(x_2)\not <_{S_{1}^{n/2}} \Read^{b_{n/2}}_1(x_2)$
then the cross edges in $(X_{\leq 2}^{\leq n/2}, S_{\leq 2}^{\leq n/2})$ are precisely the cross edges in $(X_{1}^{\leq n/2}, S_{1}^{\leq n/2})$,
and thus there are no cross edges in $(X_{2}^{\leq n/2}, S_{2}^{\leq n/2})$.

Now assume that the claim holds for $j$, and we argue that it holds for $j+1$.
By the induction hypothesis, we have that if $a_{j'}$ is not orthogonal to $b_{l'}$ for any $j'\in[j]$ and $l'\in [n/2]$, then 
$\Read^{a_{j+1}}_1(x_2)<_{S_{\leq {j+1}}^{\leq n/2}} \Write^{b_{1}}_1(x_2)$ and thus
$\Read^{a_{j+1}}_1(x_2)<_{S_{\leq {j+2}}^{\leq n/2}} \Write^{b_{1}}_1(x_2)$.
A similar analysis to the base case shows that if $a_{j+1}$ is not orthogonal to $b_{l'}$ for any $l'\in [n/2]$, then 
$\Read^{a_{j+2}}_1(x_2)<_{S_{\leq {j+2}}^{\leq n/2}} \Write^{b_{1}}_1(x_2)$.
On the other hand, if $a_{j+1}$  is orthogonal to some $b_{l'}$, it can be easily seen that the cross edges in $(X_{\leq j+2}^{\leq n/2}, S_{\leq j+2}^{n/2})$ are precisely the cross edges in $(X_{\leq j+1}^{\leq n/2}, S_{\leq j+1}^{n/2})$.
Thus, there are no cross edges in $(X_{j_1+1}^{\leq n/2}, S_{j_1+1}^{\leq n/2})$, where $j_1=j'+1$.
Finally, if $a_{j'}$ is orthogonal to $b_{l'}$ for some $j'\in [j]$ and $l'\in [n/2]$, the statement follows easily.

The desired result follows.
\end{proof}

We next have three lemmas that each is symmetric to \cref{lem:ovproof1}, \cref{item:ovproof2B} of \cref{lem:ovproof2}, and \cref{item:ovproof3B} of \cref{lem:ovproof3}, respectively.
The proof of each lemma is analogous to its symmetric lemma, and is omitted here.

\smallskip
\begin{restatable}{lemma}{lemovproof1b}\label{lem:ovproof1b}
For any $j,l\in [n/2]$, consider the rf-poset $\ov{\OPosetQ}_{j}^{l}=(X_{j}^{l}, \ov{Q}_{j}^{l}, \Observation_{j}^{l})$,
where $\ov{Q}_{j}^{l}$ has a single cross edge $\Write^{b_l}_1(x_2)<_{\ov{Q}_{i}^{j}} \Read^{a_j}_1(x_2)$.
Let  $\ov{\OPosetS}_{j}^{l}=(X_{j}^{l}, \ov{S}_{j}^{l}, \Observation_{j}^{l})$ be the closure of $\ov{\OPosetQ}_{j}^{l}$.
We have $\Read^{b_l}_1(x_1)<_{\ov{S}_{j}^{l}} \Write^{a_j}_1(x_1)$ iff $a_j$ and $b_l$ are not orthogonal.
\end{restatable}

\begin{restatable}{lemma}{lemovproof2b}\label{lem:ovproof2b}
For any $j\in[n/2]$ and $l\geq 2$, consider the rf-poset $\ov{\OPosetQ}_{j}^{\geq l-1}=(X_{j}^{\geq l-1}, \ov{Q}_{j}^{\geq l-1}, \Observation_{j}^{\geq l-1})$,
where $\ov{Q}_{j}^{\geq l-1}$ has a single cross edge $\Write^{b_{n/2}}_l(x_2)<_{\ov{Q}_{j}^{\geq l-1}} \Read^{a_j}_1(x_2)$.
Let  $\ov{\OPosetS}_{j}^{\geq l-1}=(X_{j}^{\geq l-1}, \ov{S}_{j}^{\geq l-1}, \Observation_{j}^{\geq l-1})$ be the closure of $\ov{\OPosetQ}_{j}^{\geq l-1}$.
If $a_j$ is orthogonal to $b_{l'}$, for some $l'\in [l]$, then there are no cross edges in $(X_{j}^{l_2-1\geq l'\geq l-1}, S_{j}^{l_2-1\geq l'\geq l-1})$,
where $l_2$ is the largest $l'$ such that $a_j$ is orthogonal to $b_{l_2}$.
\end{restatable}

\smallskip
\begin{restatable}{lemma}{lemovproof3b}\label{lem:ovproof3b}
For any $j\geq 2$ with $j>1$, consider the rf-poset $\ov{\OPosetQ}_{\geq j-1}^{\geq 1}=(X_{\geq j-1}^{\geq 1}, \ov{Q}_{\geq j-1}^{\geq 1}, \Observation_{\geq j-1}^{\geq 1})$,
where $\ov{Q}_{\geq j-1}^{\geq 1}$ has a single cross edge $\Write^{b_{n/2}}_1(x_2)<_{\ov{Q}_{\geq j-1}^{\geq 1}} \Read^{a_{n/2}}_1(x_2)$.
Let  $\ov{\OPosetS}_{\geq j-1}^{\geq 1}=(X_{\geq j-1}^{\geq 1}, \ov{S}_{\geq j-1}^{\geq 1}, \Observation_{\geq j-1}^{\geq 1})$ be the closure of $\ov{\OPosetQ}_{\geq j-1}^{\geq 1}$.
If $a_{j'}$ is orthogonal to $b_{l'}$ for some $j'\in[j]$ and $l'\in[n/2]$, then there are no cross edges in $(X_{j_2-1\geq j'\geq j-1}^{\geq 1}, \ov{S}_{j_2-1\geq j'\geq j-1}^{\geq 1})$, where $j_2$  is the largest $j'$ such that $a_{j'}$ is orthogonal to $b_{l'}$, for some $l'\in [n/2]$.
\end{restatable}

Using \cref{lem:ovproof1}, \cref{lem:ovproof2}, \cref{lem:ovproof3}  we can now prove \cref{lem:closure_ov}.

\smallskip
\lemclosureov*
\begin{proof}
We show that $\OPoset$ is realizable iff there exist $j,l\in [n/2]$ such that $a_j$ is orthogonal to $b_l$.
We prove each direction separately.

\noindent{$(\Rightarrow)$.}
Since $\OPoset$ is realizable, by \cref{rem:closure}, we have that the closure of $\OPoset$ exists.
Let $\OPosetK=(X, K, \Observation)$ be the closure of $\OPoset$.
Let $\OPosetQ=(X,Q,\Observation)$ where $Q$ has a single cross edge 
$\Write^{a_1}_1(x_1)<_{Q} \Read^{b_1}_1(x_1)$,
and $\OPosetS=(X,S,\Observation)$ be the closure of $\OPosetQ$.
Note that $P\Refines Q$ and thus $K\Refines S$, and we have that 
$\Write^{a_{n/2}}_1(x_2)\not<_{S} \Read^{b_{n/2}}_1(x_2)$.
If $\Write^{a_{n/2}}_1(x_1)<_{S} \Write^{b_{n/2}}_1(x_1)$, by \cref{lem:ovproof1}, we have that $a_{n/2}$ is orthogonal to $b_{n/2}$.
Otherwise, if $\Write^{a_{n/2}}_1(x_1)<_{S} \Read^{b_{1}}_1(x_1)$,
by \cref{item:ovproof2A} of \cref{lem:ovproof2}, we have that $a_{n/2}$ is orthogonal to $b_{l}$, for some $l\in [n/2]$.
Finally, if $\Write^{a_{n/2}}_1(x_1)\not <_{S} \Read^{b_{1}}_1(x_1)$, by \cref{item:ovproof3A} of \cref{lem:ovproof3}, we have that
$a_j$ is orthogonal to $b_l$, for some $j\leq n/2-1$ and $l\in [n/2]$.

\noindent{$(\Leftarrow)$.}
Let $\OPosetQ=(X,Q,\Observation)$ where $Q$ has a single cross edge 
$\Write^{a_1}_1(x_1)<_{Q_{\leq j+1}^{\leq n/2}} \Read^{b_1}_1(x_1)$, and
$\ov{\OPosetQ}=(X,\ov{Q},\Observation)$ where $\ov{Q}$ has a single cross edge
$\Write^{b_{n/2}}_1(x_2)<_{\ov{Q}}\Read^{a_{n/2}}_1(x_2)$.
Let $\OPosetS=(X, S, \Observation)$ and $\ov{\OPosetS}=(X,\ov{S}, \Observation)$ be the closures of $\OPosetQ$ and $\ov{\OPosetQ}$, respectively.
We argue that $(X, S\cup \ov{S})$ is a poset.
Note that this implies the lemma, as $(X, S\cup \ov{S}, \Observation)$ is the closure of $\OPoset$.

Let $(j_1, l_1)$ (resp., $(j_2, l_2)$) be the lexicographically smallest (resp., largest) pair of integers in $[n/2]\times [n/2]$ such that $a_{j_1}$ is orthogonal to $b_{l_1}$ (resp., $a_{j_2}$ is orthogonal to $b_{l_2}$).
We assume wlog that $j_1=j_2=j$ and $l_1=l_2=l$, as in any other case, the relation $S\cup \ov{S}$ is strictly smaller.
First, assume that $(X_{j}^{l}, S_{j}^{l}\cup \ov{S}_{j}^{l})$ is a poset.
By \cref{item:ovproof2B} of \cref{lem:ovproof2}, we have that there are no cross edges in 
$\OPosetS_{l}^{l+1\leq l'\leq n/2}$, 
while by \cref{lem:ovproof2b}, we have that there are no cross edges in
$\ov{\OPosetS}_{j}^{l-1\geq l'\geq 1}$.
Similarly, by \cref{item:ovproof3B} of \cref{lem:ovproof3}, we have that there are no cross edges in 
$\OPosetS_{j+1\leq j'}^{\leq n/2}$, 
while by \cref{lem:ovproof3b}, we have that there are no cross edges in
$\ov{\OPosetS}_{j-1\geq j'}^{\geq 1}$.
It follows that $(X, S\cup \ov{S})$ is a poset.

It remains to argue that $(X_{j}^{l}, S_{j}^{l}\cup \ov{S}_{j}^{l})$ is a poset.
The relation $S_{j}^{l}\cup \ov{S}_{j}^{l}$ is the transitive closure of a relation
that consists of the total orders $\Trace_A\Project X_{j}^{l}$ and $\Trace_B\Project X_{j}^{l}$,
together with the following relations, for each $i\in [D]$.
\begin{align*}
\Write^{a_j}_i(x_1)<_{S_{j}^{l}} \Write^{b_l}_i(x_1)
\quad \text{and}\quad
&\Write^{a_j}_i(x_3)<_{S_{j}^{l}} \Write^{b_l}_i(x_3)
\quad\text{and}\quad\\
\Write^{b_l}_i(x_2)<_{\ov{S}_{j}^{l}} \Write^{a_j}_i(x_2)
\quad \text{and}\quad
&\Write^{b_l}_i(x_6)<_{\ov{S}_{j}^{l}} \Write^{a_j}_i(x_6)
\end{align*}
It is straightforward to verify that if $S_{j}^{l}\cup \ov{S}_{j}^{l}$ has a cycle then, for some $i\in[D]$, we have
\begin{align*}
\Write^{a_j}_i(x_2) <_{\SeqTrace_A} \Write^{a_j}_i(x_1)
\quad\text{and}\quad
\Write^{b_l}_i(x_1) <_{\SeqTrace_B} \Write^{b_l}_i(x_2)\ .
\end{align*}
Thus by construction, $a_j[i]=1$ and $b_l[i]=1$, which contradicts the fact that $a_j$ is orthogonal to $b_l$.

The desired result follows.
\end{proof}

Finally, we are ready to prove \cref{them:tree_topologies_lb}.

\smallskip
\themtreetopologieslb*
\begin{proof}
We complete the proof of \cref{them:tree_topologies_lb} by showing that, in the above construction, $\OPoset$ has a closure iff $(\Write(z), \Read(z))$ is a data race of $\Trace$, witnessed by a correct reordering $\Trace^*$.
Observe that $\OPosetQ$ has $2$ threads, and is thus tree-inducible.
Due to \cref{lem:tree_topologies}, it suffices to show that $\OPoset$ is realizable iff $(\Write(z), \Read(z))$ is a data race of $\Trace$.

\noindent{\em $(\Rightarrow)$.}
Assume that $\OPoset$ is realizable and $\Trace'$ is a witness trace.
We construct the correct reordering $\Trace^*$ as follows.
We have $\Events{\Trace^*}=\Events{\Trace}\setminus \{ \Release_2(\ell), \Write(z), \Read(z) \}$.
We make $\Trace^*$ identical to $\Trace'$, i.e., $\Trace^*\Project X = \Trace'$.
For the events $\Events{\Trace^*}\setminus X$, we make
(i) $\Write(y)$ appear right after $\Write^{b_{n/2}}_1 (x_2)$,
(ii) $\Read(y)$ appear right before $\Read_1^{a_{n/2}}(x_2)$,
(iii)~$\Acquire_A(\ell), \Release_A(\ell)$ appear right after $\Write^{a_1}(x_1)$, and
(iv)~$\Acquire_B(\ell)$ appear last in $\Trace^*$.
Observe that $\Trace^*$ is a correct reordering in which $\Write(z)$ and $\Read(z)$ are enabled, hence $\Trace^*$ is a witness of the data race $(\Write(z), \Read(z))$.

\noindent{\em $(\Leftarrow)$.}
Assume that $(\Write(z), \Read(z))$ is a predictable data race of $\Trace$,
and let $\Trace^*$ be a correct reordering of $\Trace$ witnessing the data race.
We construct a trace $\Trace'$ as $\Trace'=\Trace^*\Project X$, and argue that $\Trace'$ realizes $\OPoset$.
Clearly $\Observation_{\Trace'}=\Observation$.
To see that $\Trace'$ is a linearization of $(X,P)$,
it suffices to argue that 
\begin{align*}
\Write^{b_{n/2}}_1 (x_2) <_{\Trace'} \Read^{a_{n/2}}_1(x_2)
\quad\text{and}\quad
\Write^{a_1}_1(x_1) <_{\Trace'} \Read^{b_1}_1(x_1)
\end{align*}
The first ordering follows from the fact that
\begin{align*}
\Write^{b_{n/2}}_1 (x_2) <_{\Trace^*} \Write(y) <_{\Trace^*} \Read_y <_{\Trace^*} \Read^{a_{n/2}}_1(x_2)\ .
\end{align*}
For the second ordering, observe that
$\Events{\Trace^*}=\Events{\Trace}\setminus \{ \Release_2(\ell), \Write(z), \Read(z) \}$.
Hence, 
\begin{align*}
\Write^{a_1}(x_1) <_{\Trace^*} \Release_1(\ell)<_{\Trace^*}\Acquire_2(\ell) <_{\Trace^*} \Read_a^{b_1}(x_1)\ .
\end{align*}

The desired result follows.
\end{proof}
\section{Details of {\cref{sec:small_distance}}}\label{sec:proofs_small_distance}

In this section we present the full proof of \cref{lem:reversal_realizability}.

\smallskip
\lemreversalrealizability*
\begin{proof}
Let $\OPoset=(X,P,\Observation)$ be the canonical rf-poset of $X$,
and the task is to decide the realizability of $\OPoset$ with $\ell$ reversals.
We describe a recursive algorithm for solving the problem for some rf-poset $\OPosetQ=(X,Q,\Observation)$  with $\ell'$ reversals, for some $\ell'\leq \ell$, where initially $Q=P$ and $\ell'=\ell$.
We first give a high-level description and argue about its correctness.
Afterwards, we describe some low-level details that allow us  to reason about the complexity.

\noindent{\em Algorithm and correctness.}
Consider the set
\begin{align*}
C=&\{ (\Write_1, \Write_2)\in \WritesAcquires{X}\times \WritesAcquires{X}\colon \Confl{\Write_1}{\Write_2}
\quad\text{and}\\
&\Unordered{\Write_1}{Q}{\Write_2} \text{ and } \Write_1<_{\Trace}\Write_2 \}\ .
\end{align*}
We construct a graph $G_1=(X,E_1)$, where $E_1=(\TOO\Project X) \cup C$.
Note that $G_1$ is write-ordered. 
If it is acyclic, we construct the read extension $G_2$ of $G_1$.
Observe that if $G_2$ is acyclic then any linearization $\Trace^*$ of $G$ realizes $\OPosetQ$, hence we are done.
Now consider that either $G_1$ or $G_2$ is not acyclic,
and let $G=G_1$ if $G_1$ is not acyclic, otherwise $G=G_2$.
Given a cycle $\Cycle$ of $G$, represented as a collection of edges, 
define the set of \emph{cross-edges} of $\Cycle$ as $\Cycle \setminus Q$.
Observe that since there are $k$ threads, $G$ has a cycle with $\leq k$ cross edges.
Indeed, if there are more than $k$ cross edges in $\Cycle$, then there are two cross edges that go out of two nodes that belong to the same thread and hence are connected by a path in $G$.
Then, we can simply remove one of these edges and $G$ will still have a cycle.
Hence, we can repeat this process until we end up with a cycle that has $\leq k$ cross edges.
In addition, any trace $\Trace^*$ that realizes $\OPosetQ$ must linearize 
an rf-poset $(X,Q_a, \Observation)$
 where $a=(\Event_1, \Event_2)$ ranges over the cross-edges of $\Cycle.$
In particular, we take $Q_a=Q\cup \{ b \}$, where
\begin{align*}
b=
\begin{cases}
(\Event_2, \Event_1), & \text{ if } a\in \WritesAcquires{X}\times \WritesAcquires{X}\\
(\Observation(\Event_2), \Event_1), & \text{ if } a \in \WritesAcquires{X} \times \ReadsReleases{X}\\
(\Event_2, \Observation(\Event_1)), & \text{ if } a \in \ReadsReleases{X} \times \WritesAcquires{X}\ .
\end{cases}
\end{align*}

Observe that any such choice of $b$ reverses the order of two conflicting write events or lock-acquire events in $\Trace$.
Since there are $\leq k$ cross edges in $\Cycle$, there are $\leq k$ such choices for $Q_a$.
Repeating the same process recursively for the rf-poset $(X, Q_a, \Observation)$ for $\ell'-1$ levels solves the $\ell'$-distance-bounded realizability problem for $\OPosetQ$.
Since initially $\ell'=\ell$ and $Q=P$, this process solves the same problem for $\OPoset$ and thus for $X$.

\noindent{\em Complexity.}
The recursion tree above has branching $\leq k$ and depth $\leq \ell$, hence there will be at most $k^{\ell}$ recursive instances.
We now provide some lower-level algorithmic details which show that each instance can be solved in $O(k^{O(1)}\cdot n)$ time.
The main idea is that each of the graphs $G_1$ and $G_2$ have a sparse transitive reduction~\cite{Aho72} of size $O(k\cdot n)$, and thus we can represent $G_1$ and $G_2$ with few edges.

For the graph $G_1$, we construct a sparse graph $G'_1$ that preserves the reachability relationships of $G_1$ as follows.
We traverse the trace $\Trace$ top-down. 
For the current event $\Write$ such that $\Write\in \WritesAcquires{X}$,
for every $i\in [k]$, let $\Write_1$ be the last event of thread $i$ which precedes $\Write$ in $\Trace$, and such that $\Write_1\in \WritesAcquires{X}$ and $\Confl{\Write}{\Write_1}$.
Let $\Write'_1$ be the last event of $\Trace$ such that $\Write'_1\leq_{\TO} \Write_1$
and $\Write\not <_{Q} \Write'_1$ (note that possibly $\Write'_1=\Write_1$).
If such $\Write'_1$ exists, since $X$ is an ideal, we have $\Write'_1\in X$.
We introduce the edge $\Write'_1\to \Write$ in $G'_1$.

Similarly, for the graph $G_2$, we construct a sparse graph $G'_2$ that preserves the reachability relationships of $G_2$ as follows.
We iterate over all read and lock-release events $\Read\in X$.
For each such $\Read$ and thread $i\in [k]$, we insert two edges $(\Write_2, \Read)$ and $(\Read, \Write_3)$ in $G'_2$, where $\Write_2$ (resp., $\Write_3$) is the latest predecessor (resp., earliest successor) of $\Observation(\Read)$ in thread $i$ that conflicts with $\Read$.

We now argue that all events $\Write_1$, $\Write'_1$, $\Write_2$ and $\Write_3$ above can be identified in $O(n)$ total time.
For every $i\in [k]$ and memory location $x$, we construct a total order $\SeqTrace_i^{x}$ of all write events or lock-acquire events on location $x$ of thread $\Process_i$.
Clearly all such total orders can be constructed in $O(n)$ time.
\begin{compactenum}
\item Events $\Write_1$:~as we traverse $\Trace$ top-down, we simply remember for each memory location $x$, the last write event or lock-acquire event on location $x$ for each thread.
Hence, the total time for identifying all such events is $O(n)$.
\item Events $\Write'_1$:~Given the event $\Write_1$ on memory location $x$, we simply traverse the total order $\SeqTrace_i^{x}$ from $\Write'_1$ backwards until we find either
(i)~an event $\Write'_1$ with the desired properties, or
(ii)~an event $\Write''_1$ that has been examined before when inserting edges to thread $\Process_i$.
In the first case, we simply add the edge $(\Write'_1\to \Write)$ as described above.
In the second case we do nothing, as the desired ordering is already present due to transitivity through a write event $\Write'$ of thread $\Process_i$ that has been examined before (thus $\Write'<_{\TO}\Write$).
Hence, every event in every total order $\SeqTrace_{i}^{x}$ is examined $O(1)$ times, thus the total time for identifying all $\Write'_1$ events is $O(n)$.
\item Events $\Write_2$ and $\Write_3$:~after we have constructed $G'_1$, each pair of such events can be retrieved from $G'_1$ in $O(1)$ time. Hence the total time for identifying all such events is $O(n)$.
\end{compactenum}

Finally, the above process creates graphs $G'_1$ and $G'_2$ that have $O(k^{O(1)}\cdot n)$ edges.
Detecting a cycle $\Cycle'$ in either $G'_1$ or $G'_2$ is done by a simple DFS, which takes linear time in the number of edges.
Converting $\Cycle'$ to a cycle $\Cycle$ such that $\Cycle$ has at most $k$ cross edges can be done in $O(|\Cycle'|)=O(k^{O(1)}\cdot n)$ time, by removing multiple edges whose endpoints are $\TO$-ordered.

The desired result follows.
\end{proof}

\end{document}